%% file: main.tex
\documentclass[11pt]{article}

\usepackage{amsmath}
\usepackage{amsthm}
\usepackage{amssymb}
\usepackage{amsfonts}
\usepackage{forest}
\usepackage{mathrsfs}
\usepackage{bbm}
\usepackage{bbold}
\usepackage{setspace}
\usepackage{thmtools}
\usepackage{thm-restate}
\usepackage{fullpage}
\usepackage{tcolorbox}
\usepackage[all]{xy}
\usepackage[title, titletoc]{appendix}
\usepackage[ruled,linesnumbered]{algorithm2e}
\usepackage{mathalfa}
\usepackage{amsbsy}
\usepackage{comment}
\usepackage{xr}
\usepackage{xifthen}
\usepackage{hyperref}
\usepackage{enumitem}
\usepackage{mathtools}
\usepackage{mathalpha}
\usepackage{tikz-cd}
\usepackage{bm}
\usepackage{empheq}
\usepackage{hypbmsec}

\newtheorem{theorem}{Theorem}[section]
\newtheorem{definition}[theorem]{Definition}
\newtheorem{lemma}[theorem]{Lemma}
\newtheorem{claim}[theorem]{Claim}

\newtheorem{corollary}[theorem]{Corollary}

\newtheorem{notation}[theorem]{Notation}
\newtheorem{example}[theorem]{Example}
\newtheorem{remark}[theorem]{Remark}
\newtheorem{note}[theorem]{Note}
\newtheorem*{notation*}{Notation}
\newtheorem*{lemma*}{Lemma}
\newtheorem*{proposition*}{Proposition}
\newtheorem*{note*}{Note}
\newtheorem*{theorem*}{Theorem}
\newtheorem*{corollary*}{Corollary}

\newcommand{\innerifequals}[3]{\ifthenelse{\equal{#1}{#2}}{#3}{}}
\newenvironment{switch}[1]{\newcommand{\case}{\innerifequals{#1}}}{}
\newcommand{\definedas}[0]{\coloneqq}

\newcommand{\parens}[1]{\left( #1 \right)}
\newcommand{\sparens}[1]{\left[ #1 \right]}

\newcommand{\expectation}[2]{\mathbbm{E}_{#1} \sparens{#2}}
\newcommand{\set}[1]{\left\{ #1 \right\}}
\newcommand{\suchthat}[0]{\middle|}

\newcommand{\basefield}[0]{\mathbbm{F}}
\newcommand{\blocklength}[0]{n}
\newcommand{\field}[0]{\basefield^{\blocklength}}
\newcommand{\titlefield}[0]{\texorpdfstring{$\field$}{Fn}}
\newcommand{\degree}[0]{d}
\newcommand{\basefieldsize}[0]{p}

\newcommand{\series}[3]{#1_{#2},\ldots,#1_{#3}}

\newcommand{\charfunc}[1]{e \sparens{#1}}

\newcommand{\genprime}[0]{p}
\NewDocumentCommand{\genset}{o}{\IfNoValueTF{#1}{A}{
    \begin{switch}{#1}
        \case{1}{A}
        \case{2}{B}
        \case{3}{C}
    \end{switch}
}}
\NewDocumentCommand{\onvarfunc}{o}{\IfNoValueTF{#1}{f}{
    \begin{switch}{#1}
        \case{1}{f}
        \case{2}{g}
        \case{3}{h}
    \end{switch}
}}
\NewDocumentCommand{\onvarfuncset}{o}{\IfNoValueTF{#1}{\mathfrak{f}}{
    \begin{switch}{#1}
        \case{1}{\mathfrak{f}}
        \case{2}{\mathfrak{g}}
        \case{3}{\mathfrak{h}}
    \end{switch}
}}

\NewDocumentCommand{\genfunc}{o}{\IfNoValueTF{#1}{F}{
    \begin{switch}{#1}
        \case{1}{F}
        \case{2}{G}
        \case{3}{H}
    \end{switch}
}}

\NewDocumentCommand{\genfuncset}{o}{\IfNoValueTF{#1}{\mathfrak{F}}{
    \begin{switch}{#1}
        \case{1}{\mathfrak{F}}
        \case{2}{\mathfrak{G}}
        \case{3}{\mathfrak{H}}
    \end{switch}
}}

\newcommand{\funcdef}[3]{#1:#2\rightarrow#3}
\newcommand{\restrictfunc}[2]{#1|_{#2}}

\NewDocumentCommand{\genpoly}{o}{\IfNoValueTF{#1}{P}{
    \begin{switch}{#1}
        \case{1}{P}
        \case{2}{Q}
        \case{3}{H}
    \end{switch}
}}
\NewDocumentCommand{\genpolyset}{o}{\IfNoValueTF{#1}{\mathcal{P}}{
    \begin{switch}{#1}
        \case{1}{\mathcal{P}}
        \case{2}{\mathcal{Q}}
        \case{3}{\mathcal{H}}
    \end{switch}
}}

\newcommand{\factor}[0]{\mathcal{B}}
\newcommand{\semrefine}[0]{\succeq_{\text{sem}}}
\newcommand{\semrefineex}[1]{\succeq_{\text{sem}|#1}}
\newcommand{\relsemrefine}[1]{\succeq^{#1}_{\text{sem}}}
\newcommand{\relsemrefineex}[1]{\succeq^{#1}_{\text{sem}}}
\newcommand{\synrefine}[0]{\succeq_{\text{syn}}}

\newcommand{\allpolyset}[3]{Poly_{#1}(#2\rightarrow#3)}

\newcommand{\cube}[2]{(#1|#2)}

\newcommand{\cubes}[2]{C_{#1}(#2)}

\NewDocumentCommand{\lowdegpoly}{o}{\IfNoValueTF{#1}{\alpha}{
    \begin{switch}{#1}
        \case{1}{\alpha}
        \case{2}{\beta}
    \end{switch}
}}

\newcommand{\naturalnumbersset}[0]{\mathbbm{N}}
\newcommand{\realnumbersset}[0]{\mathbbm{R}}

\NewDocumentCommand{\varpoly}{o}{\IfNoValueTF{#1}{L}{
    \begin{switch}{#1}
        \case{1}{L}
        \case{2}{M}
    \end{switch}
}}
\NewDocumentCommand{\varpolyset}{o}{\IfNoValueTF{#1}{\tilde{\mathcal{L}}}{
    \begin{switch}{#1}
        \case{1}{\matchcal{L}}
        \case{2}{\matchcal{M}}
    \end{switch}
}}

\newcommand{\relativeremainder}[1]{\overline{#1}}
\newcommand{\remainderpoly}[0]{R}

\NewDocumentCommand{\onvarpoly}{o}{\IfNoValueTF{#1}{p}{
    \begin{switch}{#1}
        \case{1}{p}
        \case{2}{q}
        \case{3}{h}
    \end{switch}
}}
\NewDocumentCommand{\onvarpolyset}{o}{\IfNoValueTF{#1}{\mathfrak{p}}{
    \begin{switch}{#1}
        \case{1}{\mathfrak{p}}
        \case{2}{\mathfrak{q}}
        \case{3}{\mathfrak{h}}
    \end{switch}
}}

\newcommand{\variety}{{\tilde{X}}}
\newcommand{\titlevariety}{\texorpdfstring{$\variety$}{X}}
\newcommand{\varietypolycount}{{\tilde{c}}}
\newcommand{\varietydeg}[0]{\tilde{\degree}}

\newcommand{\lift}[1]{\widehat{#1}}
\newcommand{\homopart}[1]{\mathrm{h}({#1})}

\newcommand{\directionderivative}[0]{D}

\newcommand{\zerofunc}[1]{Z \parens{#1}}
\newcommand{\existfunc}[1]{1_{#1}}

\newcommand{\rank}[1]{rank \parens{#1}}
\newcommand{\drank}[2]{rank_{#1} \parens{#2}}

\newcommand{\schmrank}[1]{schmrank \parens{#1}}

\newcommand{\relschmrank}[2]{shcmrank_{#1}\parens{{#2}}}
\newcommand{\relrank}[2]{rank_{#1} \parens{#2}}
\newcommand{\drelrank}[3]{rank_{{#1}, {#2}} \parens{#3}}

\newcommand{\rankval}[0]{{{r}}}
\newcommand{\rankfunc}[0]{{{{r}}}}
\newcommand{\rankbiasfunc}[0]{{\tilde{r}}}
\newcommand{\varietyrankval}[0]{{\bar{r}}}

\newcommand{\epsilonlimitedrankbias}[0]{{\tilde{\epsilon}}}

\newcommand{\bias}[1]{bias({#1})}
\newcommand{\relbias}[2]{bias_{#1}{({#2})}}

\newcommand{\norm}[1]{\left\lVert{#1}\right\rVert}
\newcommand{\abs}[1]{\left| #1 \right|}

\newcommand{\dist}[1]{dist \parens{#1}}

\newcommand{\prex}[2]{\Pr_{#1}\left[ #2 \right]}

\NewDocumentCommand{\linearform}{o}{\IfNoValueTF{#1}{l}{
    \begin{switch}{#1}
        \case{0}{l}
        \case{1}{r}
    \end{switch}
}}

\newcommand{\gencode}{\mathfrak{C}}

\newcommand{\quotientcode}{\gencode_{\variety}}

\newcommand{\reedmullercodeex}[3]{RM_{#1, #2}(#3)}
\newcommand{\listpolycount}[4]{\ell_{#1, #2}(#3, #4)}
\newcommand{\listdecodingradiusex}[3]{LDR_{#1, #2}(#3)}
\newcommand{\normalizedcodedistance}[2]{\delta_{#1}(#2)}
\newcommand{\normalizedcodedistanceex}[3]{\delta_{#1, #2}(#3)}

\DeclareRobustCommand{\topbot}{\genfrac{}{}{0pt}{}}

\title{List Decoding Quotient Reed-Muller Codes}
\author{
Omri Gotlib
\footnote{Department of Computer Science, Bar-Ilan University, Email: gotlib.omri@gmail.com.}
\and
Tali Kaufman
\footnote{Department of Computer Science, Bar-Ilan University, Email: kaufmant@mit.edu, supported by ISF.}
\and
Shachar Lovett
\footnote{Department of Computer Science and Engineering, UC San Diego, Email: shachar.lovett@gmail.com, supported by NSF award 2425349 and a Simons investigator award.}
}

\begin{document}
    \maketitle
    \begin{abstract}
        Reed-Muller codes consist of evaluations of
        $\blocklength$-variate polynomials over a finite field $\basefield$ with degree at most $\degree$.
        Much like every linear code, Reed-Muller codes can be characterized by constraints, where a codeword is valid if and only if it satisfies all \emph{degree-$\degree$} constraints.

        For a subset $\variety \subseteq \field$,
        we introduce the notion of \emph{$\variety$-quotient} Reed-Muller code.
        A function $\funcdef{\genfunc}{\variety}{\basefield}$ is a valid codeword in the quotient code if it satisfies all the constraints of degree-$\degree$ polynomials \emph{lying in $\variety$}.
        This gives rise to a novel phenomenon: a quotient codeword may have \emph{many} extensions to original codewords.
        This weakens the connection between original codewords and quotient codewords which introduces a richer range of behaviors along with substantial new challenges.

        Our goal is to answer the following question: what properties of $\variety$ will imply that the quotient code inherits its distance and list-decoding radius from the original code?
        \newline
        We address this question using techniques developed by Bhowmick and Lovett~\cite{bhowmick2014list},
        identifying key properties of $\field$ used in their proof and extending them to general subsets $\variety \subseteq \field$.
        By introducing a new tool, we overcome the novel challenge in analyzing the quotient code that arises from the weak connection between original and quotient codewords.
        This enables us to apply known results from additive combinatorics and algebraic geometry~\cite{kazhdan2018polynomial, kazhdan2019extendingweaklypolynomialfunctions, lampert2021relative}
        to show that when $\variety$ is a \emph{high rank variety}, $\variety$-quotient Reed-Muller codes inherit the distance and list-decoding parameters from the original Reed-Muller codes.
    \end{abstract}
    \input{new_introduction}
    \input{preliminaries}

    \input{polynomials_in_X}

    \input{relative_rank_relative_bias_property}

    \input{regularization_relative_to_X}
    \input{distance_of_reed_muller_over_X}
    \input{list_decoding_reed_muller_general_degrees}

    \begin{appendices}
        \input{equidistribution}
        \input{comparing_ranks}

    \end{appendices}

    \bibliography{main}
    \bibliographystyle{alpha}

\end{document}

%% file: new_introduction.tex
\section{Introduction}\label{sec:introduction}
Let $\basefield$ be a finite field, $\blocklength \in \naturalnumbersset$, and let $\variety \subseteq \field$ be a subset
\footnote{As a convention, we use $\tilde{\square}$ to denote properties of the subset, and thus also the subset itself.}
.
We begin by introducing a new definition applicable to any linear code over $\basefield$: the \emph{$\variety$-quotient code}.
We then illustrate this novel definition using Reed-Muller codes, and present a property of $\variety$ which we use to show that $\variety$-quotient Reed-Muller code \emph{inherits its distance and list decoding radius} from the original Reed-Muller code.
Finally, leveraging known results from additive combinatorics and algebraic geometry, we establish as a corollary that this inheritance holds when $\variety$ is a \emph{high-rank variety}.

\paragraph{The Quotient Code}
Let $\gencode$ be a linear code over $\basefield$.
Each codeword of $\gencode$ can be described as a function $\funcdef{\genfunc}{\field}{\basefield}$ that is in the span of the columns of the code's \emph{generator matrix}.
An equivalent way to describe $\gencode$ is using a \emph{parity check matrix}, where a function $\genfunc$ is a codeword if and only if it satisfies the constraints represented by parity-check matrix.
Each such constraint can be thought of as a requirement over a few inputs of $\genfunc$ from $\field$: the requirement that their weighted sum will equal $0$.

The first novel definition we introduce is the definition of the \emph{$\variety$-induced} code:
\begin{definition}[The $\variety$-Induced Code]
    We define the \emph{$\variety$-induced code $\quotientcode$} to be
    the set of all functions $\funcdef{\onvarfunc}{\variety}{\basefield}$
    \footnote{By convention, we use uppercase letters to denote functions with domain $\field$ and lowercase letters to denote functions with domain $\variety$.}
    that satisfy all the constraints \emph{that lie in $\variety$}.
\end{definition}

Let us briefly describe the connection between codewords in $\field$ and $\variety$-induced codewords.
One can easily verify that each original codeword \emph{restricted} to $\variety$ is a valid codeword in the induced code.
\newline
We call an extension of an $\variety$-induced codeword $\funcdef{\onvarfunc}{\variety}{\basefield}$ to valid codeword in the original code (extending its domain to $\field$), a \emph{lift} of $\onvarfunc$.
When each induced codeword has a unique lift, there is a natural 1-to-1 correspondence between the original and induced codeword.
This becomes substantially more interesting for subsets $\variety$ in which induced codewords have \emph{multiple} lifts.
This non-uniqueness weakens the connection between the original codewords and induced codewords, and leads to a richer range of phenomena (and interesting new challenges).

We also note that the other direction is not always true: For a general subset $\variety$, there might be an induced codeword (a valid codeword in the induced code) that \emph{cannot be lifted} to a valid codeword in $\field$.
We are interested to better understand $\quotientcode$ using $\gencode$ and vice-versa, and therefore we introduce a new notion, which is the notion of the \emph{$\variety$-quotient code}:
\begin{definition}[The $\variety$-Quotient Code]
    Let $\gencode$ be a linear code, and let $\quotientcode$ be the $\variety$-induced code of $\gencode$.
    We say $\quotientcode$ is a \emph{$\variety$-quotient code}
    if every quotient codeword $\onvarfunc \in \quotientcode$ has a lift to $\field$.
\end{definition}
In the case described above, we also say that $\variety$ is a \emph{lift-enabler} for $\gencode$ and that the code $\gencode$ is a \emph{covering code} for the code $\quotientcode$.
\newline
The novelty of this definition is that it captures subsets in which \emph{there is} a correspondence between codewords in $\variety$ and in $\field$,
and the correspondence may be \emph{1-to-many}.

\paragraph{Importance of Definition}
This timely definition extends a fundamental and useful concept previously introduced for graphs and complexes—namely, the notion of a \emph{covering graph} or alternatively, the \emph{quotient graph}.
This concept gained an increasing prominence in theoretical computer science, where it was recently employed to construct \emph{high dimensional expanders}~\cite{dikstein2022newhighdimensionalexpanders, yaacov2024sparsehighdimensionalexpanders}
and achieve improved \emph{local testing} results~\cite{gotlib2022listagreementexpansioncoboundary, dikstein2024agreementtheoremshighdimensional, bafna2024characterizingdirectproducttesting},
where the latter also played a crucial role in constructions of PCPs.
Consequently, the study of covering spaces for graphs has found usages in theoretical computer science and specifically in development of PCPs with enhanced properties.
We believe our question, which explores the analogous question for codes, will similarly lead to meaningful applications in theoretical computer science.

In addition to that, the question of \emph{puncturing} of codes has caught much attention recently, in a line of work~\cite{brakensiek2024genericreedsolomoncodesachieve, alrabiah2024randomlypuncturedreedsolomoncodes, brakensiek2024generalizedgmmdspolynomialcodes, brakensiek2024agcodesachievelistdecoding},
followed by the resolution of the GM-MDS conjecture~\cite{DBLP:journals/corr/abs-1803-02523, DBLP:journals/corr/abs-1803-03752}.
Where the question of puncturing is focused exclusively on the case where the lift is \emph{unique},
the study of quotient codes also tackles subsets $\variety \subseteq \field$ where the lift is \emph{not unique}.
Notably, in the unique-lift case there are well-established lower-bounds for the size of $\variety$ such as~\cite[Theorem 1.1]{DBLP:journals/cc/DoronTT22}.
In contrast, the size of $\variety$ in quotient codes may be \emph{much smaller} than its lower-bound in punctured code (for example in Reed-Muller codes), suggesting the potential for new insights and improved results.

\paragraph{Our Question}
Our goal is to answer the following question:
what properties of $\variety$ will imply that the quotient code inherits its distance and list-decoding radius from the original code?

This question is analogous to the study of quotients of expander graphs—just as not all quotients of an expander necessarily preserve expansion,
not all subsets $\variety$ necessarily yield a well-behaved quotient code.
Understanding the conditions under which expansion is preserved has been a fundamental problem in the study of expanders,
and similarly, identifying the conditions under which a quotient code retains key properties of the original code is a central challenge in our work.
Given this parallel, we believe our question may have broader implications for future research in both coding theory and theoretical computer science.

We answer this question in the context of \emph{Reed-Muller codes}.
Notably, our approach does \emph{not only} address the case of where there are multiple lifts,
but also introduces a novel framework for analyzing unique-lift (puncturing) setting when the field size is constant-a scenario that is typically considered more challenging.

\paragraph{Reed-Muller Codes}
Let $\basefield$ be a finite field, and let $\blocklength, \degree$ be integers.
Each codeword in Reed-Muller code $\reedmullercodeex{\basefield}{\field}{\degree}$,
is defined by a polynomial over $\basefield$ in $\blocklength$ variables with total degree $\leq \degree$
\footnote{We focus on the regime where $\degree, \abs{\basefield}$ are considered constants and $\blocklength$ is considered very large.}
.
The message that one wishes to encode is represented in the code as a polynomial $\funcdef{\genpoly}{\field}{\basefield}$, whose coefficients are the different message characters.
The encoding of the message is a vector of the different evaluation of $\genpoly$ over \emph{all} possible points in $\field$.

Alternatively, one can describe Reed-Muller codes using a set of local constraints.
A function $\funcdef{\genfunc}{\field}{\basefield}$ is a polynomial of degree $\leq \degree$
if and only if the (alternating) sum of each possible \emph{cube}, which is a set of points of the form $\set{x + \sum_{i \in S} y_i}_{S \subseteq \sparens{\degree + 1}}$ for $x, y_1,...,y_{\degree+1} \in \field$, equals $0$.
The set of all cubes is \emph{the set of constraints of degree-$\degree$ polynomials}.

Next, we present our notations for the induced Reed-Muller code:
\begin{notation}[The $\variety$-Induced Reed-Muller Code]
    We say a function $\funcdef{\genfunc}{\variety}{\basefield}$ is a \emph{polynomial of degree $\leq \degree$ \emph{in $\variety$}}
    if it satisfies all the constraints of degree-$\degree$ polynomials \emph{that lie in $\variety$}.
    \newline
    We denote the $\variety$-induced Reed-Muller code:
    \[
        \reedmullercodeex{\basefield}{\variety}{\degree} = \set{\funcdef{\onvarpoly}{\variety}{\basefield} \suchthat \onvarpoly \text{ is a polynomial of degree } \leq \degree \text{ in } \variety}
    \]
\end{notation}

\paragraph{Properties of Induced Reed-Muller Codes}
A study of Ziegler and Kazhdan~\cite{kazhdan2018polynomial, kazhdan2019extendingweaklypolynomialfunctions, kazhdan2020propertieshighranksubvarieties}
shows that if $\variety$ is a \emph{high rank variety}
\footnote{Under some conditions we describe later.}
, then $\variety$ is a \emph{lift-enabler} for $\reedmullercodeex{\basefield}{\field}{\degree}$.
In other words, the authors showed that the $\variety$-induced Reed-Muller code is in fact a \emph{$\variety$-quotient Reed-Muller code}.
We rely on this property of $\variety$ as a black-box.
See Section~\ref{sec:polynomials_in_X} for more details in this regard.

An additional property of $\variety \subseteq \field$ we rely on is the connection between \emph{algebraic structure} and \emph{random behavior (equidistribution)} of polynomials in $\variety$.
\newline
For $\field$, this connection is a well-studied result~\cite{green2007distribution, kaufman2008worst, DBLP:journals/corr/0001L15}.
It lies in the heart of many results in higher-order Fourier analysis,
and specifically was used in~\cite{bhowmick2014list} to analyze the list decoding radius of Reed-Muller code in $\field$.
\newline
The equivalent of this relation for subsets $\variety \subseteq \field$ was studied in~\cite{lampert2021relative, gowers2022equidistributionhighrankpolynomialsvariables}.
These works captured the measure of algebraic-structure in $\variety$ by a definition called \emph{relative rank},
and captured the lack of random behavior in $\variety$ by a definition called \emph{relative bias}.
We note that for subsets, the definition of algebraic structure of a polynomial in $\variety$ considers the algebraic structure of \emph{all its possible} lifts.
It was shown in~\cite{lampert2021relative} that when $\variety$ is a high-rank variety, high relative rank implies low relative bias
\footnote{
    Note that even though Gowers and Karam~\cite{gowers2022equidistributionhighrankpolynomialsvariables} also acheived a similar relation for a type of subsets,
    the definition of rank they used is slightly different than the standard definition of rank.
    While this difference may seem unharmful at first, it is, to our knowledge, does not allow to do a \emph{regularization} process
    (note that a generalization of this process is the heart of our proof).
}
.
\newline
We use this property as a black box as well.
When a subset $\variety \subseteq \field$ has such property for polynomials of degree $\leq \degree$, we say that it has the \emph{$\degree$-relative rank-bias property}.
See Section~\ref{sec:relative-rank-bias-property} for more details.

\paragraph{Our Results}
Next, let us present our main theorem more concretely.
Our work focuses on the regime where $\degree < \abs{\basefield}$ for prime finite fields $\basefield = \basefield_p$.
Throughout this paper, we always assume these two assumptions.
Denote the \emph{minimum normalized distance of $\reedmullercodeex{\basefield}{\field}{\degree}$} by $\normalizedcodedistanceex{\basefield}{\field}{\degree}$,
shorthand by $\normalizedcodedistance{\basefield}{\degree}$.
We have:
\[
    \normalizedcodedistance{\basefield}{\degree} = 1 - \degree/\abs{\basefield}
\]
Moreover, we define the \emph{list decoding count} of $\reedmullercodeex{\basefield}{\field}{\degree}$ by:
\[
    \listpolycount{\basefield}{\field}{\degree}{\tau} \definedas
    \max_{\funcdef{\genfunc}{\field}{\basefield}}
        {\abs{\set{\genpoly \in \allpolyset{\leq \degree}{\field}{\basefield} \suchthat {\dist{\genpoly, \genfunc} \leq \tau}}}}
\]
Let $\listdecodingradiusex{\basefield}{\field}{\degree}$ be the \emph{list decoding radius} of $\reedmullercodeex{\basefield}{\field}{\degree}$,
which is the maximum $\tau$ for which $\listpolycount{\basefield}{\field}{\degree}{\tau - \epsilon}$ is bounded by a \emph{constant} depending only on $\epsilon, \abs{\basefield}, \degree$.
\newline
In the paper~\cite{bhowmick2014list} it was shown that for constant field size and degree, the list decoding radius \emph{reaches the distance of the code}, as conjectured earlier by~\cite{10.1145/1374376.1374417}
\footnote{Note that it is known that $\listdecodingradiusex{\basefield}{\field}{\degree} \leq \normalizedcodedistance{\basefield}{\degree}$,
    and therefore, in a sense, their result is \emph{optimal in $\field$} assuming $\degree, \abs{\basefield}$ are considered as constants.}
.
We denote the corresponding distance parameter of $\variety \subseteq \field$ by $\normalizedcodedistanceex{\basefield}{\variety}{\degree}$ and $\listdecodingradiusex{\basefield}{\variety}{\degree}$ respectively.

We next present our main theorem, which establishes that the \emph{list decoding radius} of the quotient Reed-Muller code is \emph{at least as good} as the that of the original code:
\begin{theorem*}[List Decoding Quotient Reed-Muller Code]
\footnote{Informal, for formal see Theorem~\ref{thm:list-decoding-RM-in-X}.}
Let $\basefield$ be a finite (prime) field of constant size, let $\degree \in \naturalnumbersset$ be a constant such that $\degree < \abs{\basefield}$,
and let $\blocklength \in \naturalnumbersset$ be an integer.
\newline
Let $\variety \subseteq \field$ be a subset that is a lift-enabler for $\reedmullercodeex{\basefield}{\field}{\degree}$ and has the $\degree$-relative rank-bias property.
\newline
Then, $\reedmullercodeex{\basefield}{\variety}{\degree}$ inherits its \emph{list decoding radius} from $\reedmullercodeex{\basefield}{\field}{\degree}$, i.e:
\[
    \listdecodingradiusex{\basefield}{\variety}{\degree} \geq \listdecodingradiusex{\basefield}{\field}{\degree}
\]
\end{theorem*}

In addition, we also achieve a (simpler) result regarding the \emph{distance} of the quotient Reed-Muller code (Theorem~\ref{thm:distance-of-RM-in-X}):
Under the conditions described above,
$\reedmullercodeex{\basefield}{\variety}{\degree}$ also inherits its \emph{distance} from $\reedmullercodeex{\basefield}{\field}{\degree}$, i.e
$\normalizedcodedistanceex{\basefield}{\variety}{\degree} \geq \normalizedcodedistanceex{\basefield}{\field}{\degree}$
\footnote{Our techniques also show that also the other direction is true, which yields an \emph{equality} in the distance of the two codes.}
.

As a corollary, using results studied in~\cite{kazhdan2018polynomial, kazhdan2019extendingweaklypolynomialfunctions, lampert2021relative} regarding high-rank varieties, we obtain the following:
\begin{corollary*}[List Decoding Quotient Reed-Muller Code: High Rank Variety]
    Let $\variety \subseteq \field$ be a \emph{high rank variety},
    that is, $\variety$ is the set of common zeros of a collection of polynomials $\varpolyset = (\varpoly_1,...,\varpoly_{\varietypolycount})$
    that is of \emph{high rank}
    \footnote{We note that the higher the rank of the collection is, the more accurate the greater or equal in the theorem is.}
    \footnote{We also note that for this result some assumptions are needed regarding the field size or the degree of the polynomials in the collection.}
    , i.e. $\variety = \zerofunc{\varpolyset} = \set{x \suchthat \forall i: \varpoly_i(x) = 0}$.
    \newline
    Then, $\reedmullercodeex{\basefield}{\variety}{\degree}$ inherits its distance parameters from $\reedmullercodeex{\basefield}{\field}{\degree}$, i.e:
    \begin{enumerate}
        \item $\normalizedcodedistanceex{\basefield}{\variety}{\degree} \geq \normalizedcodedistanceex{\basefield}{\field}{\degree}$.
        \item $\listdecodingradiusex{\basefield}{\variety}{\degree} \geq \listdecodingradiusex{\basefield}{\field}{\degree}$.
    \end{enumerate}
\end{corollary*}

\paragraph{Main Technical Challenge}
We achieve these results by combining the two black-box properties of subsets $\variety \subseteq \field$ we presented.
Analysis of the polynomials in $\variety$ raises a new challenge, as previous techniques that were used to analyze low-degree polynomials,
both regarding $\field$~\cite{green2007distribution} and regarding subsets $\variety$~\cite{lampert2021relative},
were focused on maintaining the behavior of polynomials \emph{in the set they work on} ($\field$ and $\variety$ accordingly).
\newline
The novelty of our new technique is that it uses a similar approach to analyze polynomials $\variety$ as commonly used in $\field$,
\emph{while simultaneously maintaining a connection} between polynomials in $\variety$ to polynomials in $\field$.
This connection allows us to deduce that polynomials in $\variety$ behave similarly to polynomials in $\field$.
Informally, given a question regarding a polynomial in $\variety$, our new technique allows us to
associate it with a ``correct'' lift of it, and answer the question \emph{using properties of its lift}.
We emphasize that the correct lift (the one we later choose to use) \emph{depend} on the question,
thus we cannot pick a single canonical lift to \emph{generally} describe each polynomial in $\variety$.

Next we describe this challenge in more detail.
\newline
Analyses of polynomials in $\field$ were commonly based on the structure-randomness connection of polynomials in $\field$.
To use this connection, a procedure introduced by~\cite{green2007distribution}, which is called the \emph{regularization process}, is often used~\cite{kaufman2008worst, tao2011inverse, hatami2011higher, bhattacharyya2013locally, bhattacharyya2013algorithmic, DBLP:journals/corr/0001L15}.
This procedure takes any collection of polynomials, and constructs from it another collection of polynomials that has \emph{equidistriubtion in $\field$}
and ``captures'' all functions ``captured'' by the previous collection.
This notion of ``capturing'' is formulated by a definition called \emph{measurable},
and thus it is required that every function measurable by the old collection will be measurable by the new collection.

We note that the regularization procedure achieves random behavior in $\field$ by requiring the collection to have an \emph{extremely low algebraic structure}:
This implies the new collection has random behavior (equidistributed) as it is a property of $\field$.
The notion of structure is captured by a definition called \emph{rank}, where a polynomial with high rank has extremely low structure.
Additionally, the notion of lack of random behavior is captured by a definition called \emph{bias}, where a polynomial with low bias behaves randomly (equidistributed).
Therefore, the equidistribution is achieved in the regularization process by constructing a collection with high rank, as \emph{in $\field$ high rank implies low bias}.

To generalize these ideas to $\variety$, one must achieve a similar result in $\variety \subseteq \field$:
Given any collection of polynomials, construct a new collection of polynomials that is both equidistributed in $\variety$
and captures every function in $\variety$ that was previously captured.
In our case, however, we must also ensure that the new collection also captures all functions that were previously-captured \emph{in $\field$},
as in our case we wish to use the connection of polynomials in $\variety$ to polynomials in $\field$.
This can be summarized by 3 requirements:
\begin{enumerate}
    \item The polynomials in the new collection will behave random in $\variety$.
    \item Every function that was measurable in $\variety$ by the old collection will be measurable by the new collection in $\variety$.
    \item Every function that was measurable \emph{in $\field$} by the old collection will be measurable by the new collection \emph{in $\field$}.
\end{enumerate}
Alas, this third-requirement is incompatible with the way we achieve the first requirement.
Achieving the first requirement, which is the random behavior in $\variety$, is done by requiring an extremely low algebraic structure \emph{according to relative rank}.
This requires one to consider all possible lifts of polynomials in the collection to avoid any structure.
\newline
More accurately
\footnote{As the polynomials we have here are polynomials in $\field$ we can not discuss their lift.}
, for a polynomial $\funcdef{\genpoly}{\field}{\basefield}$,
we define an \emph{$\variety$-equivalent polynomial for $\genpoly$} to be
a polynomial in $\field$ that coincides with $\genpoly$ on $\variety$ and has the same degree bound as $\genpoly$
\footnote{This is the same as considering all lifts of the polynomial $\restrictfunc{\genpoly}{\variety}$, assuming such lift exist.}
.
Using this definition, the definition of relative rank requires examining all possible \emph{$\variety$-equivalent} polynomials,
and ensuring non of them exhibit structure.
\newline
Typically (in $\field$ for example), avoiding structure is achieved by replacing every structured polynomial by a \emph{small}
set of less-structured polynomials that capture it.
We note that it is \emph{crucial} that the set is small, and from reasons we did not explain here (see definition~\ref{definition:rank}), it is promised because the polynomial we wish to replace is structured.
\newline
For $\variety$, we aim to avoid \emph{all} $\variety$-equivalent polynomials of a polynomial from being structured.
Achieving this, while keeping the collection small,
requires one to replace the polynomial by a set of less-structured polynomials that capture a \emph{structured-$\variety$-equivalent} of it.
Therefore, this process creates a new collection that captures this $\variety$-equivalent polynomial,
but does not necessarily capture the original polynomial!
\newline
In summary, the challenge is that avoiding the structure of \emph{all} the lifts of a polynomial to achieve equidistribution in $\variety$,
without adding too many polynomials, may harm the functions we capture in $\field$.

\paragraph{Introducing New Tools}
We overcome this challenge by presenting a new definition that relaxes the notion of \emph{measurable} we required for functions in $\field$,
which we call \emph{$\variety$-measurable}.
This enables us to describe a relaxed version of the regularization process,
in which we require that every function in $\field$ that was $\variety$-measurable by the old collection will still be $\variety$-measurable by the new collection.
In contrast to the original regularization process, which mandated that functions that were measurable by the old collection will be measurable by the collection,
this relaxed definition only requires such functions to be \emph{$\variety$-measurable} by the new collection.

Even though we no longer need to capture all previously captured functions in $\field$,
it is important that the new relaxed-definition is strict enough to keep the connection between polynomials in $\variety$ and in $\field$.
Therefore, maintaining the $\variety$-measurable functions throughout the regularization process cannot be done trivially,
and this is handled in a procedure we call \emph{the $\variety$-relative regularization process} which is a stronger-version of the regularization process that is used in $\field$.
This new definition and procedure are thoroughly described in Section~\ref{sec:regularization-relative-to-X}.

We note that these new definition and procedure are a novel contribution of this work, and we believe they can
be useful in future research of the quotient Reed-Muller code.

\subsection{Comparison to Related Work}\label{subsec:previous-work}
In~\cite{bhowmick2014list} the authors studied the list decoding radius of Reed Muller codes $\field$.
They proved that, for prime fields, the list decoding radius \emph{reaches the distance of the code}, as conjectured earlier by~\cite{10.1145/1374376.1374417}
\footnote{Note that it is known that $\listdecodingradiusex{\basefield}{\field}{\degree} \leq \normalizedcodedistance{\basefield}{\degree}$,
    and therefore, in a sense, their result is \emph{optimal in $\field$} assuming $\degree, \abs{\basefield}$ are considered as constants.}
\footnote{We also note that their work also apply to the regime $\degree \geq \abs{\basefield}$. }
.
Formally, they showed the following theorem:
\begin{theorem}~\cite[Theorem 1]{bhowmick2014list}
Let $\basefield$ be a prime field.
Let $\epsilon > 0$ and $\degree, \blocklength \in \naturalnumbersset$.
There exists a constant
\footnote{It is important to note that $c$ is \emph{independent of $\blocklength$}.}
$c \definedas c(\abs{\basefield}, \degree, \epsilon)$ such that:
\[
    \listpolycount{\basefield}{\field}{\degree}{\normalizedcodedistance{\basefield}{\degree}- \epsilon} \leq c
\]
\end{theorem}
Our work gives new tools for analyzing polynomials in $\variety \subseteq \field$,
which we later use to follow their line of proof and show equivalent result \emph{in $\variety$}.

We next present related work regarding the study of polynomial codes in subsets $\variety \subseteq \field$.
Before presenting them specifically, we note that our work has a \emph{fundamental difference} than that of the previous study of polynomials in subsets.
Most works which studied polynomials over subsets $\variety \subseteq \field$ were focused on subsets in which every polynomial has a \emph{unique} lift.
This ensures that there is a 1-to-1 correspondence between polynomials in $\variety$ and in $\field$
and therefore allows easier connection between polynomials in $\variety$ and in $\field$.
\newline
We note that our work is non-trivial even in this case:
it extracts the properties of $\field$ that were used in~\cite{bhowmick2014list}, in a way they can be used to analyze quotient Reed-Muller codes.
However, as described earlier, our work addresses an additional substantial challenge which arise when the lift is \emph{not} unique.
Thus our work is only comparable to other works in the unique-lift case, which is the less-challenging case we address.

The first line of work we mention is this regard is the study of hitting sets for low degree polynomials~\cite{6243404, 10.1145/2554797.2554828, 6875485},
and a stronger variant of it which is the study of pseudorandom-generators against low degree polynomials~\
    \cite{10.1145/1060590.1060594, 4389478, 10.1145/1374376.1374455, 4558816,  Cohen2013PseudorandomGF, derksen2022fooling, dwivedi2024optimalpseudorandomgeneratorslowdegree}
Both definitions capture subsets
\footnote{Sometimes this subset is allowed to be a \emph{multiset}.}
$\variety \subseteq \field$ such that every polynomial over $\field$ has a non-negligible distance from $0$ \emph{when restricted to $\variety$}.
This requirement implicitly implies that every low degree polynomial over $\variety$ has at most a \emph{single} lift.

Another line of work worth mentioning in this regard is~\cite{4558818, guruswami2017efficientlylistdecodablepuncturedreedmuller},
which studied \emph{puncturing of Reed-Muller codes}.
This line of work studied the construction of sets $\variety \subseteq \field$,
such that puncturing Reed-Muller codes over $\variety$, that is, taking every original codeword and \emph{restricting} it to $\variety$, will yield a good error-correction code.
To perform their analysis, it was important that every polynomial in $\variety$ has at most a single lift,
and therefore it was an assumption in their work.

The papers~\cite{brakensiek2024genericreedsolomoncodesachieve, alrabiah2024randomlypuncturedreedsolomoncodes, brakensiek2024generalizedgmmdspolynomialcodes}
also studied similar questions.
This line of work is followed by the resolution of the \emph{GM-MDS conjecture}, which was proved by~\cite{DBLP:journals/corr/abs-1803-02523, DBLP:journals/corr/abs-1803-03752}.
\newline
We note that these works
were focused on the regime where the field is \emph{large}.
More specifically,
they require that the field is \emph{large in respect of $\blocklength$}, i.e $\Omega(\blocklength)$.
We emphasize that our work is focused on \emph{constant fields}.
Moreover, their results were regarding \emph{random} puncturing, while our result makes an \emph{explicit} puncturing.

%
We also note that most studies presented above also achieved results regarding the \emph{rate} of the punctured code.
This property of the code can be analyzed naturally when each polynomial over $\variety$ has a \emph{unique} lift, as such assumption implies that the number of polynomials remains the same in $\variety$ as of in $\field$.
As our work does \emph{not} assume such uniqueness, the rate of the code we consider is not analyzed in our work, and thus remained \emph{an open problem}
\footnote{Note that is highly dependent on $\variety$, as additional assumptions are needed to acheive good results in this regard.}
.

\subsection{Proof Overview}\label{subsec:our-work}
In this subsection we present our main technical contribution, which is how we address the challenge of \emph{non-unique lift}.
This is done by introducing the definition of being \emph{$\variety$-measurable}, and by presenting a new tool which is the \emph{relative regularization process}.

To describe them clearly, we first elaborate more on two definitions we described briefly.
\paragraph{Measurable}
Suppose we have a collection of polynomials
\footnote{In this context we think of $c$ as a small (constant for example).}
$\genpolyset[1] = \parens{\genpoly_1,...,\genpoly_c}$ where $\funcdef{\genpoly_i}{\field}{\basefield}$ is a polynomial of degree $\leq \degree$.
We say a function $\funcdef{\genfunc}{\field}{\basefield}$ is \emph{measurable in respect of $\genpolyset[1]$} if it can be determined by the values of $\genpoly_1,...,\genpoly_c$:
if one knows the values of $\genpoly_1(x),...,\genpoly_c(x)$, then she also knows the value of $\genfunc(x)$.
This mathematical-analysis notion, which was first used in a similar context in~\cite{green2007primescontainarbitrarilylong}, is formally defined as follows:
\begin{definition}[Measurable]
    We say a function $\funcdef{\genfunc}{\field}{\basefield}$ is \emph{measurable in respect of $\genpolyset[1] = \parens{\genpoly_1,...,\genpoly_c}$} if
    there exists $\funcdef{\Gamma_{\genfunc}}{\basefield^c}{\basefield}$ such that:
    \[
        \genfunc(x) = \Gamma_{\genfunc}(\genpoly_1(x),...,\genpoly_c(x))
    \]
\end{definition}
This definition can be thought of as the collection $\genpolyset$ ``captures'' the function $\genfunc$
\footnote{Note that this definition also generalizes to every collection of functions.
For now, one can think of the collcetion as a collection of bounded degree polynomials.}
\footnote{One can think of this definition as a generalization of linear span: the collection \emph{spans} the function, where $\Gamma$ is some notion of a span.}
.

Moreover, it would have been useful had this collection of polynomials been ``pseudo-random'', i.e the vector $\parens{\genpoly_1(x),...,\genpoly_c(x)}$ would be equidistributed over a random input $x \in \field$.
This equidistribution would allow us to better understand functions $\genfunc$ that are measurable in respect of $\genpolyset$.

As $\field$ has the rank-bias property, this equidistribution can be achieved by requiring $\genpolyset$ to be a collection of high-rank.
This is a fundamental idea behind the regularization process, first presented in~\cite{green2007distribution}.
Given a collection of polynomials $\genpolyset$, the regularization process constructs another collection $\genpolyset[2]$ of polynomials (with the same degree bound),
such that $\genpolyset[2]$ is a collection of high-rank (and therefore equidistributed) that \emph{refines} $\genpolyset$.
By refine, we mean that every function that was measurable by the first collection $\genpolyset$ is also measurable by the new collection $\genpolyset[2]$ (See definition~\ref{def:semantic-refinement}).

\paragraph{Relative Rank}
We remind the reader that $\variety$-relative rank is a notion that measures the algebraic structure of a polynomial in a subset $\variety \subseteq \field$, by considering the structure of all of its $\variety$-equivalent polynomials.
This notion was presented by~\cite{gowers2022equidistributionhighrankpolynomialsvariables, lampert2021relative}, and is used to achieve equidistribution in $\variety$ assuming $\variety$ has relative rank-bias property.
It is defined as follows:
\begin{definition}[Relative Rank, informal.
See definition~\ref{def:relative-rank-of-polynomial}]
    Let $\variety \subseteq \field$ be a subset,
    let $\degree \in \naturalnumbersset$, and let $\funcdef{\genpoly}{\field}{\basefield}$ be a polynomial of degree $= \degree$.
    The $\variety$-relative rank of $\genpoly$ is defined as follows:
    \[
        \relrank{\variety}{\genpoly} \definedas \min \set{\rank{\genpoly - \relativeremainder{\genpoly}} \suchthat
        \relativeremainder{\genpoly} \in \allpolyset{\leq \degree}{\field}{\basefield}, \restrictfunc{\relativeremainder{\genpoly}}{\variety} \equiv 0}
    \]
\end{definition}

\subsubsection{\titlevariety-measurable and The \titlevariety-Relative Regularization Process}
In this subsection we discuss the generalization of the regularization process to subsets $\variety \subseteq \field$ using the equivalent of rank-bias relation in $\variety$.
We name this tool \emph{the relative regularization process}.

Practically, we use this tool to show that given a specific question in mind, every $\funcdef{\onvarpoly}{\variety}{\basefield}$ has some polynomial $\funcdef{\genpoly}{\field}{\basefield}$ that behave ``similarly'' in respect to this question.
This allows us to pull properties of $\genpoly$ to better understand $\onvarpoly$.
The perfect candidate for such $\genpoly$ is a \emph{lift} of $\onvarpoly$.
\newline
In order to use $\genpoly$ to deduce properties of $\onvarpoly$, we use the well-studied properties of polynomials in $\field$ to acheive properties of $\genpoly$, and relate these to properties of $\onvarpoly$.
More specifically, assume that $\onvarpoly$ and $\genpoly$ are measurable in respect of a collection of polynomials $\genpolyset$ (each in its domain).
Our strategy is to use $\genpoly$ to deduce properties of $\Gamma_{\genpoly}$, and then use the properties of $\Gamma_{\genpoly}$ to deduce properties of $\onvarpoly$.

Now let us describe the extra challenge.
We start by following the ideas of the regularization process we described for $\field$.
Assuming the collection is not a collection of $\variety$-relative high rank, then there must exist a polynomial in the collection that has low \emph{relative} rank, which we denote by $\genpoly^\star$
\footnote{More precisely, some linear combination of polynomials has low relative rank.}
.
Note that in relative rank, this does not necessarily mean that $\genpoly^\star$ is of low rank, but that there exists another $\variety$-equivalent polynomial that has a low rank.
Thus, even if we remove the low-rank $\variety$-equivalent polynomial and add to the collection all the polynomials that decomposed it,
we cannot require that every function that was measurable by the old collection will still be measurable by the new collection:
even the polynomial we removed is not necessarily measurable by the new collection!
\newline
To allow such regularization process to still apply, we note that while $\genpoly$ might not be measurable in respect of the new collection, a $\variety$-equivalent polynomial of $\genpoly$ \emph{is} measurable with respect of it.
Therefore, we relax the notion of being measurable to being \emph{$\variety$-measurable}.
\newline
We say a function $\genfunc$ is $\variety$-measurable in respect of $\genpolyset$ if it can be determined by the polynomials of $\genpolyset$
\emph{up to a valid $\variety$-remainder}.
We first describe an incomplete definition, then present the challenge that rises with it, and finally present its resolution.
\begin{definition}[$\variety$-measurable, Incomplete Definition]
\footnote{This incomplete definition lacks the requirement of the \emph{validity} of the $\variety$-remainder}
We say a function $\genfunc$ is $\variety$-measurable
if there exists a function $\funcdef{\Gamma}{\basefield^c}{\basefield}$
and a $\variety$-remainder, i.e a function $\funcdef{\relativeremainder{\genfunc}}{\field}{\basefield}$ with $\restrictfunc{\relativeremainder{\genfunc}}{\variety} \equiv 0$
such that:
\[
    \forall a \in \field: \genfunc(a) = \Gamma(\genpoly_1(a),...,\genpoly_c(a)) + \relativeremainder{\genfunc}(a)
\]
\end{definition}

Previous works analyzing polynomials in $\field$ were able to deduce two things from $\genfunc$ being measurable by $\genpolyset$:
that the structure of $\Gamma$ is similar to the structure of $\genfunc$, and that a random input of $\Gamma$ behave similarly to a random input of $\genfunc$.
\newline
To study polynomials in $\variety$, we wish to connect $\onvarpoly$ to $\genpoly$ (which is a lift of $\onvarpoly$).
Thus, we think of $\genfunc = \genpoly$, and require two similar things.
Firstly, we want the structure of $\Gamma$ to be similar to the structure of $\genfunc$ (in this case, $\genpoly$), which we understand as $\genfunc$ is a polynomial in $\field$.
Secondly, we want a random input of $\Gamma$ to behave similarly to a random input of $\onvarpoly$, as $\onvarpoly$ is the polynomial we wish to understand.
The latter is easily achieved using the fact high $\variety$-relative rank implies equidistribution in $\variety$.
The former, however, might be damaged by the remainder as we defined it: we can only learn the structure of $\Gamma$ using the structure of $\genfunc - \relativeremainder{\genfunc} = \Gamma(\genpoly_1,...,\genpoly_c)$.
However, the structure of $\genfunc - \relativeremainder{\genfunc}$ can be very different from the structure of $\genfunc$,
as we did not require any structure of the $\variety$-remainder $\relativeremainder{\genfunc}$.
Thus, we can not deduce the structure of $\Gamma$ via the structure of $\genfunc$ using the incomplete definition described above.

To handle this issue, we add one more requirement regarding the $\variety$-remainder,
which ensures that the structure of $\genfunc$ can be understood via the structure of $\Gamma$:
\[
    \deg(\genfunc - \relativeremainder{\genfunc}) \leq \deg(\genfunc)
\]
If the $\variety$-remainder also has this property, we say it is a \emph{valid} $\variety$-remainder for $\genfunc$.
This can be summarized by the following (complete) definition:
\begin{definition}[$\variety$-measurable]
    We say a function $\genfunc$ is $\variety$-measurable
    if there exists a function $\funcdef{\Gamma}{\basefield^c}{\basefield}$
    and a \emph{valid} $\variety$-remainder, i.e a function $\funcdef{\relativeremainder{\genfunc}}{\field}{\basefield}$
    with $\restrictfunc{\relativeremainder{\genfunc}}{\variety} \equiv 0$ and $\deg(\genfunc - \relativeremainder{\genfunc}) \leq \deg(\genfunc)$
    such that:
    \[
        \forall a \in \field: \genfunc(a) = \Gamma(\genpoly_1(a),...,\genpoly_c(a)) + \relativeremainder{\genfunc}(a)
    \]
\end{definition}
We use this new definition the following way:
Instead of using $\genfunc$ to understand $\Gamma$, we use $\genfunc - \relativeremainder{\genfunc}$ to do so.
We choose $\genfunc - \relativeremainder{\genfunc}$ as it has the same structure as $\genfunc$, but it is ``closer'' to the function $\Gamma$ as $\genfunc - \relativeremainder{\genfunc} = \Gamma(\genpoly_1,...,\genpoly_c)$
\footnote{One can think of this step as "taking the right $\variety$-equivalent" in respect of $\genpolyset$.}
.
Finally, as $\Gamma$ behaves similarly to $\onvarpoly$ for random inputs, we can use $\Gamma$ to deduce properties regarding $\onvarpoly$.
\newline
With this in hand, let us finish describing the relative-regularization process.
The requirement on the validity of the $\variety$-remainder raises a new challenge in the $\variety$-relative regularization process:
we need to somehow control the structure of the $\variety$-remainder, even though this ``error''
is substituted in $\Gamma$ each time we wish to replace a polynomial in our collection.
We address this challenge using a Lemma proved in~\cite{DBLP:journals/corr/0001L15} called the ``faithful composition lemma'',
which allows us to deduce strong properties regarding the structure of $\Gamma$ given the collection was of a high (regular) rank in the first place.
Therefore, we add to each step of the relative-regularization process a (regular) regularization, which ensures $\Gamma$ is very structured.
This strong structure of $\Gamma$ is later used to control the error and deduce it is in the form of a valid $\variety$-remainder.
For the exact details, see Theorem~\ref{theorem:regularization-in-X}.
We conclude this by informally stating our main technical theorem, which is the relative regularization process we just described:
\begin{theorem}[Relative Regularization Process, Informal, See Theorem~\ref{theorem:regularization-in-X}]
Let $\rankval, \degree \in \naturalnumbersset$ be integers that represents a requested rank and degree respectively,
and let $\genpoly_1,...,\genpoly_c$ be a collection of polynomials of degree $\leq \degree$.
Then, there is another collection $\genpoly^{\prime}_1,...,\genpoly^{\prime}_{c^\prime}$ of polynomials of degree $\leq \degree$,
such that:
\begin{enumerate}
    \item Every function that is $\variety$-measurable in respect to the first collection is also $\variety$-measurable in respect to the new collection.
    \item The new collection is of $\variety$-relative rank $\geq \rankval$.
    \item The new collection is of bounded size, i.e $c^\prime \leq C_{\rankfunc, \degree, c}$.
\end{enumerate}
\end{theorem}

\subsubsection{List Decoding in \titlevariety via \titlevariety-Relative Regularization}
In this subsection, we demonstrate how to use the relative regularization process to achieve our main theorem: analysis of the list decoding radius of $\reedmullercodeex{\basefield}{\variety}{\degree}$.

We follow the line of proof of~\cite{bhowmick2014list}, but this time, we are interested in bounding the amount of polynomials \emph{in $\variety$} around every function \emph{in $\variety$}.
More specifically, we wish to show that there is a constant number of words that are $(\normalizedcodedistance{\basefield}{\degree} - \epsilon)$-close to any fixed function in $\variety$.

Let $\funcdef{\onvarfunc}{\variety}{\basefield}$ be a received word.
First, we apply a lemma proved in~\cite[Corollary 3.3]{bhowmick2014list}.
The lemma shows that there is a constant-sized (depending on $\epsilon$) collection of polynomials in $\variety$, denoted by $\onvarpolyset[3]$,
such that the distance of $\onvarfunc$ to \emph{any} polynomial can be approximated by the distance of $\onvarfunc$ to some function that is measurable by $\onvarpolyset[3]$ in $\variety$.
This means that instead of bounding the number of polynomials in the radius of $\onvarfunc$, one can bound the number of polynomials in the radius of some function measurable by $\onvarpolyset[3]$.
Thereby, every polynomial-specific measurable function can be thought of as a \emph{low complexity proxy} for $\onvarfunc$ in respect to the polynomial.

Next, we lift each polynomial from $\onvarpolyset[3]$ and apply the \emph{relative regularization process}.
This yields a new collection of polynomials in $\field$ that is constant sized and randomly-behaving (in $\field$).
Denote this new collection by $\genpolyset[3]^{\prime}$
\footnote{We use the same notations as the original proof for clearannce.}
.
Thereby, the question of list decoding is reduced to the following question:
We have a specific constant-sized randomly-behaving collection of polynomials $\genpolyset[3]^\prime = \set{\genpoly[3]_1^{\prime},...,\genpoly[3]^{\prime}_{c^\prime}}$
that was constructed using the function $\onvarfunc$.
We need to bound the amount of polynomials in $\variety$ that are $(\normalizedcodedistance{\degree}{\basefield} - \epsilon / 2)$-close to be measurable by this collection in $\variety$.
Note that the randomly-behaving property was achieved using the \emph{relative rank-bias property} of $\variety$.
Additionally, we note the collection $\genpolyset[3]^\prime$ is a collection of polynomials in $\field$ which we obtained by using the \emph{lift-enabler property} of $\variety$.

From there (and similarly to the analysis in $\field$),
the strategy is to show that polynomials that are that close to being measurable by the randomly-behaving collection $\genpolyset[3]^\prime$, are in fact \emph{measurable} by it.
This will bound the number of such polynomials by the amount of possible functions that are measurable by $\genpolyset[3]^\prime$, which is constant as the collection is of constant size.

Let $\funcdef{\onvarpoly}{\variety}{\basefield}$ be a polynomial of degree $\leq \degree$, and consider a lift of it $\funcdef{\genpoly}{\field}{\basefield}$.
Consider the collection $\genpolyset[3]^\prime \cup \set{\genpoly}$.
Surely, $\genpoly$ is measurable by this collection in $\field$.
Applying $\variety$-relative-regularization to this collection yields a new collection $\genpolyset[3]^{\prime\prime}$ that is equidistributed in $\variety$, such that every $\variety$-measurable function by the old collection is $\variety$-measurable by the new collection.
By a reason we have not explained in this brief explanation, we can ensure this collection is of the form $\genpolyset[3]^{\prime\prime} = \genpolyset[3]^\prime \cup \set{\genpoly[3]_{1}^{\prime\prime},...,\genpoly[3]_{c^{\prime\prime}}^{\prime\prime}}$.

As $\genpoly$ was $\variety$-measurable by $\genpolyset[3]^\prime \cup \set{\genpoly}$ (it was even measurable), $\genpoly$ is $\variety$-measurable by the new collection $\genpolyset[3]^{\prime\prime}$:
That is, $\genpoly$ is measurable by $\genpolyset[3]^{\prime\prime}$ up to a \emph{valid} remainder, denoted by $\relativeremainder{\genpoly}$.
\newline
This means there exists $\funcdef{\Phi}{\basefield^{c^\prime + c^{\prime\prime}}}{\basefield}$ such that:
\[
    \forall a \in \field: \genpoly(a) = \Phi(\genpoly[3]^\prime_1(a),...,\genpoly[3]^\prime_{c^\prime}(a), \genpoly[3]^{\prime\prime}_1(a),...,\genpoly[3]^{\prime\prime}_{c^{\prime\prime}}(a))) + \relativeremainder{\genpoly}(a)
\]

In $\field$, the proof would follow by studying the structure of the function $\Phi$ and use it to induce that $\Phi$ does not depend on its last $c^{\prime\prime}$ variables.
This implies that $\genpoly$ is measurable by the original collection $\genpolyset[3]^{\prime}$ which concludes the proof
\footnote{Note that in $\field$ there is no remainder, so the equation above (with the last $c^{\prime\prime}$ variables as constants) implies measurability by $\genpolyset[3]^{\prime}$.}
.

More accurately, the analysis in $\field$ used the fact that substituting \emph{randomly behaving} polynomials in $\Phi$ yields a structured function
\footnote{In our notations, this structured function is $\genpoly$, which is a polynomial of degree $\leq \degree$ and thus structured}
.
This is used to show that $\Phi$ as a function by itself, with inputs from $\basefield^{c^\prime + c^{\prime\prime}}$, is a very structured function.
The strong structure of $\Phi$, with the fact that $\Phi$ (with inputs substitued to be the functions of $\genpolyset[3]^{\prime\prime}$) is close to the function $\onvarfunc$,
are then combined to deduce that $\Phi$ does not depend on its last $c^{\prime\prime}$ variables.

This paradigm can not be extended effortlessly to our case.
In $\variety$, deducing that $\Phi$ is very structured requires a one-more major step.
This is because we do \emph{not} know any correspondence in the behavior of $\Phi$ (which we want to understand) with the behavior of $\genpoly$ (which we know is structured).
We only know there is a correspondence between $\Phi$ to another function $\genpoly - \relativeremainder{\genpoly}$, which apriori we do not know is structured!

Fortunately, the relative regularization process (Theorem~\ref{theorem:regularization-in-X}) mandates that the remainder of the measurement is \emph{valid}.
That is, if $\genpoly$ was structured (a polynomial of degree $\leq \degree$), then so does $\genpoly - \relativeremainder{\genpoly}$.
This is \emph{crucial}, as it allows us to use the relation between $\Phi$ and $\genpoly - \relativeremainder{\genpoly}$ to deduce that $\Phi$ is structured,
and continue the original outline of the proof of~\cite{bhowmick2014list}.
For more details in this regard, see Theorem~\ref{thm:list-decoding-RM-in-X}.

\subsection{Organization}\label{subsec:organization}
In Section~\ref{sec:preliminaries} we present some basic notations and conventions,
and define the preliminaries we have regarding high-order Fourier analysis in $\field$: polynomials, rank and regularization.
We later generalize each component we presented in Section~\ref{sec:preliminaries} to study polynomials in $\field$ to also study polynomials in $\variety$:
in Section~\ref{sec:polynomials_in_X} we present the set of polynomials in $\variety$ and present the \emph{lift-enabler property};
in Section~\ref{sec:relative-rank-bias-property} we present the \emph{$\variety$-relative rank-bias property};
and in Section~\ref{sec:regularization-relative-to-X} we present the $\variety$-measurable notion, and our main tool, which is the \emph{$\variety$-relative regularization process}.
Next, we present two applications regarding the distance parameters of Reed-Muller codes in $\variety$:
In Section~\ref{sec:radius-of-RM-over-X} we prove the inheritance of the \emph{distance} of the code;
and in Section~\ref{sec:list-decoding-reed-muller-over-X} we prove the inheritance of the \emph{list decoding distance} of it (which is much more involved).


%% file: preliminaries.tex
\section[Preliminaries]{Preliminaries}\label{sec:preliminaries}

\subsection{Basic Definitions and Notations}\label{subsec:definitions_and_notations}
We denote by $\naturalnumbersset$ the set set of integers, i.e natural numbers (excluding 0).
For an integer $k$ we denote $\sparens{k} \definedas \set{1,2,...,k}$.
We use $y = x \pm \epsilon$ to denote $y \in \sparens{x - \epsilon, x + \epsilon}$, and similarly $y = x \mp \lambda$ to denote $y \in \sparens{x + \lambda, x - \lambda}$ (usually when $\lambda < 0$).
\newline
Fix a prime field $\basefield = \basefield_{\genprime}$.
Denote by $\abs{\cdot}$ the natural map of $\basefield$ to $\set{1,...,\genprime - 1} \in \naturalnumbersset$.
We denote the character from $\basefield$ by $\charfunc{x} \definedas e^{2 \pi i \cdot \abs{x}}$

Generally speaking and unless stated otherwise, we use the following conventions:
We use $\blocklength \in \naturalnumbersset$ to denote the number of variables in Reed-Muller code.
We use $\degree$ to denote a degree (typically the degree of the polynomials in our code), and $\variety$ to denote the subset of $\field$ we work in i.e $\variety \subseteq \field$.
Properties of the subset $\variety$ will usually be denoted with $\tilde{\square}$.
We use $\genfunc[1], \genfunc[2], \genfunc[3]$ to denote general functions with domain $\field$, and $\onvarfunc[1], \onvarfunc[2], \onvarfunc[3]$ to denote functions with domain $\variety$.
We use $\genfuncset[1], \genfuncset[2], \genfuncset[3]$ and $\onvarfuncset[1], \onvarfuncset[2], \onvarfuncset[3]$ respectively to denote sets of such functions.
Similarly, we use $\genpoly[1], \genpoly[2], \genpoly[3]$ to denote polynomials with domain $\field$, and $\onvarpoly[1], \onvarpoly[2], \onvarpoly[3]$ polynomials with domain $\variety$.
We use $\genpolyset[1], \genpolyset[2], \genpolyset[3]$ and $\onvarpolyset[1], \onvarpolyset[2], \onvarpolyset[3]$ respectively to denote sets of such polynomials.

\subsection[Polynomials in \titlefield]{Polynomials in \titlefield}\label{subsec:polynomials_in_Fn}
We start by presenting a standard definition for a polynomial over a finite field.
\begin{definition}[Polynomial: Global Definition]
    Let $\degree \in \naturalnumbersset$ be a constant.
    A function $\funcdef{\genpoly}{\field}{\basefield}$ is called a \emph{polynomial of degree $\leq \degree$} if
    it is of the following form:
    \[
        \genpoly(x_1,...,x_{\blocklength}) =
            \sum_{0 \leq \degree_1,...,\degree_{n} : \sum_{i=1}^{\blocklength}{\degree_i} \leq \degree}
                {c_{\degree_1,...,\degree_n} \prod_{i=1}^{\blocklength} x_i^{\degree_i}}
    \]
    We denote the set of all polynomials of degree $\leq \degree$ by $\allpolyset{\leq \degree}{\field}{\basefield}$.
    The value $\degree$ in the definition above is called the \emph{global degree} of the function $\genpoly$, shorthand by the \emph{degree} of $\genpoly$,
    and it is denoted by $\deg(\genpoly) = \degree$.
    \newline
    Additionally, the set of all polynomials from $\field$ to $\basefield$ of degree $\leq \degree$ is denoted by:
    \[
        \allpolyset{\leq \degree}{\field}{\basefield}
    \]
\end{definition}
Next, we present a known equivalent definition for a polynomial using derivatives.
To do so, we first define a derivative in the case of finite fields.
\begin{definition}[Derivative]
    Given a function $\funcdef{\genfunc}{\field}{\basefield}$ and $a \in \field$,
    we define the derivative of $\genfunc$ in direction $a$ as a function $\funcdef{\directionderivative_{a}\genfunc}{\field}{\basefield}$
    defined as follows:
    \[
        \directionderivative_a \genfunc(x) \definedas \genfunc(x + a) - \genfunc(x)
    \]
\end{definition}
\begin{lemma}
    Let $\degree \in \naturalnumbersset$.
    A function $\funcdef{\genfunc}{\field}{\basefield}$ is a polynomial of degree $\leq \degree$
    if and only if $\directionderivative_{a}\genfunc$ is a polynomial of degree $\leq \degree - 1$ for all $a \in \field$.
\end{lemma}
This leads us to a natural definition of a \emph{degree} of a function using derivatives.
\begin{definition}[Local Degree]
    For a function $\funcdef{\genfunc}{\field}{\basefield}$, we define its \emph{local degree},
    to be the least integer $\degree \in \naturalnumbersset$ such that
    for all $a_1,...,a_{\degree + 1}, x \in \field$:
    \[
        \directionderivative_{a_{\degree + 1}}...\directionderivative_{a_{1}} \genfunc(x) = 0
    \]
\end{definition}
In $\field$, the two definitions of degree coincide, and we get a single definition of a degree:
\begin{lemma}[Equivalance of definitions of a degree]\label{lemma:alternative-definition-for-polynomial-using-derivatives}
    Let $\funcdef{\genfunc}{\field}{\basefield}$ be a function, and let $\degree \in \naturalnumbersset$ be an integer.
    Then, the global degree of a $\genfunc$ \emph{equals} its local degree.
\end{lemma}
\begin{remark}
    We sometimes refer to the requirement that the local degree of a function is $\leq \degree$,
    as the \emph{local criteria} of degree $\leq \degree$ polynomials.
\end{remark}

\subsection[Rank-bias in \titlefield]{Rank-bias in \titlefield}
We start by defining the notion of \emph{bias}, which is a measure of how the function is far from being equidistributed (see Appendix~\ref{sec:equidistribution-of-functions} for the exact details).
\begin{definition}[Bias]
    Let $\funcdef{\genfunc}{\field}{\basefield}$.
    The \emph{bias} of the function $\genfunc$ is defined in the following way:
    \[
        \bias{\genfunc} \definedas 1 / \abs{\field} \cdot \sum_{x \in \field}{\charfunc{\genfunc(x)}}
    \]
    Moreover, for a subset $\variety \subseteq \field$, we define the \emph{bias of $\genfunc$ in $\variety$} to be:
    \[
        \relbias{\variety}{\genfunc} \definedas 1 / \abs{\variety} \cdot \sum_{x \in \variety}{\charfunc{\genfunc(x)}}
    \]
\end{definition}
Next, we present a standard definition of rank of a polynomial, which is a notion that measures how \emph{structured} is the function.
Note that low rank implies the polynomial is highly structured.
Formally we have the following definition:
\begin{definition}[Rank of a Polynomial]\label{definition:rank}
    Given a constant $\degree \in \naturalnumbersset$ and a polynomial $\genpoly$,
    the \emph{$\degree$-rank} of $\genpoly$, denoted as $\drank{\degree}{\genpoly}$ is defined to be
    the smallest integer $\rankval$ such that $\genpoly$ can be computed given $\rankval$ polynomials of degree $< \degree$.
    In other wards, we say $\drank{\degree}{\genpoly} = \rankval$ if $\rankval$ is the smallest integer such that
    there exists $\rankval$ polynomials $\genpoly[2]_1,...,\genpoly[2]_\rankval \in \allpolyset{\leq \degree - 1}{\field}{\basefield}$
    and a function $\funcdef{\Gamma}{\field}{\basefield}$ such that:
    \[
        \genpoly(x) = \Gamma \parens{\genpoly[2]_1(x),...,\genpoly[2]_\rankval(x)}
    \]
    If $\degree = 1$, then $1$-rank is defined to be $\infty$ if $\genpoly$ is non constant, and $0$ otherwise.
    \newline
    Moreover, for a polynomial $\genpoly$ of degree $\deg(\genpoly) = \degree$ we denote $\rank{\genpoly} \definedas \drank{\degree}{\genpoly}$.
    \newline
    We call such function $\Gamma$ a \emph{decomposition} or a \emph{computation} of $\genpoly$ using lower-degree polynomials.
\end{definition}
Let us now define a factor.
Note that we focus our discussion to factors in $\field$,
but define the basic definitions over a general set $U$ so they will apply for factors over a general sets.
This is necessary as we will later use them also for other sets such as $\variety \subseteq \field$.
\begin{definition}[Factor]
    Let $U$ be a set.
    A \emph{factor over $U$}, denoted by $\factor$, is simply a partition of the set $U$.
    Each subset in the partition is called an \emph{atom}.
    \newline
    A collection of function $\funcdef{\genfunc_1,...,\genfunc_c}{U}{\basefield}$ defines a factor $\factor_{\genfunc_1,...,\genfunc_c}$ over $U$
    with atoms:
    \[
        \set{u \in U \suchthat f_1(u)=b_1, ..., f_c(u)=b_c}
    \]
    for all $b_1,...,b_c \in \basefield$.
    \newline
    Additionally, we use $\factor$ to also denote the map $\factor(u) \rightarrow (\genfunc_1(u),...,\genfunc_c(u))$.
\end{definition}

\begin{notation*}
    Let $\funcdef{\genfunc_1,...,\genfunc_c}{U}{\basefield}$ be a collection of functions.
    For a factor $\factor \definedas \factor_{\genfunc_1,...,\genfunc_c}$, we denote by $\abs{\factor}$ the amount of functions that define it, i.e. $\abs{\factor} = c$.
    Moreover, we denote $\norm{\factor} \definedas \abs{\basefield}^c$, which is the maximal amount of (possibly empty) atoms.
    Additionally, the rank of the factor is defined to be the rank of the polynomials that generate it.
\end{notation*}
\begin{definition}[Polynomial Factor]
    We say a factor $\factor$ over $\field$ is a \emph{polynomial factor} if it is defined by a collection of polynomials $\funcdef{\genpoly_1,...,\genpoly_c}{\field}{\basefield}$,
    i.e. $\factor = \factor_{\genpoly_1,...,\genpoly_c}$.
    The degree of the factor, denote as $\deg(\factor)$ is the maximal degree of the polynomials $\genpoly_1,...,\genpoly_c$.
\end{definition}
Note that the notion of degree (and polynomial) are defined only for functions over $\field$,
therefore this definition is well-defined only for $U = \field$.

\begin{definition}[Rank of a Factor]\label{definition:factor-rank}
    Let $\genpolyset$ be a collection of polynomials $\funcdef{\genpoly_1,...,\genpoly_c}{\field}{\basefield}$.
    The rank of the polynomial collection is defined as:
    \[
        \rank{\genpolyset} \definedas \min \set{\drank{\degree}{\sum_{i=1}^c{\lambda_i \genpoly_i}} \suchthat 0 \neq \vec{\lambda} \in \basefield^c, \degree = \max_{i\in\sparens{c}}{\deg(\lambda_i \genpoly_i)}}
    \]
    For a factor $\factor$ defined by a collection of polynomials $\genpolyset$, we define its rank to be the rank of the collection of polynomials defining it.
    For a non-decreasing function $\funcdef{\rankfunc}{\naturalnumbersset}{\naturalnumbersset}$, a factor $\factor$ is called $\rankfunc$-regular if its rank is at least $\rankfunc(\abs{\factor})$.
\end{definition}
\begin{note*}
    Note that in the definition above, the rank of each linear combination is calculated as the $\degree$-rank,
    where $\degree$ is the maximal degree of a polynomial that participates in the linear combination non-trivially.
    This is crucial as it ensures that a high rank factor do not have linear dependence in the largest-degree homogenous component of any of its polynomials.
\end{note*}

We now present a fundamental property of high rank polynomials, that was first proved by~\cite{green2007distribution} when $\degree < \abs{\basefield}$,
later extended to general fields by~\cite{kaufman2008worst}, and further extended also to large fields by~\cite{DBLP:journals/corr/0001L15}.
This property of high rank polynomials is that they they have low bias:
\begin{theorem}[Rank-bias in $\field$]\label{high-rank-implies-low-bias}
    Let $\basefield$ be a finite field.
    Let $\epsilon > 0$ and $\degree \in \naturalnumbersset$.
    There exists $r_{\ref{high-rank-implies-low-bias}} \definedas r_{\ref{high-rank-implies-low-bias}}(\basefield, \degree, \epsilon)$,
    such that for every degree-$\degree$ polynomial $\funcdef{\genpoly}{\field}{\basefield}$:
    if $\rank{\genpoly} \geq r_{\ref{high-rank-implies-low-bias}}$ then $\bias{\genpoly} < \epsilon$.
\end{theorem}
\begin{remark}
    This property implies that a collection of polynomials that have high rank is \emph{equidistributed}.
    See Appendix~\ref{sec:equidistribution-of-functions} for more details in this regard.
\end{remark}

%
%
%

\subsection[Regularization in \titlefield]{Regularization in \titlefield}\label{subsec:regularization-in-Fn}
In this subsection we define the regularization process in $\field$.
Before doing so, let us present some definitions in this regard.
Note that we define the basic definitions over a general set $U$ so they will apply for factors over a general sets,
as this is necessary as we will later use them also for other sets such as $\variety \subseteq \field$.
\begin{definition}[Measureable]
    Let $U$ be a set, and let $A \subseteq U$.
    Let $\genfuncset = \set{\genfunc_1,...,\genfunc_c}$ be a collection of functions $\funcdef{\genfunc_i}{U}{\basefield}$.
    We say a function $\funcdef{\genfunc[2]}{U}{\basefield}$ is \emph{measurable in respect of $\genfuncset$ in $A$}, shorthand by \emph{$\genfuncset$-measurable in $A$},
    if there exists a function $\funcdef{\Gamma}{\basefield^c}{\basefield}$ such that:
    \[
        \forall a \in A: g(a) = \Gamma(\genfunc_1(a),...,\genfunc_c(a))
    \]
    When discussing the factor over $A$ defined by $\factor = \factor_{\genfunc_1,...,\genfunc_c}$,
    we also say $\genfunc[2]$ is \emph{measurable in resepct of $\factor$}.
    The function $\Gamma$ will be denoted as the \emph{measurement function} of $\genfunc[2]$ in respect of $\genfuncset$.
    Additionally, when $A = U$, we sometimes omit the specification of the domain, and say $\genfunc[2]$ is measurable in respect of $\genfuncset$.
    \newline
    Note that in this paper, we usually think of $U = \field$, and $A$ is either $\field$ or $\variety \subseteq \field$.
\end{definition}
\begin{remark}
    If $\genfunc[2]$ is $\genfuncset$-measurable in $A$, then every value of $\genfunc[2]$ in $A$
    can be determined by the values of $\genfunc_1,...,\genfunc_c$.
    In other words, the function $\genfunc[2]$ is constant inside every atom of the factor defined by $\genfuncset$.
\end{remark}

\begin{definition}[Syntactic Refinement]
    Let $\factor$ and $\factor^\prime$ be polynomial factors over $U$.
    We say a factor $\factor^\prime$ is a \emph{syntactic refinement} of the factor $\factor$, if the collection of functions defining $\factor$ is a subset of the set
    of functions defining $\factor^\prime$.
    We denote this property of $\factor^\prime$ by $\factor^\prime \synrefine \factor$.
\end{definition}
We now present a standard generalized definition of refinement, where we only require the atoms induced by the refined factors are sub-atoms of those that are induced by the original factor.
Note that in this refinement, we allow the refined factor to include completely different polynomials than the original factor.
\begin{definition}[Semantic Refinement]\label{def:semantic-refinement}
    Let $\factor$ and $\factor^\prime$ be polynomial factors on $U$ defined by $\genpolyset$ and $\genpolyset^{\prime}$ respectively.
    We say the factor $\factor^\prime$ is a \emph{semantic refinement} of the factor $\factor$ in $A \subseteq U$,
    if $x, y \in A$ with $\factor^\prime(x) = \factor^\prime(y)$ implies that $\factor(x)=\factor(y)$.
    We denote this property of $\factor^\prime$ by $\factor^\prime \semrefineex{A} \factor$.
    When $A = U$, we sometimes omit $A$ from the syntax and denote it with $\factor^{\prime} \semrefine \factor$
    \newline
    Note that $\factor^\prime \synrefine \factor$ implies $\factor^\prime \semrefineex{A} \factor$ for every $A \subseteq U$.
\end{definition}
\begin{remark}
    A handy property of semantic refinement is that if $\funcdef{\genfunc}{A}{\basefield}$ is $\genpolyset$-measurable,
    then it is also $\genpolyset^{\prime}$-measurable in $A$.
    Moreover, the other direction is also true:
    If every $\genpolyset$-measurable function $\funcdef{\genfunc}{A}{\basefield}$ in $A$ is also $\genpolyset^{\prime}$-measurable in $A$, then $\factor^{\prime} \semrefineex{A} \factor$.
\end{remark}

Next, we recall a lemma that was presented in~\cite[Theorem 4.1]{bhattacharyya2013locally}, that allows us, given a polynomial that is measurable by a high rank factor in $\field$,
to replace the polynomials in the measurement function to any collection of polynomials with smaller or equal degree, and preserve the degree of the original polynomial.
Note that we state the lemma under the constraint that $\degree < \abs{\basefield}$, but it is also valid for when $\degree \geq \abs{\basefield}$ with proper generalization of definitions to this case (See~\cite[Theorem 4.1]{bhattacharyya2013locally} for the exact statement).
\begin{lemma}[Preserving Degree in $\field$]\label{preserving-degree-starting-field}
Let $\degree>0$ an integer such that $\degree < \abs{\basefield}$, and let $\funcdef{\genpoly_1,...,\genpoly_c}{\field}{\basefield}$ be polynomials of degree at most $\degree$, that form a factor of rank $\geq  \rankfunc^{\ref{preserving-degree-starting-field}}(\basefield, \degree, c)$.
Assume that for $\funcdef{\Gamma}{\basefield^{c}}{\basefield}$, the function $\funcdef{\gamma}{\field}{\basefield}$ defined as $\gamma(a) \definedas \Gamma(\genpoly_1(a),...,\genpoly_c(a))$ is of $\deg(\gamma)=\degree^{\prime}$.
\newline
Then, for every collection of polynomials $\funcdef{\genpoly[2]_1,...,\genpoly[2]_c}{\field}{\basefield}$ that satisfy $\deg(\genpoly[2]_i) \leq \deg(\genpoly_{i})$,
the function $\gamma^{\prime}$ defined as $\gamma^{\prime}(a)=\Gamma(\genpoly[2]_1(a),...,\genpoly[2]_c(a))$ is a polynomial of $\deg(\gamma^{\prime})\leq\degree^{\prime}$.
\end{lemma}
Next, we restate a useful lemma from~\cite[Lemma 4.17]{DBLP:journals/corr/0001L15} that shows that under the conditions
above, $\Gamma$ is as a low-degree polynomial (with even stronger conditions).
Formally, they showed:
\begin{lemma}[Faithful Composition]
    In the case discussed above, the structure of $\Gamma$ is as follows:
    \[
        \Gamma(z_1,...,z_{c_1})  =
        \sum_{\alpha \in \sparens{\basefieldsize - 1}^{c}} {{C_{\alpha}} \cdot {\prod_{i = 1}^{c}}{z_i^{\alpha_i}}}
    \]
    where $C_{\alpha} = 0$ whenever $\sum_{i = 1}^{c_1}(\alpha_i \cdot \deg(\genpoly_i)) > \degree^\prime$.
    \newline
    In other words, this means that $\Gamma$ as a function $\funcdef{\Gamma}{\basefield^{c}}{\basefield}$,
    is a polynomial of degree $\leq \degree^{\prime}$, even when substituting its $i$-th input by any polynomial of degree $\leq \deg(\genpoly_i)$.
\end{lemma}

Finally, we restate the regularization process, that was first presented by~\cite[Lemma 2.3]{green2007distribution}.
The regularization process shows that every factor have a high-rank factor that semantically refines it,
without increasing the size of the factor too much (its new size is independent of $\blocklength$).
\begin{lemma}[Regularization in $\field$]\label{regularization-in-Fn-lemma}
    Let $\funcdef{\rankfunc}{\naturalnumbersset}{\naturalnumbersset}$ be a non-decreasing function and let $\degree \in \naturalnumbersset$.
    There exists $\funcdef{C_{\rankfunc, \degree}^{\ref{regularization-in-Fn-lemma}}}{\naturalnumbersset}{\naturalnumbersset}$ such that the following holds:
    Let $\factor$ be a factor on $\field$ defined by polynomials $\genpolyset = (\genpoly_1,...,\genpoly_c)$ where for all $i \in [c]$: $\funcdef{\genpoly_i}{\field}{\basefield}$ and $\deg(\genpoly_i) \leq \degree$.
    Then, there is an $\rankfunc$-regular factor $\factor^\prime$ defined by polynomials $\genpolyset[2] = (\genpoly[2]_1,...,\genpoly[2]_{c^\prime})$ where
    for all $i \in [c]$: $\funcdef{\genpoly[2]_i}{\field}{\basefield}$ and $\deg(\genpoly[2]_i) \leq \degree$ such that
    $\factor^\prime \semrefine \factor$ and $c^\prime \leq C_{\rankfunc, \degree}^{\ref{regularization-in-Fn-lemma}}(c)$.
    \newline
    Moreover, if $\factor \synrefine \bar{\factor}$ for some polynomial factor $\bar{\factor}$ with rank at least $\rankfunc(c^\prime)+c^\prime+1$,
    then we can require that $\factor^\prime \synrefine \bar{\factor}$.
\end{lemma}

%% file: polynomials_in_X.tex
\section[Polynomials in \titlevariety]{Polynomials in \titlevariety}\label{sec:polynomials_in_X}
In this section we wish to generalize the definition of degree-$\degree$ polynomials for functions $\funcdef{\onvarfunc}{\variety}{\basefield}$.
Note that we wish to define it using a property of $\onvarfunc$ that is intrinsic to $\variety$: given a function $\funcdef{\onvarfunc}{\variety}{\basefield}$,
we wish be able to determine its degree only using values of $\variety$, without considering any value outside of $\variety$ (such as values of $\field \setminus \variety$).
\newline
To define such property, we generalize the local definition of a degree that is defined for polynomials in $\field$.
We remind the reader that in $\field$, we said a function over $\field$ is a polynomial of degree $\leq \degree$ if and only if its $(\degree+1)$-derivative in every direction is $\equiv 0$.
Thus, in order to determine the $(\degree+1)$-derivative of a function in directions $y_1,...,y_{\degree+1}$, one needs to evaluate the function over all the points of the cube generated by $x,y_1,...,y_{\degree+1}$,
which is the set of points $\set{x + \sum_{i \in S}{y_i}}_{S \subseteq [\degree+1]}$.
This raises a challnge in extending this definition for functions defined over $\variety \subseteq \field$:
depending on $\variety$, the function $\funcdef{\onvarfunc}{\variety}{\basefield}$ is not be defined to all points in all the cubes of $\field$, because some of those points do not lie in $\variety$.
\newline
Therefore, to generalize the definition of a polynomial to $\variety$, we start by giving the formal definition and notation of the set of cubes in $\variety$:
\begin{definition}[Cubes]
    Let $k \in \naturalnumbersset$ be an integer and let $x, y_1,...,y_{k} \in \field$.
    We define the cube $\cube{x}{y_1,...y_{k}}$ as follows:
    \[
        \cube{x}{y_1,...,y_{k}} \definedas \set{x + \sum_{i \in S}{y_i}}_{S \subseteq [k]}
    \]
    We refer to $x$ as the \emph{offset} of the cube, and $y_1,...,y_{k}$ as the \emph{directions} of the cube.
    \newline
    Moreover, Let $\variety \subseteq \field$ be a subset.
    We define the \emph{set of cubes of $\variety$ of size $k$} as follows:
    \[
        \cubes{k}{\variety} \definedas
            \set{\cube{x}{y_1,...,y_{k}} \suchthat \forall S \subseteq[k]: (x + \sum_{i \in S}{y_i}) \in \variety}
    \]
\end{definition}

Using this definition, we can define a polynomial of degree $\leq \degree$ for subsets of $\field$:
\begin{definition}[Polynomials in $\variety$]
    Let $\degree \in \naturalnumbersset$ be an integer, and let $\variety \subseteq \field$.
    We say the degree of a function $\funcdef{\onvarfunc}{\variety}{\basefield}$ is $\degree$
    if $\degree$ is the smallest integer such that $\onvarfunc$ vanishes over all cubes of size $(\degree+1)$, i.e:
    \[
        \forall {\cube{x}{y_1,...,y_{\degree + 1}} \in \cubes{\degree+1}{\variety}}:
            \directionderivative_{y_{\degree + 1}}...\directionderivative_{y_1} \onvarpoly(x) = 0
    \]
    A function over $\variety$ of degree $\leq \degree$ is also called a \emph{polynomial of degree $\leq \degree$}.
    We denote the set of polynomials of degree $\leq \degree$ over $\variety$ by $\allpolyset{\leq \degree}{\variety}{\basefield}$.
\end{definition}
\begin{note*}
    For $\variety = \field$, the definition above coincides with the local definition of polynomials.
\end{note*}

\subsection[Lifting Polynomials]{Lifting Polynomials}\label{subsec:lifting-polynomials}
Our goal to achieve good properties for polynomials over $\variety$.
To do so, we wish to connect the desired properties of polynomials defined over $\variety$, to properties known for polynomials over $\field$.
Following such strategy raises a question: given a polynomial $\funcdef{\onvarpoly}{\variety}{\basefield}$, which polynomial over $\field$ should we consider to deduce properties of $\onvarpoly$?
To find such a polynomial over $\field$, it would have been useful that all polynomials over $\variety$ actually "came from" polynomials over $\field$.
More formally, it would have been useful that all polynomials $\funcdef{\onvarpoly}{\variety}{\basefield}$ would be equal to a restriction of some polynomial $\funcdef{\genpoly}{\field}{\basefield}$ of degree $\leq \degree$, to the set $\variety$.
This would give us a "good candidate" (or candidates) to polynomials over $\field$, that using their known properties, we could achieve the properties we desire for polynomials over $\variety$.
\newline
Generally speaking, the existence of such polynomial $\funcdef{\genpoly}{\field}{\basefield}$ is not trivial by itself, and it mapy depend on the polynomial $\onvarpoly$ and the set $\variety$.
In this subsection, we discuss sets $\variety \subseteq \field$ that have this property for every polynomial $\funcdef{\onvarpoly}{\variety}{\basefield}$.
Before formulating the notion above, we start by a simple remark:
\begin{remark}
    By the local criteria for $\field$, we have that a restriction of a polynomial of degree $\leq \degree$ over $\field$ to $\variety$ is a polynomial of degree $\leq \degree $ over $\variety$.
    Therefore, the other direction is true: every restriction of a polynomial over $\field$ to $\variety$ is a polynomial over $\variety$.
\end{remark}
%

Next, let us define subsets $\variety \subseteq \field$ that have the desired property, which we call \emph{$\degree$-lift-enabler variety}.
\begin{definition}[$\degree$-lift-enabler Subset]
    Let $\basefield$ be a field, and $\blocklength>0$ be an integer.
    For an integer $\degree > 0$, we say a subset $\variety\subseteq\field$ is \emph{$\degree$-lift-enabler} if for every $\degree^{\prime} \leq \degree$,
    for every polynomial $\onvarpoly \in \allpolyset{\degree^{\prime}}{\variety}{\basefield}$
    there exist a polynomial $\lift{\onvarpoly} \in \allpolyset{\degree^{\prime}}{\field}{\basefield}$ such that $\restrictfunc{\onvarpoly}{\variety}=\restrictfunc{\lift{\onvarpoly}}{\variety}$.
\end{definition}
\begin{remark}
    Using the local criterion of polynomials and the fact that that $\cubes{\degree+1}{\variety} \subseteq \cubes{\degree+1}{\field}$,
    one can see that for a polynomial $\funcdef{\onvarpoly}{\variety}{\basefield}$ with $\deg(\onvarpoly) = \degree$,
    every extension $\funcdef{\genpoly}{\field}{\basefield}$ with $\onvarpoly = \restrictfunc{\genpoly}{\variety}$
    holds the bound $\deg(\lift{\onvarpoly}) \geq \degree$.
    The other direction is not true in the general case, but it is specifically promised when the variety is $\degree$-lift-enabler.
\end{remark}

This definition naturally raises the following definition:
\begin{definition}[The Lift Operator]
    Let $\degree \in \naturalnumbersset$ be an integer.
    Let $\variety \subseteq \field$ be a $\degree$-lift-enabler subset.
    We define \emph{the $\degree$-lift operator} to be an operator $\funcdef{\lift{\square}}{\allpolyset{\leq \degree}{\variety}{\basefield}}{\allpolyset{\leq \degree}{\field}{\basefield}}$ the following way:
    \newline
    Let $\degree^\prime \leq \degree$.
    Given a polynomial $\funcdef{\onvarpoly}{\variety}{\basefield}$ of degree $\degree^\prime$, the operator $\lift{\square}$ returns a polynomial $\funcdef{\lift{\onvarpoly}}{\field}{\basefield}$ of degree $\degree^\prime$
    such that $\onvarpoly = \restrictfunc{\lift{\onvarpoly}}{\variety}$.
    Note that we did not require the lift to be unique.
    Thus, in case there are multiple valid lifts for a polynomial $\onvarpoly \in \allpolyset{\leq \degree}{\variety}{\basefield}$, the lift operator picks a single (consistent) one of them.
    Moreover, the lift always exists because the subset $\variety$ is $\degree$-lift-enabler.
    \newline
    In addition, for a collection $\onvarpolyset = (\onvarpoly_1,...,\onvarpoly_c)$ of polynomials $\onvarpoly_i \in \allpolyset{\leq \degree}{\variety}{\basefield}$,
    we denote $\lift{\onvarpolyset} \definedas (\lift{\onvarpoly_1},...,\lift{\onvarpoly_c})$
\end{definition}

In the following subsections, we give example to two concrete sets $\variety \subseteq \field$ that are $\degree$-lift-enablers.
Before doing so, we define an algebraic variety:
\begin{definition}[Algebraic Variety]
    For a collection of functions $\genfuncset \definedas \set{\genfunc_1,...\genfunc_c}$ such that $\funcdef{\genfunc_i}{\field}{\basefield}$,
    we denote $\zerofunc{\genfuncset} \definedas \set{x \in \field \suchthat \forall i: \genfunc_i(x) = 0}$.
    \newline
    If the collection is a collection of polynomials, we call $\zerofunc{\genfuncset}$ an \emph{algebraic variety}, shorthand by \emph{variety}.
    \newline
    The degree of the variety is defined to be the maximal degree of polynomials in the collection that defines it.
    The rank of the variety is defined to be the rank of the collection that defines the variety (as a collection).
\end{definition}

\subsection[High Rank Varieties of High Minimal Degree]{High Rank Varieities of High Minimal Degree}\label{subsec:high-rank-varieities-of-high-degree}
We now present a theorem proved in \cite[Corollary 1.10]{kazhdan2018polynomial}, that shows that high rank varieties are $\degree$-lift-enabler when the polynomials defining the variety are of degree $>\degree$:
\begin{theorem}\label{subsec:high-rank-varities-are-d-lift-enabler}
    Let $\basefield$ be a finite field, and let $\varietydeg$, $\varietypolycount>0$ representing parameters of a variety.
    Let $\degree < \varietydeg$ a positive integer representing a degree of a polynomial which we wish to lift.
    There exists $\varietyrankval=\varietyrankval(\basefield,\varietydeg,\varietypolycount)>0$ such that for for all $\blocklength \in \naturalnumbersset$, any variety $\variety = \zerofunc{\varpolyset} \subseteq \field$ for $\varpolyset = (\varpoly_1,...,\varpoly_{\varietypolycount})$
    such that $\rank{\varpolyset} > \varietyrankval$, degree $\deg(\varpolyset) = \varietydeg$, with all defining polynomials of degree $\deg(\varpoly_i)> \degree$,
    it holds that $\variety$ is a $\degree$-lift-enabler subset.
\end{theorem}
\begin{remark}
    Under the conditions stated above, it was proved in ~\cite{kazhdan2018polynomial} that the lift is in fact \emph{unique}.
    Formally, if $\funcdef{\onvarpoly}{\variety}{\basefield}$ is a polynomial of degree $\leq \degree$,
    then there exists a \emph{unique} polynomial $\funcdef{\genpoly }{\field}{\basefield}$ such that $\restrictfunc{\genpoly}{\variety} \equiv \onvarpoly$.
\end{remark}

\subsection[High Rank Varieties on a Large Field]{High Rank Varieties on a Large Field}\label{subsec:high-rank-varieties-on-a-large-field}
In this subsection, we recall a theorem proved by~\cite[Theorem 1.7]{kazhdan2019extendingweaklypolynomialfunctions} regarding high rank varieties that are defined on "large" fields.
We note that the fields are large in respect of the degree $\degree$ one wish to lift, but still does not depend on $\blocklength$.

Next we define a weakly polynomial, which generalizes our definition of a polynomial in $\variety$, that was used in~\cite[Definition 1.1]{kazhdan2019extendingweaklypolynomialfunctions}:
\begin{definition}
    Let $\variety \subseteq \field$ be a set.
    We say a function $\funcdef{\genfunc}{\variety}{\basefield}$ is a \emph{weakly polynomial of degree $\leq \degree$}
    if for any affine subspace $L \subseteq \variety$, the restriction $\restrictfunc{\genfunc}{L}$ is a polynomial of degree $\leq \degree$.
\end{definition}
\begin{remark}
    By the local criteria of a polynomial, it is easy to see that every $\genpoly \in \allpolyset{\leq \degree}{\variety}{\basefield}$ is a weakly polynomial of degree $\leq \degree$.
\end{remark}
And now, we can present the lifting theorem for large fields, as proved in~\cite[Theorem 2.17]{kazhdan2020propertieshighranksubvarieties}.
\begin{theorem}\label{high-rank-varieties-over-large-fields-are-d-lift-enablers}\cite[Theorem 2.17]{kazhdan2020propertieshighranksubvarieties}
    Let $\degree, \varietydeg \in \naturalnumbersset$,
    and let $\basefield$ be a finite field such that $\abs{\basefield} > \degree \cdot \varietydeg$.
    There exists $\rankfunc_{\ref{high-rank-varieties-over-large-fields-are-d-lift-enablers}} = \rankfunc_{\ref{high-rank-varieties-over-large-fields-are-d-lift-enablers}}(\varietydeg, \degree)$
    such that for any variety $\variety \subseteq \field$ of maximal degree $\leq \varietydeg$ and rank $\geq \rankfunc_{\ref{high-rank-varieties-over-large-fields-are-d-lift-enablers}}$,
    have the following property:
    Every weakly polynomial function $\funcdef{\onvarpoly }{\variety}{\basefield}$ of degree $\leq \degree$
    can be lifted to a polynomial function $\funcdef{\genpoly}{\field}{\basefield}$ of degree $\leq \degree$.
\end{theorem}
\begin{note*}
    Note that we stated the theorem above to finite fields, but it is also valid for infinite algebraically closed fields.
\end{note*}
The theorem above implies the following corollary:
\begin{corollary}
    Let $\degree, \varietydeg \in \naturalnumbersset$,
    and let $\basefield$ be a finite field such that $\abs{\basefield} > \degree \cdot \varietydeg$.
    There exists $\rankfunc_{\ref{high-rank-varieties-over-large-fields-are-d-lift-enablers}} = \rankfunc_{\ref{high-rank-varieties-over-large-fields-are-d-lift-enablers}}(\varietydeg, \degree)$
    such that for any variety $\variety \subseteq \field$ of degree $\varietydeg$ and rank $\geq \rankfunc_{\ref{high-rank-varieties-over-large-fields-are-d-lift-enablers}}$
    is a $\degree$-lift-enabler.
\end{corollary}

%% file: relative_rank_relative_bias_property.tex
\section[Relative Rank-Bias Property]{Relative Rank-Bias Property}\label{sec:relative-rank-bias-property}
In this section, we generalize the relation between rank and bias that is known for $\field$ also for $\variety \subseteq \field$.
Specifically, in Lemma~\ref{high-rank-implies-low-bias}, it was shown that high rank factors have low bias in $\field$.
We wish to define an alternative definition of rank for $\variety \subseteq \field$, called \emph{$\variety$-relative rank},
such that high $\variety$-relative rank implies low bias in $\variety$.
This type of relation (and definition) was shown previously to a few sets;
in~\cite[Theorem 1.8]{lampert2021relative} for sets $\variety = \zerofunc{\genpolyset[2]}$ where $\genpolyset[2]$ is a collection of polynomials of high rank;
and in~\cite[Theorem 1.4]{gowers2022equidistributionhighrankpolynomialsvariables} for sets $\variety = S^\blocklength$ for $S \subset \basefield$.

To understand this notion, we first introduce a simple example that demonstrates the need for a different definition of rank to achieve equidistribution properties in subsets of $\field$.
\begin{example}
    Let $\variety = \set{x \in \field \suchthat x_1 = 0}$.
    Define $\funcdef{\genpoly}{\field}{\basefield}$ by $\genpoly(x) \definedas x_1$.
    \newline
    In $\field$, $\genpoly$ has rank $\infty$ as it can not be decomposed polynomials of degree $< 1$ (constants).
    Additionally, it is perfectly equidistributed.
    This is the simplest example of the rank-bias relation in $\field$.
    \newline
    However, when restricting $\genpoly$ to $\variety$, we get $\restrictfunc{\genpoly}{\variety} \equiv 0$.
    As $0$ is a constant function, it is the \emph{least} equidistributed possible in $\variety$.
    Therefore, we see that the way we defined rank in $\field$ does not imply the desired equidistribution in $\variety$:
    we found a polynomial with high rank (infinity) that has a very high bias \emph{in $\variety$} (the maximal).
\end{example}
The reason the definition of rank in $\field$ fails to capture equidistribution even on subsets that are really similar to $\field$ (isomorphic to $\basefield^k$), is because of the following reason:
Even though our polynomial $\genpoly$ does \emph{not} have a decomposition to a few lower-degree polynomials by itself,
there \emph{exists} a \emph{$\variety$-equivalent} polynomial that has such structured decomposition.
Here, by $\variety$-equivalent we mean a polynomial in $\field$ that is bounded by the same degree bound, and is equal to $\genpoly$ in $\variety$.
In the example described above, this equivalent polynomial is the constant function $0$, and its decomposition is the trivial one (any function decomposes a constant function).
An alternative perspective which we use throughout this paper to $\variety$-equivalence is that both polynomials are equal up to a \emph{valid $\variety$-remainder}: a bounded degree polynomial that is $\equiv 0$ in $\variety$.

Generally speaking, high $\variety$-relative rank may not imply low bias in $\variety$.
Therefore, this structure-randomness relation is not true for a general subset $\variety \subseteq \field$,
but is a \emph{property} of the subset $\variety$.
Thus, we say that a subset has the \emph{relative rank-bias property} if this relation holds,
i.e. if high $\variety$-relative rank implies equidistribution in $\variety$.

Let us now formally define our definition for relative rank, inspired by the two different definitions of relative rank presented in~\cite[Definition 1.6]{lampert2021relative}
and in~\cite[Definition 1.3]{gowers2022equidistributionhighrankpolynomialsvariables}:
\begin{definition}[Relative Rank of a Polynomial]\label{def:relative-rank-of-polynomial}
    Let $\variety \subseteq \field$ and let $\degree \in \naturalnumbersset$.
    For an integer $\degree \in \naturalnumbersset$ and a polynomial $\funcdef{\genpoly}{\field}{\basefield}$, we define its \emph{$\degree$-relative rank} in respect of $\variety$ as:
    \[
        \drelrank{\degree}{\variety}{\genpoly} \definedas \min \set{\drank{\degree}{\genpoly - \relativeremainder{\genpoly}} \suchthat
        \relativeremainder{\genpoly} \in \allpolyset{\leq \deg(\genpoly)}{\field}{\basefield}, \restrictfunc{\relativeremainder{\genpoly}}{\variety} \equiv 0}
    \]
    For a polynomial $\genpoly$ of degree $\deg(\genpoly) = \degree$ we denote $\relrank{\variety}{\genpoly} \definedas \drelrank{\degree}{\variety}{\genpoly}$.
\end{definition}
\begin{definition}[$\variety$-equivalent and $\variety$-remainder]
    Moreover, we say a polynomial is \emph{$\variety$-equivalent} to $\genpoly$ if its restriction to $\variety$ is $\equiv \restrictfunc{\genpoly}{\variety}$.
    We say it is \emph{valid $\variety$-equivalent to $\genpoly$} if it is $\variety$-equivalent to $\genpoly$ and its of the \emph{same} degree of $\genpoly$.
    \newline
    Similarly, we say a polynomial is \emph{$\variety$-remainder} of $\genpoly$ if its restriction to $\variety$ is \emph{$\equiv 0$}.
    We say it is \emph{valid $\variety$-remainder} if it is $\variety$-remainder of $\genpoly$ and its of degree \emph{smaller or equal} of the degree of $\genpoly$.
    \newline
    We typically denote such polynomial as $\relativeremainder{\genpoly}$.
\end{definition}
In other words, the $\degree$-relative rank of a polynomial $\genpoly$ is the smallest $\degree$-rank of all valid $\variety$-equivalents of $\genpoly$.

\begin{note}\label{note:comparison-to-gowers-rank}
    Note that~\cite[Definition 1.3]{gowers2022equidistributionhighrankpolynomialsvariables} defines rank in a substantially different way than our definition,
    and consequentially our results will not apply to the sets they presented.
    One of the main differences in the definition of rank occurs for $\degree = 1$.
    In the definition we use for rank, the rank of every (non-constant) degree-$1$ polynomial is $\infty$, where in the definition used in~\cite{gowers2022equidistributionhighrankpolynomialsvariables} it is a finite number (which is possibly very small).
    This difference is crucial, as for example, it makes regularization according to their definition not-trivially possible,
    where it is known to be possible when rank is defined by the definition we use (Lemma~\ref{regularization-in-Fn-lemma}).
\end{note}
\begin{definition}[Relative Rank of a Factor]
    Let $\variety \subseteq \field$.
    Let $\genpolyset$ be a set of polynomials $\genpolyset = \set{\genpoly_1,...,\genpoly_c}$.
    The rank of the polynomial set $\genpolyset$ relative to the subset $\variety$ is defined as:
    \[
        \relrank{\variety}{\genpolyset} \definedas \min \set{\drelrank{\degree}{\variety}{\sum_{i=1}^c{\lambda_i \genpoly_i}} \suchthat 0 \neq \vec{\lambda} \in \basefield^c, \degree = \max_{i \in \sparens{c}}{\deg(\lambda_i \genpoly_i)}}
    \]
    For a factor $\factor$ defined by a collection of polynomials, we define its relative rank relative to $\variety$ to be the relative rank of the collection of polynomials defining it, relative to the set $\variety$.
    \newline
    For a non-decreasing function $\funcdef{\rankfunc}{\naturalnumbersset}{\naturalnumbersset}$, a factor $\factor$ is called $\rankfunc$-$\variety$-regular if its relative rank in respect to $\variety$ is at least $\rankfunc(\abs{\factor})$.
\end{definition}

\subsection[Relative Rank-Bias Property]{Relative Rank-Bias Property}
\begin{definition}[Relative Rank-Bias property]
    Let $\basefield$ be a finite field, and let $\degree \in \naturalnumbersset$ be an integer.
    Let $\funcdef{\rankbiasfunc}{\realnumbersset^{+}}{\naturalnumbersset}$ be a function that represents the rank-bias relation for a fixed $\degree, \basefield$.
    \newline
    We say a set $\variety \subseteq \field$ has the \emph{$(\rankbiasfunc, \basefield, \degree)$-relative rank-bias property} if
    for every $\epsilon > 0$,
    for every polynomial $\genpoly$ of degree $\leq \degree$ with $\relrank{\variety}{\genpoly} \geq \rankbiasfunc(\epsilon)$ we have:
    \[
        \relbias{\variety}{\genpoly} < \epsilon
    \]
\end{definition}

As an immediate corollary of Lemma~\ref{high-rank-implies-low-bias} that shows that high rank implies low bias, we have that $\variety = \field$ has the relative rank-bias property.
\begin{corollary}[$\field$ has the relative rank-bias property]
    For every finite field $\basefield$ and $\degree \in \naturalnumbersset$, let $\funcdef{\rankbiasfunc}{\realnumbersset^{+}}{\naturalnumbersset}$ defined as $\rankbiasfunc(\epsilon) \definedas \rankval_{\ref{high-rank-implies-low-bias}}(\basefield, \degree, \epsilon)$.
    Then, we have that the set $\variety = \field$ has the $(\rankbiasfunc, \basefield, \degree)$-relative rank-bias property.
\end{corollary}
\begin{proof}
    This is a simple usage of Lemma~\ref{high-rank-implies-low-bias}:
    Note that when $\variety = \field$, we have that $\relrank{\variety}{\genpoly} = \rank{\genpoly}$.
    Now, if $\genpoly$ is a polynomial of degree $\leq \degree$ and $\relrank{\variety}{\genpoly} = \rank{\genpoly} \geq \rankbiasfunc(\epsilon) = \rankbiasfunc_{\ref{high-rank-implies-low-bias}}(\basefield, \degree, \epsilon)$,
    then:
    \[
        \relbias{x \in \field}{\genpoly(x)} < \epsilon
    \]
\end{proof}

\subsection[Limited Relative Rank-Bias Property]{Limited-Relative Rank-Bias Property}\label{subsec:limited-relative-rank-bias-property}
Sometimes, however, we can not request $\variety$ to be such that high relative rank implies low bias for every $\epsilon > 0$,
but only for $\epsilon^\prime \geq \epsilon$ for some constant $\epsilon > 0$.
This leads to defining the \emph{limited relative rank-bias property}, which will be used to discuss such sets $\variety \subseteq \field$.
\newline
As we will later see, this definition raises naturally where $\variety$ is a high rank variety,
in which for the relative rank-bias property to hold for some $\epsilon > 0$, the rank of the variety should be greater than a value that is dependent of $\epsilon$.
Thus, to have the relative rank-bias property for a high rank variety but without requiring an infinitely large rank, we must limit the relative rank-bias property for $\epsilon^\prime \geq \epsilon$
We formulate the definition of this property as follows:
\begin{definition}[Limited Relative Rank-bias property]
    Let $\basefield$ be a finite field, let $\degree \in \naturalnumbersset$ be an integer, and let $\epsilon > 0$ be a constant.
    Let $\funcdef{\rankbiasfunc}{[\epsilon, \infty]}{\naturalnumbersset}$ be a function that represents the limited-relative-rank-bias relation.
    \newline
    We say a set $\variety \subseteq \field$ has the \emph{$(\rankbiasfunc, \basefield, \degree, \epsilon)$-limited-relative-rank-bias property} if
    for every $\epsilon^\prime \geq \epsilon$,
    for every polynomial $\genpoly$ of degree $\leq \degree$ with $\relrank{\variety}{\genpoly} \geq \rankbiasfunc(\epsilon^\prime)$ we have:
    \[
        \relbias{\variety}{\genpoly} < \epsilon^{\prime}
    \]
    As a convention, we denote by $\epsilonlimitedrankbias$ the $\epsilon$ such that the limited-relative-rank-bias property holds for $\variety$.
\end{definition}

\subsubsection[High Rank Varieities]{High Rank Varieities}
In this subsection, we are discussing specifically $\variety \subseteq \field$ that are in the form $\variety = \zerofunc{\genpolyset[2]}$ for a set of polynomials $\genpolyset[2]$ that form a high rank factor.
Let us present some known results of the relative rank-bias relation for high rank varieites:
In the scenario when we are working relative to $\variety$, the equivalent for Theorem~\ref{high-rank-implies-low-bias} is also known when we assume $\degree < char(\basefield)$, as shown in~\cite[Theorem 1.8]{lampert2021relative}:
\begin{theorem}[High relative rank implies low bias in high rank varieties]\label{high-relative-schmidt-rank-implies-low-relative-bias}
    Let $\basefield$ be a finite field and let $0 \leq \degree < char(\basefield)$.
    Let $\epsilon > 0$ be a constant, and let $\varietypolycount \in \naturalnumbersset$.
    There exist $\varietyrankval^{\ref{high-relative-schmidt-rank-implies-low-relative-bias}} = \varietyrankval^{\ref{high-relative-schmidt-rank-implies-low-relative-bias}}(\basefield, \degree, \varietypolycount, \epsilon)$
    and $\rankval^{\ref{high-relative-schmidt-rank-implies-low-relative-bias}} = \rankval^{\ref{high-relative-schmidt-rank-implies-low-relative-bias}}(\basefield, \degree, \epsilon)$ such that the following holds:
    \newline
    Let $\varpolyset = (\varpoly_1,...,\varpoly_{\varietypolycount})$ be a collection of polynomials of degrees $\leq \degree$ with $\rank{\varpolyset} \geq \varietyrankval^{\ref{high-relative-schmidt-rank-implies-low-relative-bias}}$ and let $\genpoly$ be a polynomial of degree $\leq \degree$.
    \newline
    Then, if $\relrank{\varpolyset}{\genpoly} \geq \rankval^{\ref{high-relative-schmidt-rank-implies-low-relative-bias}}$, we have:
    \[
        \relbias{\varpolyset}{\genpoly} < \epsilon
    \]
\end{theorem}
\begin{note*}
    Note that the original statements in~\cite{lampert2021relative} are stated for a different definition of rank, noted as \emph{schmidt rank}.
    In the appendix~\ref{sec:comparing-ranks} we compare the two different definitions, and show that our definition of rank is comprehensive
    enough in a sense that a polynomial with high rank also has high schmidt rank.
    Additionally, we show that for a given $\rankval \in \naturalnumbersset$, the lower bound of rank required for a polynomial to be of schmidt rank $\geq \rankval$,
    is only $c \cdot \rankval$ for some constant $c \in \naturalnumbersset$.
\end{note*}
\begin{remark}
    Note that in the original statement of theorem~\ref{high-relative-schmidt-rank-implies-low-relative-bias} as stated in~\cite[Theorem 1.8]{lampert2021relative},
    there are good bounds on the rank needed for $\varpolyset$ and $\genpoly$ for the theorem to hold.
    \newline
    Specifically, there exist constants $A(\degree), B(\degree)$ such that for an error $\epsilon = \abs{\basefield^{-s}}$,
    if $\varietyrankval^{\ref{high-relative-schmidt-rank-implies-low-relative-bias}} = A(\varietypolycount + s)^B$ and $\rankval^{\ref{high-relative-schmidt-rank-implies-low-relative-bias}} = A(1 + s)^B$,
    then we have:
    \[
        \relbias{\zerofunc{\varpolyset}}{\genpoly} < \abs{\basefield}^{-s}
    \]
    \newline
    In our proof, it is enough that the bounds on $\rankval$ and $\varietyrankval$ are independent of $\blocklength$, thus we omit the exact bounds stated above and use the statement as stated in Theroem~\ref{high-relative-schmidt-rank-implies-low-relative-bias}.
\end{remark}

\begin{remark}\label{in-high-relative-rank-implies-low-relative-bias-epsilon-increasing-rank-requirment-decreasing}
Note that both $\rankval^{\ref{high-relative-schmidt-rank-implies-low-relative-bias}}(\basefield, \degree, \epsilon)$
and $\varietyrankval^{\ref{high-relative-schmidt-rank-implies-low-relative-bias}}(\basefield, \degree, \varietypolycount, \epsilon)$
are decreasing when $\epsilon$ is increasing.
This means for example, that for all $\epsilon^\prime \geq \epsilon$,
a variety that satisfies the theorem's rank condition for $\epsilon$ also satisfies the theorem's rank condition for $\epsilon^\prime$.
Therefore, a polynomial with rank $\geq \rankval^{\ref{high-relative-schmidt-rank-implies-low-relative-bias}}(\basefield, \degree, \epsilon^\prime)$
will have a bias $< \epsilon^\prime$.
\end{remark}
As a corollary of Theorem~\ref{high-relative-schmidt-rank-implies-low-relative-bias} and Corollary~\ref{relative-schimdt-rank-equals-schmidt-rank-if-the-variety-is-of-high-degree}, we have that
high rank varieties has the limited-relative-rank-bias property.
Formally, we have:
\begin{corollary}[High Rank Varieties Have the Limited-Relative Rank-Bias Property]\label{high-rank-variety-has-limited-rank-relative-bias-property}
    Let $\basefield$ be a finite field, and let $\varietydeg \in \naturalnumbersset$ such that $0 < \varietydeg < \abs{\basefield}$.
    Let $\epsilonlimitedrankbias > 0$ be a constant which represents the desired relative rank-bias limit.
    There exists $\funcdef{\rankbiasfunc_{\ref{high-rank-variety-has-limited-rank-relative-bias-property}}}{[\epsilonlimitedrankbias, \infty]}{\naturalnumbersset}$ with $\rankbiasfunc_{\ref{high-rank-variety-has-limited-rank-relative-bias-property}} \definedas \rankbiasfunc_{\ref{high-rank-variety-has-limited-rank-relative-bias-property}}(\basefield, \varietydeg)$
    such that the following holds:
    \newline
    Let $\varietypolycount \in \naturalnumbersset$ be an integer.
    There exists $\varietyrankval_{\ref{high-rank-variety-has-limited-rank-relative-bias-property}} \definedas \varietyrankval_{\ref{high-rank-variety-has-limited-rank-relative-bias-property}}(\basefield, \varietydeg, \varietypolycount, \epsilonlimitedrankbias)$ such that
    for every $\varpolyset = (\varpoly_1,...,\varpoly_{\varietypolycount})$ polynomial factor of degree $\leq \varietydeg$ with $\rank{\varpolyset} \geq \varietyrankval$, we have:
    \newline
    The variety $\variety = \zerofunc{\varpolyset}$ has the $(\rankbiasfunc_{\ref{high-rank-variety-has-limited-rank-relative-bias-property}}, \basefield, \varietydeg, \epsilonlimitedrankbias)$-limited-relative-rank-bias property.
\end{corollary}
\begin{proof}
    Let $\basefield$ be a finite field, and let $\varietydeg \in \naturalnumbersset$ such that $0 < \varietydeg < \abs{\basefield}$.
    Let $\epsilonlimitedrankbias > 0$.
    We choose:
    \[
        \rankbiasfunc_{\ref{high-rank-variety-has-limited-rank-relative-bias-property}}(\epsilon) \definedas
        \rankval_{\ref{high-relative-schmidt-rank-implies-low-relative-bias}}(\basefield, \varietydeg, \epsilon)
    \]
    Note that for every $\epsilon$ in its domain, $\rankbiasfunc_{\ref{high-rank-variety-has-limited-rank-relative-bias-property}}$ does not depend on $\epsilonlimitedrankbias$.
    Let $\varietypolycount \in \naturalnumbersset$.
    Now, we choose:
    \[
        \varietyrankval_{\ref{high-rank-variety-has-limited-rank-relative-bias-property}}(\basefield, \varietydeg, \varietypolycount, \epsilonlimitedrankbias) \definedas
        \varietyrankval_{\ref{high-relative-schmidt-rank-implies-low-relative-bias}}(\basefield, \varietydeg, \varietypolycount, \epsilonlimitedrankbias)
    \]
    Using Theorem~\ref{high-relative-schmidt-rank-implies-low-relative-bias} that shows high rank implies low bias in $\variety$,
    and the assumption that $\epsilon \geq \epsilonlimitedrankbias$ (specifically Remark~\ref{in-high-relative-rank-implies-low-relative-bias-epsilon-increasing-rank-requirment-decreasing})
    concludes the proof.
\end{proof}

%% file: regularization_relative_to_X.tex
\section[Regularization Relative to \titlevariety]{Regularization Relative to \titlevariety}\label{sec:regularization-relative-to-X}

In this section, we generalize the definitions and statements regarding factors and regularization in $\field$,
to their corresponding definitions and statements to relative rank in respect of $\variety \subseteq \field$.
\newline
Note that in oppose to the previous chapter that we discussed a general $U$ and $A \subseteq U$,
in this chapter we discuss only $U = A = \field$.
This is done for clearance and to avoid defining definitions we will not use in our main proof.
\begin{definition}[Measurable Relative to $\variety$]
    Let $\genfuncset = \set{\genfunc_1,...,\genfunc_c}$ be a set of functions $\funcdef{\genfunc_i}{\field}{\basefield}$.
    We say a function $\funcdef{\genfunc[2]}{\field}{\basefield}$ is \emph{measurable in respect of $\genfuncset$ relative to $\variety$},
    or \emph{$\variety$-relative $\genfuncset$-measurable},
    if there exists a function $\funcdef{\relativeremainder{\genfunc[2]}}{\field}{\basefield}$ with $\restrictfunc{\relativeremainder{\genfunc[2]}}{\variety} \equiv 0$
    and a function $\funcdef{\Gamma}{\basefield^c}{\basefield}$ such that:
    \[
        \forall a \in A: \genfunc[2](a) = \Gamma(\genfunc_1(a),...,\genfunc_c(a)) + \relativeremainder{\genfunc[2]}(a)
    \]
    And:
    \[
        \deg(\genfunc[2] - \relativeremainder{\genfunc[2]}) \leq \deg(\genfunc[2])
    \]
    We sometimes refer to $\Gamma$ as the \emph{$\variety$-relative measurement function}.
\end{definition}

\begin{note*}
    Note that if $\deg(\genfunc[2] - \relativeremainder{\genfunc[2]}) \leq \deg(\genfunc[2])$ as discussed above,
    then the same bound  also bounds the degree of the remainder, i.e. $\deg(\relativeremainder{\genfunc[2]}) \leq \deg(\genfunc[2])$.
    Therefore $\relativeremainder{\genfunc[2]}$ is a valid $\variety$-remainder of $\genfunc[2]$.
    Moreover, this requirement is equivalent to the definition above,
    as if $\deg(\relativeremainder{\genfunc[2]}) \leq \deg(\genfunc[2])$, then we also have $\deg(\genfunc[2] - \relativeremainder{\genfunc[2]}) \leq \deg(\genfunc[2])$.
\end{note*}
\begin{note*}
    Also note that without the bound on the degree of the remainder,
    being measurable relative to $\variety$ is in fact equivalent for being a measurable in $A = \variety$.
    This is true because under these conditions, the remainder $\relativeremainder{\genfunc[2]}$ has no constraints but $\restrictfunc{\relativeremainder{\genfunc[2]}}{\variety} \equiv 0$,
    thus the condition left on the measurement is just being a measurement to $\genfunc[2]$ in $\variety$.
\end{note*}
\begin{remark}
    If $\genfunc[2]$ is a function that it is $\genfuncset$-measurable relative to $\variety$,
    then every value of $\genfunc[2]$ can be determined by the values of $\genfunc_1,...,\genfunc_c$ up to a remainder $\relativeremainder{\genfunc[2]}$ of degree $\leq \degree$.
    Thus, perhaps we do not know that the function $\genfunc[2]$ is constant inside every atom of $\genfuncset$ as in a regular semantic refinement,
    but we do know that there exists a function $(\genfunc[2] - \relativeremainder{\genfunc[2]})$ that equals to $\genfunc[2]$ on $\variety$, is constant on every atom of $\genfuncset$ and it is a function with a bounded degree i.e. $\deg(\genfunc[2] - \relativeremainder{\genfunc[2]})\leq \deg(\genfunc[2])$.
\end{remark}

Next, we present a new type of refinement, which is a relaxation of semantic refinement.
This relaxation will allow us to discuss the corresponding claim of the polynomial regularity lemma (Lemma~\ref{regularization-in-Fn-lemma}) for relative rank (instead of rank).
\begin{definition}[Semantic Refinement Relative to $\variety$]
    Let $\factor$ and $\factor^\prime$ be polynomial factors on $\field$, defined by sets of polynomials $\genpolyset, \genpolyset^{\prime}$ respectively,
    and let $\degree \in \naturalnumbersset$.
    We say a factor $\factor^\prime$ is a \emph{semantic refinement relative to $\variety$} of the factor $\factor$,
    or \emph{$\variety$-relative semantic refinement},
    if the following holds:
    Every function $\funcdef{\genfunc}{\field}{\basefield}$ that is $\genpolyset$-measurable relative to $\variety$,
    is also $\genpolyset^{\prime}$-measurable relative to $\variety$.
    If the definition above holds, we denote $\factor^{\prime} \relsemrefineex{\variety} \factor$.
\end{definition}
\begin{note*}
    It is easy to see that this relation is transitive, i.e. if
    $\factor^{\prime} \relsemrefineex{\variety} \factor$ and
    $\factor^{\prime\prime} \relsemrefineex{\variety} \factor^{\prime}$,
    then $\factor^{\prime\prime} \relsemrefineex{\variety} \factor$.
\end{note*}
\begin{remark}
    In $\variety$, semantic refinements relative to $\variety$ behave the same as regular semantic refinements in the perspective of being measurable:
    every function that is $\genpolyset$-measurable in $\variety$ is also $\genpolyset^{\prime}$-measurable in $\variety$.
    However, the two definitions behave differently in the perspective of being measurable in $\field$.
    Specifically, in relative semantic refinements,
    if $\genfunc[2]$ is a $\genpolyset$-measurable function it is not necessarily $\genpolyset^{\prime}$-measurable.
    However, it is measurable up to a remainder $\relativeremainder{\genfunc[2]}$ of degree $\leq \deg(\genfunc[2])$ such that $\restrictfunc{\relativeremainder{\genfunc[2]}}{\variety} \equiv 0$.
\end{remark}
\begin{corollary}\label{relative-semantic-refinement-is-restricted-semantic-refinement}
    If $\factor^{\prime} \relsemrefineex{\variety} \factor$, then in $\variety$ it is a regular semantic refinement, i.e. $\factor^{\prime} \semrefineex{\variety} \factor$.
\end{corollary}

Next, we present a new regularization process that allows us to increase the \emph{relative} rank of a factor without increasing the size of the factor too much (independent of $\blocklength$).
This regularization process generalizes the regularization process in $\field$, which was first presented by~\cite[2.3]{green2007distribution}.
We call this type of regularization process a \emph{relative-regularization process} relative to $\variety$, shorthand by $\variety$-regularization
For a specific function $\rankfunc$, we will sometimes call applying this lemma a \emph{$\rankfunc$-$\variety$-regularization}.
Note that to allow such a relative-regularization process to hold, we must use the relaxed definition of semantic refinement that is presented above.
\begin{theorem}\label{theorem:regularization-in-X}
    Let $\funcdef{\rankfunc}{\naturalnumbersset}{\naturalnumbersset}$ be a non-decreasing function and let $\degree \in \naturalnumbersset$.
    There exists $\funcdef{C_{\rankfunc, \degree}^{\ref{theorem:regularization-in-X}}}{\naturalnumbersset}{\naturalnumbersset}$ such that the following holds:
    Let $\factor$ be a factor defined by polynomials $\genpolyset = (\genpoly_1,...,\genpoly_c)$ where for all $i \in [c]$: $\funcdef{\genpoly_i}{\field}{\basefield}$ and $\deg(\genpoly_i) \leq \degree$.
    Then, there is an $\rankfunc$-$\variety$-regular factor $\factor^\prime$ defined by polynomials $\genpolyset^{\prime} = (\genpoly^{\prime}_{1},...,\genpoly^{\prime}_{c^\prime})$ where
    for all $i \in [c]$: $\funcdef{\genpoly^{\prime}_i}{\field}{\basefield}$ and $\deg(\genpoly^{\prime]}_{i}) \leq \degree$ such that
    $\factor^\prime \relsemrefineex{\variety} \factor$ and $c^\prime \leq C_{\rankfunc, \degree}^{\ref{theorem:regularization-in-X}}(c)$.
    \newline
    Moreover, if $\factor \synrefine \bar{\factor}$ for some polynomial factor $\bar{\factor}$ with
    relative rank of at least $\rankfunc(c^\prime)+c^\prime+1$ and rank of at least ${\rankfunc_{\ref{preserving-degree-starting-field}}(\basefield, \degree, c^{\prime})} + c^\prime + 1$,
    then we can require that $\factor^\prime \synrefine \bar{\factor}$.
\end{theorem}
\begin{proof}
    We follow the lines of the proof given by~\cite{book}[Lemma 7.29], but here, we wish to increase the \emph{relative} rank of the factor instead of its rank.
    We present an iterative process, which will eventually lead us to a factor of size $c^{\prime}$ with relative rank higher than $\rankfunc(c^\prime)$,
    that is a semantic refinement relative to $\variety$.
    Let $\degree \in \naturalnumbersset$,
    and let $\factor$ be a polynomial factor defined by $\genpolyset = (\genpoly_1,...,\genpoly_c)$ such that $\funcdef{\genpoly_i}{\field}{\basefield}$ of degree $\leq \degree$.
    Define $M(\factor) \definedas (M_{\degree},...,M_1) \in \naturalnumbersset^{\degree}$,
    where $M_i$ denotes the number of polynomials in $\genpolyset$ that have degree exactly $i$.
    Thus, $\sum_{i=1}^{\degree}M_i = c$.
    We define the lexicographical order on $\naturalnumbersset^{\degree}$ where $M > M^{\prime}$ if and only if $M_i > M^{\prime}_i$ for some $1 \leq i \leq \degree$,
    and $M_j = M^{\prime}_j$ for all $j > i$.
    This proof will be by transfinite induction on $M$ under the lexicographical order.
    Next we describe a step of the regularization process.
    \newline
    Let $\factor$ be a polynomial factor defined by $\genpolyset = (\genpoly_1,...,\genpoly_c)$.
    Note that this is an abuse of notations: the factor $\factor$ and the set $\genpolyset$ refer to the original factor in the first step, and also to the current factor in the middle of the relative-regularization process.
    If $\factor$ is $\rankfunc$-$\variety$-regular, then we are done.
    Otherwise, we change $\factor$ as follows:
    First, we denote $\rankfunc_{\ref{preserving-degree-starting-field}}^{\basefield, \degree}(c) \definedas \rankfunc_{\ref{preserving-degree-starting-field}}(\basefield, \degree, c)$,
    and we $\rankfunc_{\ref{preserving-degree-starting-field}}$-regularize $\genpolyset$ using lemma~\ref{regularization-in-Fn-lemma}
    to get a set of polynomials $\genpolyset_1 = (\genpoly^1_1,...,\genpoly^1_{c_1})$ of degree $\leq \degree$,
    which defines a factor $\factor_1$ and has a rank $\geq \rankfunc_{\ref{preserving-degree-starting-field}}^{\basefield, \degree}(c_1)$.
    Note that $M(\factor_1) \leq M(\factor)$.
    Then, again, if somehow $\factor_1$ is now $\rankfunc$-$\variety$-regular, we are done.
    \newline
    Otherwise, by definition, there exists some linear combination of the polynomials in $\genpolyset_{1}$ that
    has $\degree^\star$-relative rank less than $\rankfunc(c_1)$,
    where $\degree^\star$ is the maximal degree that participates in the linear combination.
    Let $\vec{\genpoly}(x) = \sum_{i=0}^{c_1}{\lambda_i \genpoly^{1}_i(x)}$ where $\vec{0} \neq \vec{\lambda}\in \basefield^{c_1}$,
    be the linear combination with $\drelrank{\degree^\star}{\variety}{\vec{\genpoly}} \leq \rankfunc(c_1)$ where $\degree^\star \definedas \max_{i \in \sparens{c_1}}{\deg(\lambda_i \genpoly^1_i)}$.
    By definition of relative rank, there exists $\relativeremainder{\genpoly} \in \allpolyset{\leq \deg(\vec{\genpoly})}{\field}{\basefield}$ with $\restrictfunc{\relativeremainder{\genpoly}}{\variety} \equiv 0$ such that
    $\drank{\degree^\star}{\vec{\genpoly} - \relativeremainder{\genpoly}} \leq \rankfunc(c_1)$.
    Note that $\deg(\relativeremainder{\genpoly}) \leq \degree^\star$.
    By definition of $\degree^\star$-rank, we have that we can decompose $\vec{\genpoly} - \relativeremainder{\genpoly}$ as a function of $\rankfunc(c_1)$ polynomials of degree $\leq \degree^\star - 1$.
    In other words, there exist a measurement function $\funcdef{\vec{\Gamma}}{\basefield^{\rankfunc(c_1)}}{\basefield}$ and polynomials $\genpoly[2]_1,...,\genpoly[2]_{\rankfunc(c_1)}$
    with $\deg(\genpoly[2]_i) \leq \degree^\star - 1$ such that:
    \[
        \forall a \in \field: \vec{\genpoly}(a) - \relativeremainder{\genpoly}(a) = \vec{\Gamma} \parens {\genpoly[2]_1(a),...,\genpoly[2]_{\rankfunc(c_1)}(a)}
    \]
    Now, let $\genpolyset^{\star} \subseteq \genpolyset_1$ be the set of all such maximal-degree polynomials,
    and let $i^{\star}$ be chosen such that $\genpoly^{1}_{i^\star} \in \genpolyset^{\star}$.
    Note that the set $\genpolyset^\star$ is non empty, as by definition, $\degree^{\star}$ is the maximal degree of polynomial in the expression $\sum_{i=1}^{c_1}{\lambda_{i} \genpoly^{1}_i}$ such that $\lambda_i \neq 0$.
    \newline
    For the next step, define the polynomial factor $\factor_2$ be the polynomial factor defined by the set:
    \[
        \genpolyset_2 \definedas \genpolyset_1 \setminus \set{\genpoly^1_{i^\star}} \cup \set{\genpoly[2]_1,...,\genpoly[2]_{\rankfunc(c_1)}}
    \]
    Finally, the factor $\factor_2$ will be the factor returned from the relative-regularization step.
    \newline
    It is easy to see that if the process above halts, we get a $\rankfunc$-$\variety$-regular factor.
    Now, we prove the first part of the lemma by showing the following claims:
    \begin{claim}
        The factor generated from the regularization above is of bounded size: a bound that may depend on $\rankfunc, \degree, c$, but does not depend on $\blocklength$.
        Formally, we claim that there exists $\funcdef{C^{\ref{theorem:regularization-in-X}}_{\rankfunc, \degree}}{\naturalnumbersset}{\naturalnumbersset}$
        such that we have $c^{\prime} \leq C^{\ref{theorem:regularization-in-X}}_{\rankfunc, \degree}(c)$.
    \end{claim}
    \begin{proof}
        It is enough to prove the following:
        \begin{enumerate}
            \item~\label{relative-regularization-step-factor-size-is-bounded}
            In each step, the amount of polynomials there are in $\genpolyset_1, \genpolyset_2$ are bounded by a bound that depend only on $\rankfunc, \degree, c$ (independent of $\blocklength$).
            \item~\label{relative-regularization-amount-of-steps-is-bounded}
            The number of steps of the relative-regularization process is also bounded by a bound that depends only on $\rankfunc, \degree, c$ (independent of $\blocklength$).
        \end{enumerate}
        The combination of these two will obtain the desired bound of the amount of polynomials in the last-step regularized factor, which is $C^{\ref{theorem:regularization-in-X}}_{\rankfunc, \degree}(c)$.
        Note that the bound on the last-step relative-regularized factor in not simply the multiplication of the two bounds,
        but a recursively-substitution of the bound in~\ref{relative-regularization-step-factor-size-is-bounded},
        a bounded amount of times (bounded by the bound in~\ref{relative-regularization-amount-of-steps-is-bounded}).
        \newline
        For~\ref{relative-regularization-step-factor-size-is-bounded}, we first notice that the number of polynomials in the regular regularization process is bounded,
        specifically we have $\abs{\genpolyset^1} = c_{1} \leq C^{\ref{regularization-in-Fn-lemma}}_{\rankfunc_{\ref{preserving-degree-starting-field}}, \degree}(c)$.
        Moreover, the polynomial factor $\factor_2$ is generated by adding at most $\rankfunc(c_1)$ polynomials to the factor, and thus we have $\abs{\genpolyset_2} \leq c_1 + \rankfunc(c_1)$ which is also bounded by substituting the bound on $c_1$.
        \newline
        For~\ref{relative-regularization-amount-of-steps-is-bounded}, we use the transfinite induction on $M$ we mentioned earlier to show that the process must halt after a bounded number of steps.
        Formally, we show that there exist $M^{\prime}$ which depends only on $M(\factor)$ such that $M(\factor_2) \leq M^{\prime} < M(\factor)$.
        This will bound the number of steps by a value that depend only on $M(\factor)$, which depends only on $\rankfunc, \degree, c$.
        To do so, we first notice that the regular regularization does not increase the value of $M$, i.e. $M(\factor_1) \leq M(\factor)$.
        Thus, we can focus on the second part of the relative-regularization.
        In this part, we replace a single degree $\degree^{\star}$ polynomial by at most $\rankfunc(c_1)$ polynomials of degree $\leq \degree^{\star} - 1$.
        Therefore, by choosing $M^{\prime} \definedas (M_{\degree}, ..., M_{\degree^{\star}+1}, M_{\degree^{\star}}-1, M_{\degree^{\star}-1}+\rankfunc(c_1),...,M_1+\rankfunc(c_1))$
        we get that $M(\factor_2) \leq M^{\prime} < M(\factor_1) \leq M(\factor)$, which concludes~\ref{relative-regularization-amount-of-steps-is-bounded}.
    \end{proof}
    \begin{claim}
        The factor generated from the regularization above is a $\variety$-relative semantic refinement of the original factor, i.e $\factor^\prime \relsemrefineex{\variety} \factor$.
    \end{claim}
    \begin{proof}
        It is enough to show that in each step, the factors generated by the relative-regularization process are semantic refinements relative to $\variety$ of the previous step's factor.
        Specifically, we show $\factor_{2} \relsemrefineex{\variety} \factor_{1} \relsemrefineex{\variety} \factor$ and the claim will follow from transitivity of relative semantic refinements.
        \newline
        We start by proving $\factor_1 \relsemrefineex{\variety} \factor$.
        Let $\funcdef{\genfunc}{\field}{\basefield}$ be a function that is $\genpolyset$-measurable relative to $\variety$.
        We denote $\degree_{\genfunc} \definedas \deg(\genfunc)$.
        By definition, there exists $\funcdef{\Gamma}{\basefield^c}{\basefield}$,
        $\funcdef{\relativeremainder{\genfunc}}{\field}{\basefield}$ where $\deg(\relativeremainder{\genfunc}), \deg(\genfunc - \relativeremainder{\genfunc}) \leq \degree_{\genfunc}$ and $\restrictfunc{\relativeremainder{\genfunc}}{\variety} \equiv 0$, such that:
        \[
            \forall a \in \field: \genfunc(a) = \Gamma(\genpoly_1(a),...,\genpoly_c(a)) + \relativeremainder{\genfunc}(a)
        \]
        Clearly, the function $\Gamma(\genpoly_1(a),...,\genpoly_c(a))$ is $\genpolyset$-measurable in $\field$,
        and because we have $\factor \semrefine \factor_1$, it is also $\genpolyset_1$-measurable in $\field$.
        Thus there exists $\funcdef{\Gamma_1}{\basefield^{c_1}}{\basefield}$ such that:
        \[
            \forall a \in \field: \genfunc(a) = \Gamma_1(\genpoly^{1}_{1}(a),...,\genpoly^{1}_{c_1}(a)) + \relativeremainder{\genfunc}(a)
        \]
        And therefore we have $\factor_1 \relsemrefineex{\variety} \factor$.
        \newline
        Now, we prove $\factor_2 \relsemrefineex{\variety} \factor_1$.
        Let $\funcdef{\genfunc}{\field}{\basefield}$ be a function that is $\genpolyset_1$-measurable relative to $\variety$.
        Again, we denote $\degree_{\genfunc} \definedas \deg(\genfunc)$,
        and by definition there exists $\funcdef{\Gamma_1}{\basefield^c}{\basefield}$,
        $\funcdef{\relativeremainder{\genfunc_1}}{\field}{\basefield}$ where $\deg(\genfunc - \relativeremainder{\genfunc_1}) \leq \degree_{\genfunc}$ and $\restrictfunc{\relativeremainder{\genfunc_1}}{\variety} \equiv 0$, such that:
        \begin{equation} \label{eq:f-decomposition-a}
            \forall a \in \field: \genfunc(a) = \Gamma_1(\genpoly^{1}_{1}(a),...,\genpoly^{1}_{c_1}(a)) + \relativeremainder{\genfunc}_1(a)
        \end{equation}
        Note that we also have $\deg(\relativeremainder{\genfunc_1}) \leq \degree_{\genfunc}$.
        We will refer this equation, and its simplifications we do throughout the proof, as \emph{the $\genpolyset_1$-decomposition of $\genfunc$}.
        \newline
        We wish to show that there exists $\funcdef{\Gamma_2}{\basefield^{c_2}}{\basefield}$ and
        $\funcdef{\relativeremainder{\genfunc}_2}{\field}{\basefield}$ where $\deg(\genfunc - \relativeremainder{\genfunc}_2) \leq \degree_{\genfunc}$ and $\restrictfunc{\relativeremainder{\genfunc}_2}{\variety} \equiv 0$, such that:
        \[
            \forall a \in \field: \genfunc(a) = \Gamma_2 \parens {\genpoly^1_1(a),...\genpoly^1_{i^\star - 1}(a), \genpoly^1_{i^\star + 1}(a),..., \genpoly^1_c(a), \genpoly[2]_1(a),...,\genpoly[2]_{\rankfunc(c_1)}(a)} + \relativeremainder{\genfunc}_2(a)
        \]
        We will do so using the $\genpolyset_1$-decomposition of $\genfunc$.
        Note that showing $\deg(\genfunc - \relativeremainder{\genfunc}) \leq \degree_{\genfunc}$ is equivalent of showing $\deg(\relativeremainder{\genfunc}_2) \leq \degree_{\genfunc}$.
        \newline
        First, by the way we built $\genpolyset_2$, using the same notations in the regularization step, we have:
        \[
            \forall a \in \field: \genpoly_{i^\star}^1(a) =
                \vec{\Gamma} \parens {\genpoly[2]_1(a),...,\genpoly[2]_{\rankfunc(c_1)}(a)}
                + \relativeremainder{\genpoly}(a)
                - \sum_{i \neq i^\star}{\genpoly_{i}^1(a)}
        \]
        Next, we substitute the value of $\genpoly_{i^\star}^1$ in the $\genpolyset_1$-decomposition of $\genfunc$ (\ref{eq:f-decomposition-a}),
        and get another decomposition of $\genfunc$ that does not depend on $\genpoly_{i^\star}^1$.
        Specifically we have:
        \begin{align} \label{eq:f-decomposition-b}
            \forall a \in \field: \genfunc(a) &=
        \Gamma_1 \parens
            {\genpoly^{1}_{1}(a)
                ,...,
                \parens{
                    \vec{\Gamma} \parens {\genpoly[2]_1(a),...,\genpoly[2]_{\rankfunc(c_1)}(a)}
                    + \relativeremainder{\genpoly}(a)
                    - \sum_{i \neq i^\star}{\genpoly_{i}^1(a)}}
                ,...,
                \genpoly^{1}_{c_1}(a))}\\
            &+ \relativeremainder{\genfunc}_1(a)
        \end{align}
        We wish to use the equation above to show that $\genfunc$ is $\genpolyset_2$-measurable relative to $\variety$.
        However, in order to show that the equation above is in the desired structure that proves that $\genfunc$ is $\genpolyset_2$-measurable,
        the expression inside $\Gamma_1$ must not depend on $\relativeremainder{\genpoly}$ because $\relativeremainder{\genpoly} \notin \genpolyset_2$.
        Note that this is enough as the rest of the polynomials in the expression above are in $\genpolyset_2$,
        and therefore without $\relativeremainder{\genpoly}$ the expression is $\genpolyset_2$-measurable.
        \newline
        To do so, we start by simplifying some of the notations.
        We denote:
        \[
            \vec{\genpoly}_2(a) \definedas \vec{\Gamma} \parens {\genpoly[2]_1(a),...,\genpoly[2]_{\rankfunc(c_1)}(a)} - \sum_{i \neq i^\star}{\genpoly_{i}^1(a)}
        \]
        This is the part of the sum that decomposes $\genpoly_{i^\star}^1(a)$ that is $\genpolyset_2$-measurable,
        thus the following equality applies:
        \[
            \genpoly_{i^\star}^1(a) = \vec{\genpoly}_2(a) + \relativeremainder{\genpoly}(a)
        \]
        where $\deg(\vec{\genpoly_2}), \deg(\relativeremainder{\genpoly}) \leq \degree^{\star}$.
        Using this notation, we write the $\genpolyset_1$-decomposition of $\genfunc$ (\ref{eq:f-decomposition-b}), and get:
        \begin{equation} \label{eq:f-decomposition-c}
            \forall a \in \field: \genfunc(a) = \Gamma_1 \parens
            {\genpoly^{1}_{1}(a)
                ,...,
                \parens{
                    \vec{\genpoly}_2(a)
                    - \relativeremainder{\genpoly}(a)}
                ,...,
                \genpoly^{1}_{c_1}(a))}
            + \relativeremainder{\genfunc}_1(a)
        \end{equation}
        Now, we use the following key observation:
        $\rank{\genpolyset_1} \geq \rankfunc_{\ref{preserving-degree-starting-field}}(\basefield, \degree, c_1)$,
        and as $\deg(\Gamma_1(\genpoly_1^1,...,\genpoly^1_{c_1})) \leq \degree_{\genfunc}$
        we can use Lemma~\ref{preserving-degree-starting-field} to achieve that $\Gamma_1$ is a polynomial of the form:
        \begin{equation*} \label{eq:regularization-gamma-is-a-polynomial}
            \Gamma_1(z_1,...,z_{c_1})  =
            \sum_{\alpha \in \sparens{\basefieldsize - 1}^{c_1}}
            {C_{\alpha} \cdot {\prod_{i = 1}^{c_1}}{z_i^{\alpha_i}}}
        \tag{$\star$}
        \end{equation*}
        where $C_{\alpha} = 0$ whenever $\sum_{i = 1}^{c_1}(\alpha_i \cdot \deg(\genpoly_i^1)) > \degree_{\genfunc}$.
        \newline
        Next, we substitute the polynomial structure of $\Gamma_1$ \eqref{eq:regularization-gamma-is-a-polynomial}
        in the $\genpolyset_1$-decomposition of $\genfunc$~\eqref{eq:f-decomposition-c},
        and observe what happens to each summand monomial with non-zero coefficients of $\Gamma_1$ in the expression after the substitution.
        \newline
        We will show that each such monomial is either $\genpolyset_2$-measurable,
        or a sum of a $\genpolyset_2$-measurable function with a valid $\variety$-remainder, i.e. a polynomial of degree $\leq \degree_{\genfunc}$ that is $\equiv 0$ in $\variety$.
        Note that if this is true for each monomial,
        every linear combination of such monomials is also a sum of $\genpolyset_2$-measurable function with a valid $\variety$-remainder.
        Thus, this will also be true for the entire decomposition of $\genfunc$, as it is a linear combination of such monomials summed with a valid remainder $\relativeremainder{\genfunc}_1$.
        This will conclude the proof.
        \newline
        Let $\alpha = (\alpha_1,...,\alpha_{c_1})$ be a vector of degrees that represents such a monomial.
        If $\alpha_{i^\star} = 0$, then the monomial is in the form:
        \[
            \prod_{i \in [c_1]}{{\genpoly_i}^{\alpha_i}} =
            \prod_{i \in [c_1] \setminus \set{i^{\star}}}{{\genpoly_i}^{\alpha_i}}
        \]
        and therefore it is clearly $\genpolyset_2$-measurable as all the polynomials in the expression above are in $\genpolyset_2$.
        \newline
        Next, if $\alpha_{i^\star} \neq 0$, then the monomial is in the form:
       \begin{equation} \label{eq:monomial-of-gamma}
            \prod_{i \in [c_1]}{{\genpoly_i}^{\alpha_i}} =
            (\vec{\genpoly_2} + \relativeremainder{\genpoly})^{\alpha_{i^\star}} \cdot
                \parens{\prod_{i \in [c_1] \setminus \set{i^{\star}}}{{\genpoly_i}^{\alpha_i}}}
       \end{equation}
        where $\sum_{i \in [c_1]}(\alpha_i \cdot \deg(\genpoly_i^1)) \leq \degree_{\genfunc}$.
        As $\deg(\vec{\genpoly_2} + \relativeremainder{\genpoly}) = \deg(\genpoly_{i^\star}) =\degree^{\star}$, we have:
        \[
            \deg \parens{\prod_{i \in [c_1] \setminus \set{i^{\star}}}{{\genpoly_i}^{\alpha_i}}} =
                \sum_{i \in [c_1] \setminus {i^{\star}}}(\alpha_i \cdot \deg(\genpoly_i^1))
                \leq \degree_{\genfunc} - \alpha_{i^\star} \cdot \degree^{\star}
        \]
        Now, we open the left brackets in (\ref{eq:monomial-of-gamma}), i.e $(\vec{\genpoly_2} + \relativeremainder{\genpoly})^{\alpha_{i^\star}}$.
        This enables us to separate the monomial to the part that only depend on $\vec{\genpoly_2}$ summed with a polynomial with bounded degree multiplied by $\relativeremainder{\genpoly}$ (and therefore a valid remainder).
        To be more specific, the monomial is in the form:
        \[
            (\vec{\genpoly_2} + \relativeremainder{\genpoly})^{\alpha_{i^\star}} = \vec{\genpoly_2}^{\alpha_{i^\star}} + \relativeremainder{\genpoly_{\alpha}}
        \]
        for some polynomial $\relativeremainder{\genpoly_{\alpha}}$ such that:
        \begin{enumerate}
            \item $\relativeremainder{\genpoly_{\alpha}}$ is of degree
                    $\deg(\relativeremainder{\genpoly_{\alpha}}) \leq \max \set{\deg(\vec{\genpoly_2}), \deg(\relativeremainder{\genpoly)}} \cdot \alpha_{i^\star} \leq \alpha_{i^\star} \cdot \degree^{\star}$
            \item $\relativeremainder{\genpoly_{\alpha}}$ is a multiple of $\relativeremainder{\genpoly}$, and therefore $\restrictfunc{\relativeremainder{\genpoly_{\alpha}}}{\variety} \equiv 0$
        \end{enumerate}
        Therefore, by substituting the left brackets back to the equation (\ref{eq:monomial-of-gamma}) and as $\vec{\genpoly_2}$ and $\genpoly_i$ for $i \neq i^{\star}$ are $\genpolyset_2$-measurable,
        one can see that the monomial is a sum of a $\genpolyset_2$-measurable polynomial with a valid remainder.
        Specifically, the remainder $\equiv 0$ in $\variety$, and its degree is $\leq \alpha_{i^\star} \cdot \degree^{\star} + \degree_{\genfunc} - \alpha_{i^\star} \cdot \degree^{\star} = \degree_{\genfunc}$.
        This concludes the proof of the claim.
        \end{proof}
    Now, it remains to prove the second part of the Theorem~{\ref{theorem:regularization-in-X}}.
    \begin{claim}
        If $\factor \synrefine \bar{\factor}$ for some polynomial factor $\bar{\factor}$ with
        relative rank of at least $\rankfunc(c^\prime)+c^\prime+1$ and rank of at least ${\rankfunc_{\ref{preserving-degree-starting-field}}(\basefield, \degree, c^{\prime})} + c^\prime + 1$,
        then we can require that $\factor^\prime \synrefine \bar{\factor}$.
    \end{claim}
    \begin{proof}
        We will show claim step-by-step.
        We denote by $\genpolyset, \bar{\genpolyset}, \genpolyset_1, \genpolyset_2$ the polynomial sets that generate the factors $\factor, \bar{\factor}, \factor_1, \factor_2$.
        Note that $\factor_1, \factor_2$ are the factors in the current step of the regularization process, and thus change in each step of the proof.
        We show that in each step, if $\factor \synrefine \bar{\factor}$ for some polynomial factor $\bar{\factor}$ with
        relative rank of at least $\rankfunc(c^\prime)+c^\prime+1$ and rank of at least ${\rankfunc_{\ref{preserving-degree-starting-field}}(\basefield, \degree, c^{\prime})} + c^\prime + 1$,
        then we can require that $\factor_1 \synrefine \bar{\factor}$, and also that $\factor_2 \synrefine \bar{\factor}$.
        \newline
        For the first part, we have $\factor_1 \synrefine \bar{\factor}$ by a simple usage of the second part of lemma~\ref{regularization-in-Fn-lemma},
        as:
        \[
            \rank{\bar{\genpolyset}}
            > \rankfunc_{\ref{preserving-degree-starting-field}}(\basefield, \degree, c^{\prime}) + c^\prime + 1
            \geq \rankfunc_{\ref{preserving-degree-starting-field}}(\basefield, \degree, c_1) + c_1 + 1
        \]
        \newline
        Now we prove the second part.
        We show that in the current regularization step, we could replace $\genpoly^1_{i^\star} \in \genpolyset_1$ such that $\genpoly^1_{i^\star} \notin \bar{\genpolyset}$.
        Note that this is possible whenever $\genpolyset^\star\cap \bar{\genpolyset} \neq \emptyset$ as the choice of $i^\star$ is arbitrary in polynomials which are in $\genpolyset^\star$.
        \newline
        Assume that is not possible and the factor $\genpolyset_1$ is still not $\rankfunc$-$\variety$-regular.
        Then, we have a linear combination $\vec{\genpoly}(x) \definedas \sum_{i=0}^{c_1}{\lambda_{i}\genpoly^1_i(x)}$ with $\drelrank{\degree^\star}{\variety}{\vec{\genpoly}} \leq \rankfunc(c_1)$
        where $\degree^\star = \max_{i \in \sparens{c_1}} {\deg(\lambda_i \genpoly^1_i)}$.
        We denote by $I^{\star} \subseteq [c_1]$ the set of indexes of such maximal-degree polynomials.
        By this notation, our assumption states that for all $i \in I^{\star}$ we have $\genpoly^1_i \in \bar{\genpolyset}$.
        Additionally, note that for all $i \notin I^{\star}$ we have $\deg(\genpoly^1_i) < \degree^{\star}$.
        Therefore, as the linear combination is of $\degree^\star$-relative rank $\leq \rankfunc(c_1)$,
        there exists a polynomial $\relativeremainder{\genpoly}$ of degree $\leq \deg(\vec{\genpoly}) \leq \degree^\star$ with $\restrictfunc{\relativeremainder{\genpoly}}{\variety} \equiv 0$
        such that $\drank{\degree^\star}{\vec{\genpoly} - \relativeremainder{\genpoly}} \leq \rankfunc(c_1)$.
        In other words, there exist a measurement function $\funcdef{\vec{\Gamma}}{\basefield^{\rankfunc(c_1)}}{\basefield}$ and polynomials $\genpoly[2]_1,...,\genpoly[2]_{\rankfunc(c_1)}$
        with $\deg(\genpoly[2]_i) \leq \degree^\star$ such that:
        \[
            \forall a \in \field: \vec{\genpoly}(a) - \relativeremainder{\genpoly}(a) = \vec{\Gamma} \parens {\genpoly[2]_1(a),...,\genpoly[2]_{\rankfunc(c_1)}(a)}
        \]
        By a simple calculation we have:
        \[
            \forall a \in \field: \sum_{i \in I^{\star}}{\genpoly^1_i(a)} - \relativeremainder{\genpoly}(a) =  \vec{\Gamma} \parens {\genpoly[2]_1(a),...,\genpoly[2]_{\rankfunc(c_1)}(a)} +  \sum_{i \notin I^\star}{\genpoly^1_i(a)}
        \]
        and by this we found a linear combination of polynomials in $\bar{\genpolyset}$ with maximal degree $\degree^\star$,
        that has $\degree^\star$-relative-rank $\leq \rankfunc(c_1) + c_1 + 1$.
        This is a contradiction to our assumptions on $\bar{\factor}$, which completes the proof of the claim.
    \end{proof}
    This completes the proof of the lemma.
\end{proof}

%% file: distance_of_reed_muller_over_X.tex
\section[Radius of Reed-Muller over \titlevariety]{Radius of Reed-Muller over \titlevariety}\label{sec:radius-of-RM-over-X}
We recall that the normalized distances of Reed-Muller codes over $\field$ and over $\variety$
are denoted by $\normalizedcodedistance{\basefield}{\degree}$ and $\normalizedcodedistanceex{\basefield}{\variety}{\degree}$ respectively.
We present a theorem that shows that Reed-Muller codes over a subset $\variety \subseteq \field$ that is $\degree$-lift-enabler and has the (limited) relative rank-bias property,
has (approximately) an \emph{equal} normalized distance as Reed-Muller codes over $\field$.
\begin{theorem}\label{thm:distance-of-RM-in-X}
    There exist a function $\epsilon_1(\basefield, \degree, \rankbiasfunc, \epsilonlimitedrankbias)$  such that the following holds:
    Let $\basefield$ be a finite field, and let $\degree \in \naturalnumbersset$ be an integer that represents a degree.
    Let $\epsilonlimitedrankbias > 0$, and let $\funcdef{\rankbiasfunc}{[\epsilonlimitedrankbias, \infty]}{\naturalnumbersset}$ be a limited-relative-rank-bias function.
    \newline
    Let $\variety \subseteq \field$ be a set with the following properties
    \begin{enumerate}
        \item $\variety$ is $\degree$-lift-enabler with a lift operator $\lift{\square}$.
        \item $\variety$ has the $(\rankbiasfunc, \basefield, \degree, \epsilonlimitedrankbias)$-relative-rank-bias property.
    \end{enumerate}
    Then, for $\epsilon_1 \definedas \epsilon_1(\basefield, \degree, \rankbiasfunc, \epsilonlimitedrankbias)$ we have that for all $\blocklength \in \naturalnumbersset$:
    \[
        \normalizedcodedistanceex{\basefield}{\variety}{\degree} \geq \
        \normalizedcodedistance{\basefield}{\degree} - \epsilon_1
    \]
\end{theorem}
\begin{proof}
    We wish to do a reduction of our question regarding the radius of Reed-Muller in $\variety$ to the same question about Reed-Muller in $\field$.
    Let $\basefield$ be a finite field, and let $\degree \in \naturalnumbersset$ be an integer that represents a degree.
    Let $\epsilonlimitedrankbias > 0$, and let $\funcdef{\rankbiasfunc}{[\epsilonlimitedrankbias, \infty]}{\naturalnumbersset}$ be a limited-rank-relative-bias function.
    Let $\epsilon_1 \definedas \epsilon_1(\basefield, \degree, \rankbiasfunc, \epsilonlimitedrankbias)$ be a function we will specify later.
    Let $\variety \subseteq \field$ be a set with the properties defined above.
    \newline
    Moreover, let $\epsilon > \epsilon_1$ be some positive value.
    We will show that:
    \[
        \normalizedcodedistanceex{\basefield}{\variety}{\degree} > \
        \normalizedcodedistance{\basefield}{\degree} - \epsilon
    \]
    This will be enough as if the above holds for every $\epsilon > \epsilon_1$, we get that in fact
    $\normalizedcodedistanceex{\basefield}{\variety}{\degree} \geq \normalizedcodedistance{\basefield}{\degree} - \epsilon_1$.
    \newline
    For start, we note a simple observation: as Reed-Muller over $\variety$ is a linear code, we have
    \[
        \normalizedcodedistanceex{\basefield}{\variety}{\degree} =
        \min \set{\prex{x \in \variety}{\onvarpoly(x) \neq 0} \suchthat \onvarpoly \in \allpolyset{\leq \degree}{\variety}{\basefield}}
    \]
    Now, let $\onvarpoly \in \allpolyset{\leq \degree}{\variety}{\basefield}$ be a polynomial over $\variety$,
    and denote $\degree_{\onvarpoly} \definedas \deg(\onvarpoly)$.
    We wish to lower-bound the value of $\prex{x \in \variety}{\onvarpoly(x) \neq 0}$.
    To do so, we will equivalently upper-bound the value of $\prex{x \in \variety}{\onvarpoly(x) = 0}$.
    Precisely, to complete the proof all we need to show is:
    \[
        \prex{x \in \variety}{\onvarpoly(x) = 0} \leq 1 - \normalizedcodedistance{\basefield}{\degree} + \epsilon
    \]
    \newline
    Now we begin the proof itself.
    First, we lift the polynomial $\onvarpoly$ and get a polynomial $\funcdef{\lift{\onvarpoly}}{\field}{\basefield}$
    such that $\restrictfunc{\lift{\onvarpoly}}{\variety} \equiv \onvarpoly$ and $\deg(\lift{\onvarpoly}) = \degree_{\onvarpoly}$.
    Next, denote by $\factor_{\lift{\onvarpoly}}$ the factor defined by the set of single polynomial $\genpolyset = \set {\lift{\onvarpoly}}$.
    Trivially, the polynomial $\lift{\onvarpoly}$ is measurable in respect of $\genpolyset$.
    \newline
    We define the rank function:
    \[
        \rankfunc(m) \definedas \max \set{
            \rankbiasfunc \parens {\dfrac{\epsilon / 2}{\abs{\basefield}^m}},
            \rankfunc_{\ref{high-rank-implies-low-bias}} \parens{\basefield, \degree, \dfrac{\epsilon / 2}{\abs{\basefield}^m}}}
    \]
    Then, we $\rankfunc$-$\variety$-regularize $\genpolyset$ using Lemma~{\ref{theorem:regularization-in-X}}.
    This gives us a $\rankfunc$-$\variety$-regular factor $\factor^\prime$, which is defined by a set of polynomials $\genpolyset^{\prime} \definedas \set{\genpoly^\prime_1,...,\genpoly^\prime_{c^\prime}}$
    of degree $\leq \degree$ such that $\factor^\prime \relsemrefine{\variety}\factor_{\lift{\onvarpoly}}$ with $\relrank{\variety}{\genpolyset^\prime} \geq \rankfunc$
    and with bounded amount of polynomials defining it i.e,$c^\prime \leq C_{\rankfunc, \degree}^{\ref{theorem:regularization-in-X}}(1)$.
    Therefore, from definition we have that $\lift{\onvarpoly}$ is $\genpolyset^\prime$-measurable relative to $\variety$.
    Thus, there exists a measurement function $\funcdef{\Gamma}{\basefield^{c^\prime}}{\basefield}$
    and a remainder $\funcdef{\relativeremainder{\Gamma}}{\field}{\basefield}$ with $\restrictfunc{\relativeremainder{\Gamma}}{\variety} \equiv 0$
    and degree bounded by $\degree_{\onvarpoly}$, such that:
    \[
        \forall a \in \field:
        \lift{\onvarpoly}(a) =
        \Gamma(\genpoly^\prime_1(a),...,\genpoly^\prime_{c^\prime}(a))
        + \relativeremainder{\Gamma}(a)
    \]
    Next, we denote $\genpoly^\prime \definedas \lift{\onvarpoly} - \relativeremainder{\Gamma}$.
    By definition of remainder function, we have that $\restrictfunc{\genpoly^\prime}{\variety} \equiv \onvarpoly$.
    Additionally, note that $\genpoly^\prime$ is a polynomial over $\field$ of degree $\deg(\genpoly^\prime) = \degree_{\onvarpoly} \leq \degree$,
    and hence by the definition of $\normalizedcodedistance{\basefield}{\degree}$:
    \begin{equation}\label{eq:polynomials-in-fn-are-bounded-away-from-zero-with-high-probability}
    \prex{a \in \field}{\genpoly^\prime(a) = 0} \leq 1 - \normalizedcodedistance{\basefield}{\degree}
    \end{equation}
    For the next step, we claim that $\genpoly^\prime$ equals $0$ in $\field$ approximately with the same probability it equals $0$ in $\variety$.
    Note that this is the heart of the proof: it allows use properties known in $\field$ to new properties in $\variety$.
    This is formulated as follows:
    \begin{claim}\label{claim:p-and-P-have-the-same-approximation}
        We have:
        \[
            \abs{\prex{a \in \field}{\genpoly^\prime(a)= 0} -
            \prex{x \in \variety}{\genpoly^\prime(x) = 0}} \leq  \epsilon
        \]
    \end{claim}
    \begin{proof}
        Denote $S \definedas \basefield^{c^\prime}$, and for all $s \in S$, denote:
        \[
            p_1(s) \definedas \prex{a \in \field}{(\genpoly^{\prime}_1(a),...,\genpoly^{\prime}_{c^\prime}(a)) = s}
        \]
        As of our choice of $\rankfunc$, we have $\rank{\genpolyset^{\prime}} \geq \rankfunc_{\ref{high-rank-implies-low-bias}} \parens {\basefield, \degree, \dfrac{\epsilon/2}{\abs{\basefield}^{c^\prime}}}$.
        By combining Lemma~{\ref{high-rank-implies-low-bias}} with Lemma~{\ref{every-linear-combination-has-low-bias-implies-equidistribution}},
        we have that $p_1$ is ($\epsilon/2\abs{S}$)-equidistributed, i.e:
        \[
            p_1(s) = \dfrac{1 \pm \epsilon/2}{\abs{S}}
        \]
        Similarly, denote:
        \[
            p_2(s) \definedas \prex{x \in \variety}{(\genpoly^{\prime}_1(x),...,\genpoly^{\prime}_{c^\prime}(x)) = s}
        \]
        As of our choice of $\rankfunc$, we have $\relrank{\variety}{\genpolyset^\prime} \geq \rankbiasfunc(\epsilon / 2\abs{S})$.
        Now, we wish to use the relative rank-bias relation with Lemma~\ref{every-linear-combination-has-low-bias-implies-equidistribution}
        to conclude similarly that $p_2$ is ($\epsilon/2\abs{S}$)-equidistributed, i.e:
        \[
            p_2(s) = \dfrac{1 \pm \epsilon/2}{\abs{S}}
        \]
        However, in order to so, we must first ensure that $(\epsilon/2\abs{S}) \geq \epsilonlimitedrankbias$.
        This is done by choosing a correct $\epsilon_1$, and formulated in the following claim:
        \begin{claim}
            One can choose $\epsilon_1 \definedas \epsilon_1(\basefield, \degree, \rankbiasfunc, \epsilonlimitedrankbias)$ such that if $\epsilon \geq \epsilon_1$ we have that $\epsilon/2\abs{S} \geq \epsilon_1$.
        \end{claim}
        \begin{proof}
            We need that:
            \[
                \dfrac{\epsilon}{2 \abs{\basefield}^{c^\prime}} \geq \epsilonlimitedrankbias
            \]
            As $c^\prime \leq C^{\ref{theorem:regularization-in-X}}_{\rankfunc, \degree}(1)$,
            for the term above to hold it is enough that the following will be true:
            \[
                \epsilon \geq \epsilonlimitedrankbias \cdot 2 \abs{\basefield}^{C^{\ref{theorem:regularization-in-X}}_{\rankfunc, \degree}(1)}
            \]
            and as $\rankfunc$ and thus also $C^{\ref{theorem:regularization-in-X}}_{\rankfunc, \degree}(1)$ are independent of $\blocklength$,
            we can pick $\epsilon_1 = \epsilon_1(\basefield, \degree, \rankbiasfunc, \epsilonlimitedrankbias)$ and get what we aimed for.
        \end{proof}
        Now, under that assumption of $\epsilon_1$ written above, we have that $p_2$ is ($\epsilon/2\abs{S}$)-equidistributed.
        This allows us to use the similar distributions of $\genpolyset^\prime$ in $\field$ and in $\variety$ to conclude
        that $\genpoly^\prime$ behaves similar in $\field$ and in $\variety$:
        \begin{flalign*}
            \prex{a \in \field}{\genpoly^\prime(a)= 0}
            &=\sum_{s \in S} {p_1(s) \cdot \existfunc{\Gamma(s) = 0}} \\
            &=\sum_{s \in S} {p_2(s) \cdot \existfunc{\Gamma(s) = 0}} \pm \epsilon \\
            &=\prex{x \in \variety}{\genpoly^\prime(x)= 0} \pm \epsilon
        \end{flalign*}
        which concludes the proof of the claim.
    \end{proof}

    Finally, as $\restrictfunc{\genpoly^\prime}{\variety} \equiv \onvarpoly$, we have that $\prex{x \in \variety}{\genpoly^\prime(x) = 0} = \prex{x \in \variety}{\onvarpoly(x) = 0}$.
    Thus, the claim above combining with~\eqref{eq:polynomials-in-fn-are-bounded-away-from-zero-with-high-probability}
    shows that the probability we wished to bound is bounded as we aimed for:
    \[
        \prex{x \in \variety}{\onvarpoly(x) = 0} \leq 1 - \normalizedcodedistance{\basefield}{\degree} + \epsilon
    \]
    This concludes the proof of the theorem.
\end{proof}
\begin{remark}
    Under the same conditions,
    the distance of Reed-Muller codes in $\variety$ is also bounded \emph{from above} by the distance of Reed-Muller codes in $\field$,
    and we have:
    \[
        \normalizedcodedistanceex{\basefield}{\variety}{\degree} \leq \normalizedcodedistance{\basefield}{\degree} + \epsilon_1
    \]
    \begin{proof}
        Let $\funcdef{\genfunc}{\field}{\basefield}$ be the polynomial in $\field$ with the \emph{smallest} distance from $0$ as possible, that is $\normalizedcodedistance{\basefield}{\degree}$.
        Denote $\onvarpoly \definedas \restrictfunc{\genpoly}{\variety}$.
        Note that $\onvarpoly$ is a polynomial in $\variety$.
        Now repeat the proof using these two polynomials, and by Claim~\ref{claim:p-and-P-have-the-same-approximation}, we have that a random input of $\genpoly$ yields $0$
        (approximately) the same as a random input of $\onvarpoly$ yields $0$.
        Thus as we have $\prex{x \in \field}{\genpoly(x) = 0} = 1 - \normalizedcodedistance{\basefield}{\degree}$
        we also get:
        \[
            \prex{x \in \variety}{\onvarpoly(x) = 0} \geq 1 - \normalizedcodedistance{\basefield}{\degree} - \epsilon_1
        \]
        This bounds \emph{from above} the distance of Reed-Muller code in $\variety$ and we have:
        \[
            \normalizedcodedistanceex{\basefield}{\variety}{\degree} \leq \normalizedcodedistance{\basefield}{\degree} + \epsilon_1
        \]
    \end{proof}
\end{remark}
\begin{corollary}
    If we assume $\variety$ has the limited-relative rank-bias property to \emph{any extent} (or just the relative rank-bias property),
    then the theorem above proves an exact equality $\normalizedcodedistanceex{\basefield}{\variety}{\degree} = \normalizedcodedistance{\basefield}{\degree}$.
\end{corollary}

%% file: list_decoding_reed_muller_general_degrees.tex
\section[List Decoding Reed Muller Over \titlevariety]{List Decoding Reed Muller Over \titlevariety}\label{sec:list-decoding-reed-muller-over-X}
In this section, we prove our main theorem:
we prove the list decoding radius of Reed-Muller codes in is $\variety$ \emph{at least}
the list decoding radius of Reed-Muller codes in $\field$,
assuming $\variety$ is lift-enabler and has the relative rank-bias property.
We start by presenting formally the list decoding radius in $\variety$.
\begin{definition}[List Decoding in $\variety$]
    Let $\basefield$ be a finite field.
    Let $\degree, \blocklength \in \naturalnumbersset$, and let $\variety \subseteq \field$.
    \newline
    We define the reed muller list-decoding count in $\variety$ at distance $\tau$ as follows:
    \[
        \listpolycount{\basefield}{\variety}{\degree}{\tau} \definedas
        \max_{\funcdef{\genfunc}{\variety}{\basefield}}
            {\abs{\set{\genpoly \in \allpolyset{\leq \degree}{\variety}{\basefield} \suchthat {\dist{\genpoly, \genfunc} \leq \tau}}}}
    \]
    Additionally, we define $\listdecodingradiusex{\basefield}{\variety}{\degree}$ to be the \emph{list decoding radius}, which is
    the maximum $\tau$ for which $\listpolycount{\basefield}{\variety}{\degree}{\tau - \epsilon}$ is bounded by a \emph{constant} depending only on $\epsilon, \abs{\basefield}, \degree$.
\end{definition}

We recall that it was shown in~\cite[Theorem 1]{bhowmick2014list} that the list decoding radius of Reed Muller is $\normalizedcodedistance{\basefield}{\degree}$.
To be more precise, it was shown that for every $\epsilon > 0$, the list-decoding count is constant (independent of $\blocklength$) in distance $\tau = \normalizedcodedistance{\basefield}{\degree} - \epsilon$.
Formally, they have shown the following theorem:
\begin{theorem}[List Decoding RM in $\field$]\label{list-decoding-RM-in-Fn}
    There exists a function $c(\basefield, \degree, \epsilon)$ such that the following holds:
    Let $\basefield$ be a finite field, let $\epsilon > 0$, and let $\degree, \blocklength \in \naturalnumbersset$.
    Then, we have:
    \[
        \listpolycount{\basefield}{\field}{\degree}{\normalizedcodedistance{\basefield}{\degree} - \epsilon}
        \leq c(\basefield, \degree, \epsilon)
    \]
\end{theorem}
Additionally, we recall a lemma that was presented in~\cite[Corollary 3.3]{bhowmick2014list}, and was used in the analysis of the list decoding radius of Reed-Muller codes in $\field$:
\begin{lemma}[Low Complexity Approximation]~\cite[Corollary 3.3]{bhowmick2014list}\label{every-function-can-be-approximated-by-a-few-functions}
Let $\funcdef{\genfunc[2]}{A}{B}$, and let $\epsilon > 0$.
Let $\genfuncset \subseteq B^A$ be a collection of functions from $A$ to $B$.
Then there exists $c \leq 1/\epsilon^2$ functions $\genfunc_1,...,\genfunc_c \in \genfuncset$ such that
for every $\genfunc \in \genfuncset$, there is a function $\funcdef{\Gamma_{\genfunc}}{B^c}{B}$ such that:
\[
    \prex{x \in A}{\Gamma_{\genfunc}(\genfunc_1(x),...,\genfunc_c(x)) = \genfunc(x)}
    \geq \prex{x \in A}{\genfunc[2](x) = \genfunc(x)} - \epsilon
\]
\end{lemma}
The lemma shows that $\genfunc[2]$ can ``estimated'' by a only a few functions from $\genfuncset$.
Note that the estimation is close to $\genfunc[2]$ in compare to every $\genfunc \in \genfuncset$ and not necessarily close to $\genfunc[2]$ itself.

Finally, we present our main theorem, which shows that under assumptions on the subset $\variety \subseteq \field$, the list decoding radius of polynomials in $\variety$ will be similar to the list decoding radius in $\field$.
\newline
In more details (and informally), we show that if $\variety \subseteq \field$ is lift-enabler, has the limited-relative-rank-bias-property,
the list-decoding count is constant (independent of $\blocklength$) for every valid $\epsilon$ in distance $\tau = \normalizedcodedistance{\basefield}{\degree} - \epsilon$.
Note that not every $\epsilon > 0$ will be valid: the valid values of $\epsilon$ will depend on the limitations of the rank-bias property.
Formally, we show the following:
\begin{theorem}[List Decoding RM in $\variety$]\label{thm:list-decoding-RM-in-X}
    There exist functions $c_1(\basefield, \degree, \rankbiasfunc, \epsilonlimitedrankbias)$ and $c_2(\basefield, \degree, \rankbiasfunc, \epsilon)$ such that the following holds:
    Let $\basefield$ be a finite field, and let $\degree \in \naturalnumbersset$ be an integer that represents a degree.
    Let $\epsilonlimitedrankbias > 0$, and let $\funcdef{\rankbiasfunc}{[\epsilonlimitedrankbias, \infty]}{\naturalnumbersset}$ be a limited-relative-rank-bias function.
    \newline
    Let $\variety \subseteq \field$ be a set with the following properties
    \begin{enumerate}
        \item $\variety$ is $\degree$-lift-enabler with a lift operator $\lift{\square}$.
        \item $\variety$ has the $(\rankbiasfunc, \basefield, \degree, \epsilonlimitedrankbias)$-relative-rank-bias property.
    \end{enumerate}
    Then, for every $\epsilon \geq c_1(\basefield, \degree, \rankbiasfunc, \epsilonlimitedrankbias)$ it holds:
    \[
        \listpolycount{\basefield}{\variety}{\degree}{\normalizedcodedistanceex{\basefield}{\field}{\degree} - \epsilon} \leq
        c_2(\basefield, \degree, \rankbiasfunc, \epsilon)
    \]
\end{theorem}
\begin{proof}
    We follow the lines of the proof of~\cite[Theorem 1]{bhowmick2014list}.
    Let $\basefield$ be a finite field, and let $\degree \in \naturalnumbersset$ be an integer that represents a degree.
    Let $\epsilonlimitedrankbias > 0$, and let $\funcdef{\rankbiasfunc}{[\epsilonlimitedrankbias, \infty]}{\naturalnumbersset}$ be a limited-rank-relative-bias function.
    Let $c_1(\basefield, \degree, \rankbiasfunc, \epsilonlimitedrankbias)$ be a function we will specify later.
    Let $\variety \subseteq \field$ be a set with the properties defined above.
    \newline
    Finally, let $\epsilon \geq c_1(\basefield, \degree, \rankbiasfunc, \epsilonlimitedrankbias)$ for $c_1$ that we will specify later, and let $\funcdef{\onvarfunc}{\variety}{\basefield}$ be a received word.
    We wish to bound the amount of polynomials in $\allpolyset{\leq \degree}{\variety}{\basefield}$ that are $(\normalizedcodedistance{\basefield}{\degree} - \epsilon)$-close to $\onvarfunc$.
    \newline
    Apply Lemma~\ref{every-function-can-be-approximated-by-a-few-functions} with $A = \variety$, $B = \basefield$, $\genfunc[2] = \onvarfunc$, $\genfuncset = {\allpolyset{\leq \degree}{\variety}{\basefield}}$
    and approximation parameter $\epsilon / 2$ to obtain $\onvarpolyset[3] \subset \allpolyset{\leq \degree}{\variety}{\basefield}$, defined by $\onvarpolyset[3] = (\onvarpoly[3]_1,...,\onvarpoly[3]_c)$ where $c \leq 4/\epsilon^2$,
    such that for every $\onvarpoly \in \allpolyset{\leq \degree}{\variety}{\basefield}$ there is a function $\funcdef{\Gamma_{\onvarpoly}}{\basefield^c}{\basefield}$ that approximates $\onvarfunc$ in $\variety$ relative to $\allpolyset{\leq \degree}{\variety}{\basefield}$ i.e.:
    \[
        \forall \onvarpoly \in \allpolyset{\leq \degree}{\variety}{\basefield}: \prex{x \in \variety}{\Gamma_{\onvarpoly}(\onvarpoly[3]_1(x),...,\onvarpoly[3]_c(x)) = \onvarpoly(x)} \geq \prex{x \in \variety}{\onvarfunc(x) = \onvarpoly(x)} - \epsilon / 2
    \]
    \newline
    Let $\funcdef{\rankfunc_1, \rankfunc_2}{\naturalnumbersset}{\naturalnumbersset}$ be two non-decreasing functions that represents rank that we will specify later.
    For $\rankfunc_1$, we will require that for all $m \geq 1$:
    \[
        \rankfunc_1(m) \geq \max { \set{
        \topbot{
            {\rankfunc_2(C_{\rankfunc_2, \degree}^{\ref{theorem:regularization-in-X}}(m + 1)) + C_{\rankfunc_2, \degree}^{\ref{theorem:regularization-in-X}}(m + 1) + 1,}}
            {\rankfunc_2(C_{\rankfunc_{\ref{preserving-degree-starting-field}}, \degree}^{\ref{theorem:regularization-in-X}}(m + 1))
            + C_{\rankfunc_{\ref{preserving-degree-starting-field}}, \degree}^{\ref{theorem:regularization-in-X}}(m + 1) + 1}
        }}
    \]
    Note that in the expression above, we denote $\funcdef{\rankfunc_{\ref{preserving-degree-starting-field}}}{\naturalnumbersset}{\naturalnumbersset}$,
    as follows: $\rankfunc_{\ref{preserving-degree-starting-field}}(c) \definedas \rankfunc_{\ref{preserving-degree-starting-field}}(\basefield, \degree, c)$.
    \newline
    The reason we chose this $\rankfunc_1$, is that by our choice of $\rankfunc_1$ we can use the second part of Lemma~\ref{theorem:regularization-in-X}.
    Specifically, if we start with $\rankfunc_1$-$\variety$-regular factor and we $\rankfunc_2$-$\variety$-regularize it,
    we get that the $\rankfunc_2$-$\variety$-regular factor that we received is a syntactic refinement of the $\rankfunc_1$-$\variety$-regular factor we started with.
    \newline
    As a first step, we lift the polynomial factor to get $\genpolyset[3] \definedas \lift{\onvarpolyset[3]}$.
    Note that because $\forall x \in \variety: \lift{\onvarpoly[3]_i}(x) = \onvarpoly[3]_i(x)$, for all $\onvarpoly \in F$ we have:
    \[
        \prex{x \in \variety}{\Gamma_{\onvarpoly}(\lift{\onvarpoly[3]_1}(x),...,\lift{\onvarpoly[3]_c}(x)) = \onvarpoly(x)} \geq \prex{x \in \variety}{\onvarfunc(x) = \onvarpoly(x)} - \epsilon / 2
    \]
    Next, we $\rankfunc_1$-$\variety$-regularize the factor $\genpolyset[3]$ by Theorem~\ref{theorem:regularization-in-X}.
    This gives us a $\rankfunc_1$-$\variety$-regular factor $\factor^\prime$, which is defined by a set of polynomials $\genpolyset[3]^\prime \definedas (\genpoly[3]^\prime_1,...,\genpoly[3]^\prime_{c^\prime})$ of degree $\leq \degree$
    such that $\factor^\prime \relsemrefine{\variety} \factor$,
    with $\relrank{\variety}{\genpolyset[3]^\prime} \geq \rankfunc(c^\prime)$
    and with bounded amount of polynomials defining it i.e. $c^\prime \leq C_{\rankfunc_1, \degree}^{\ref{theorem:regularization-in-X}}(c)$.
    We apply Corollary~\ref{relative-semantic-refinement-is-restricted-semantic-refinement}
    and get that $\factor^{\prime} \semrefineex{\variety} \factor$.
    We then use the fact that $\Gamma_{\onvarpoly}(\lift{\onvarpoly[3]_1}(x),...,\lift{\onvarpoly[3]_c}(x))$ is measurable in respect of $\genpolyset$ \emph{in $\variety$},
    and deduce we have a similar approximation of $\onvarpoly$ using $\genpolyset^\prime$ as the approximation of $\onvarpoly$ using $\genpolyset$.
    Formally, there exists a function $\funcdef{\Gamma_{\onvarpoly}^\prime}{\basefield^{c^\prime}}{\basefield}$ such that:
    \[
        \prex{x \in \variety}
        {\Gamma^{\prime}_{\onvarpoly}(\genpoly[3]^{\prime}_1(x)),...,\genpoly[3]^{\prime}_{c^\prime}(x)) = \onvarpoly(x)} \geq
            \prex{x \in \variety}{\onvarfunc(x) = \onvarpoly(x)} - \epsilon / 2
    \]
    Now we recall that we wished to bound the amount of polynomials $\onvarpoly \in \allpolyset{\leq \degree}{\variety}{\basefield}$ such that
    $\prex{x \in \variety}{\onvarfunc(x) \neq \onvarpoly(x)} < \normalizedcodedistance{\basefield}{\degree} - \epsilon$.
    Let $\onvarpoly \in \allpolyset{\leq \degree}{\variety}{\basefield}$ be a polynomial as we just described.
    We will show that such $\onvarpoly$ is measurable with respect to $\genpolyset[3]^\prime$ in $\variety$.
    This will upper bound the amount of possible polynomials $\onvarpoly$ by the amount of possible different $\funcdef{\Gamma^{\prime}_{\onvarpoly}}{\basefield^{c^\prime}}{\basefield}$,
    which is $\abs{\basefield}^{\norm{\factor^\prime}} = \basefieldsize^{(\basefieldsize^{c^\prime})}$, and thus $c_2(\basefield, \degree, \rankbiasfunc, \epsilon) \leq \basefieldsize^{(\basefieldsize^{c^\prime})}$.
    \newline
    By our choice of $c^\prime$ we have that $c^\prime \leq C_{\rankfunc_1, \degree}^{\ref{theorem:regularization-in-X}}(4/\epsilon^2)$, and thus $c_2$ is bounded by a function of $(\basefield, \degree, \rankfunc_1, \epsilon)$.
    Note that we have not yet specified the value of $\rankfunc_1$, because it is determined by the choice of $\rankfunc_2$ that we will later define its exact values.
    The important thing about our future choice of $\rankfunc_2$ is that the value of $\rankfunc_2$ must be independent of $\blocklength$,
    but can depend on $(\basefield, \degree, \rankbiasfunc, \epsilon)$.
    This will conclude the proof.
    \newline
    Now, consider a lift of $\onvarpoly$, i.e $\genpoly \definedas \lift{\onvarpoly}$.
    Note that by the definition of lift $\forall x \in \variety: \genpoly(x) = \onvarpoly(x)$.
    We will show that $\genpoly$ is measurable in respect of $\genpolyset[3]^\prime$ in $\variety$.
    \newline
    We consider the factor $\factor_{\genpoly}$ that is generated by $\genpolyset[3]_{\genpoly} \definedas \genpolyset[3]^\prime \cup \set{\genpoly}$.
    By using Theorem~\ref{theorem:regularization-in-X}, we can $\rankfunc_2$-$\variety$-regularize it and get the polynomial factor $\factor^{\prime\prime}$ that relative-refines $\factor_{\genpoly}$.
    We denote the set of polynomials in the factor as $\genpolyset[3]^{\prime\prime}$.
    \newline
    Next, notice that the factor $\factor^{\prime\prime}$ is a $\rankfunc_2$-regular factor, therefore by our choice of $\rankfunc_1$ and the second part of Theorem~\ref{theorem:regularization-in-X},
    we in fact have $\factor^{\prime\prime} \synrefine \factor^{\prime}$.
    This is true because by our choice of $\rankfunc_1$:
    \[
        \relrank{\variety}{\genpolyset[3]^\prime} \geq
        \rankfunc_1(c^\prime) \geq
        \rankfunc_2(C_{\rankfunc_2, \degree}^{\ref{theorem:regularization-in-X}}(c^\prime + 1)) + C_{\rankfunc_2, \degree}^{\ref{theorem:regularization-in-X}}(c^\prime + 1) + 1 \geq
        \rankfunc_2(\abs{\factor^{\prime\prime}}) + \abs{\factor^{\prime\prime}}+1
    \]
    And as rank is always bigger than relative rank, we also have:
    \[
        \rank{\genpolyset[3]^\prime} \geq
        \rankfunc_1(c^\prime) \geq
        \rankfunc_2(C_{\rankfunc_{\ref{preserving-degree-starting-field}}, \degree}^{\ref{theorem:regularization-in-X}}(c^\prime + 1))
        + C_{\rankfunc_{\ref{preserving-degree-starting-field}}, \degree}^{\ref{theorem:regularization-in-X}}(c^\prime + 1) + 1
    \]
    Thus, the polynomials defining $\factor^{\prime\prime}$ are in the form $\genpolyset[3]^{\prime\prime} \definedas \genpolyset[3]^{\prime} \cup \set{\genpoly[3]^{\prime\prime}_1,...,\genpoly[3]^{\prime\prime}_{c^{\prime\prime}}}$.
    Note that as promised in Theorem~{\ref{theorem:regularization-in-X}}, we have $\abs{\genpolyset[3]^{\prime\prime}} = c^\prime+c^{\prime\prime} \leq C^{\ref{theorem:regularization-in-X}_{\rankfunc_2, \degree}}(c^\prime)$.
    \newline
    Additionally, by the way we built $\genpolyset[3]_{\genpoly}$, the function $\genpoly$ is measurable in respect of it.
    Therefore, as $\factor^{\prime\prime} \relsemrefine{\variety} \factor_{\genpoly}$, we have that $\genpoly$ is $\genpolyset[3]^{\prime\prime}$-measurable relative to $\variety$.
    In other words, there exists $\funcdef{\Phi}{\basefield^{c^\prime + c^{\prime\prime}}}{\basefield}$
    and $\funcdef{\relativeremainder{\genpoly}}{\field}{\basefield}$ with $\deg(\relativeremainder{\genpoly}), \deg(\genpoly - \relativeremainder{\genpoly}) \leq \deg(\genpoly) \leq \degree$ and $\restrictfunc{\relativeremainder{\genpoly}}{\variety} \equiv 0$ such that:
    \[
        \forall a \in \field: \genpoly(a) = \Phi(\genpoly[3]^\prime_1(a),...,\genpoly[3]^\prime_{c^\prime}(a), \genpoly[3]^{\prime\prime}_1(a),...,\genpoly[3]^{\prime\prime}_{c^{\prime\prime}}(a))) + \relativeremainder{\genpoly}(a)
    \]
    And specifically in $\variety$ we have:
    \[
        \forall x \in \variety: \genpoly(x) = \Phi(\genpoly[3]^\prime_1(x),...,\genpoly[3]^\prime_{c^\prime}(x), \genpoly[3]^{\prime\prime}_1(x),...,\genpoly[3]^{\prime\prime}_{c^{\prime\prime}}(x)))
    \]
    Denote $\genpoly^{\prime} \definedas \genpoly - \relativeremainder{\genpoly}$.
    We will show the polynomial $\genpoly^{\prime}$ does not depend on its last $c^{\prime\prime}$ variables, and thus $\Phi$ does not depend on its last $c^{\prime\prime}$ variables.
    This will imply that $\genpoly$ is measurable in respect of $\genpolyset[3]^\prime$ in $\variety$, which will conclude the proof.
    \newline
    Now, we choose $\rankfunc_2$ to be such that:
    \[
        \rankfunc_2(m) \geq \max \set {
            \rankbiasfunc \parens{\dfrac{\epsilon / 4}{\abs{\basefield}^m}},
            \rankval_{\ref{high-rank-implies-low-bias}}\parens{\dfrac{\epsilon / 4}{\abs{\basefield}^m}},
            \rankval_{\ref{preserving-degree-starting-field}}(m)
        }
    \]
    Note that in the expression above we are discussing fixed field and degree, i.e. $\basefield, \degree$.
    Therefore we denote $\funcdef{\rankfunc_{\ref{preserving-degree-starting-field}}}{\naturalnumbersset}{\naturalnumbersset}$
    as $\rankfunc_{\ref{preserving-degree-starting-field}}(c) \definedas \rankfunc_{\ref{preserving-degree-starting-field}}(\basefield, \degree, c)$
    and $\funcdef{\rankfunc_{\ref{high-rank-implies-low-bias}}}{\naturalnumbersset}{\naturalnumbersset}$
    as $\rankfunc_{\ref{high-rank-implies-low-bias}}(\epsilon) \definedas \rankfunc_{\ref{high-rank-implies-low-bias}}(\basefield, \degree, \epsilon)$.
    \newline
    Next, we show that even if we change the polynomials in the factor to have a disjoint set of inputs in $\field$,
    we still obtain a polynomial in the same degree, which have an approximation close to the approximation we had in $\variety$.
    Note that after this step, the proof becomes very similar to the proof of list decoding Reed Muller in $\field$~\cite[Theorem 1]{bhowmick2014list}:
    we omit the dependence of $\variety$ and get the same approximation by functions of multiple variables,
    as we had in $\field$.
    This is done by the following lemma:
    \begin{lemma}
        Let $\set{a^{i}, b^{j}}, i \in [c^\prime], j \in [c^{\prime\prime}]$  be pairwise disjoint sets of $\blocklength$ variables each.
        Let $\blocklength^{\prime} \definedas \blocklength(c^\prime + c^{\prime\prime})$.
        Let $\funcdef{\vec{\genpoly}}{\basefield^{\blocklength^\prime}}{\basefield}$ and $\funcdef{\vec{\onvarfunc}}{\basefield^{\blocklength^{\prime}}}{\basefield}$
        be functions of $\blocklength^\prime$ variables defined as follows:
        \[
            \vec{\genpoly^{\prime}}(\vec{a}) \definedas
                \Phi \parens{\genpoly[3]^\prime_1(a^1),...,\genpoly[3]^\prime_{c^\prime}(a^{c^\prime}), \genpoly[3]^{\prime\prime}_1(b^{1}),...,\genpoly[3]^{\prime\prime}_{c^{\prime\prime}}(b^{c^{\prime\prime}})}
        \]
        and:
        \[
            \vec{\onvarfunc}(\vec{a}) \definedas \Gamma^{\prime}_{\onvarpoly}(\genpoly[3]^{\prime}_1(a^{1})),...,\genpoly[3]^{\prime}_{c^\prime}(a^{c^\prime}))
        \]
        Note that $\vec{\onvarfunc}$ is a function that receives $\blocklength^\prime$ variables, and ignores its last $c^{\prime\prime}$ variables.
        \newline
        Then:
        \begin{enumerate}
            \item The degree of $\vec{\genpoly^\prime}$ remains bounded, i.e. $\deg(\vec{\genpoly^\prime}) \leq \degree$.
            \item The approximation of $\vec{\onvarfunc}$ to $\vec{\genpoly^\prime}$ in $\basefield^{n^\prime}$ is close to the approximation of $\Gamma^{\prime}_{\onvarpoly}$ to $\onvarpoly$ in $\variety$.
            Specifically, we show:
            \[
                \abs{
                    \prex{\vec{a} \in \basefield^{\blocklength^\prime}}{\vec{\onvarfunc}(\vec{a}) = \vec{\genpoly^\prime}(\vec{a})} -
                    \prex{x \in \variety}{\Gamma^{\prime}_{\onvarpoly}(\genpoly[3]^{\prime}_1(x)),...,\genpoly[3]^{\prime}_{c^\prime}(x)) =\onvarpoly(x)}
                }
                \leq \epsilon/4
            \]
        \end{enumerate}
        \begin{proof}
            We start by proving the first part of the lemma: bounding the degree of $\vec{\genpoly^\prime}$ by $\degree$.
            First, we recall that $\genpoly^{\prime} = \genpoly - \relativeremainder{\genpoly}$ where $\relativeremainder{\genpoly}$ is a valid remainder.
            Specifically, we have $\deg(\genpoly^{\prime}) = \deg(\genpoly - \relativeremainder{\genpoly}) \leq \deg(\genpoly) \leq \degree$.
            In addition, by the way we built $\Phi$ we have:
            \[
                \forall a \in \field: \genpoly^\prime(a) = \Phi(\genpoly[3]^\prime_1(a),...,\genpoly[3]^\prime_{c^\prime}(a), \genpoly[3]^{\prime\prime}_1(a),...,\genpoly[3]^{\prime\prime}_{c^{\prime\prime}}(a)))
            \]
            Thus the function above is of degree $\leq \degree$.
            Moreover, we have:
            \[
                \rank{\genpolyset[3]^{\prime\prime}} \geq
                \relrank{\variety}{\genpolyset[3]^{\prime\prime}} \geq
                \rankfunc_2(\abs{\genpolyset[3]^{\prime\prime}})\geq
                \rankfunc_{\ref{preserving-degree-starting-field}}(\abs{\genpolyset[3]^{\prime\prime}})
            \]
            Therefore we can use Lemma~\ref{preserving-degree-starting-field} to get that $\deg(\vec{\genpoly^{\prime}}) \leq \deg(\genpoly^\prime) \leq \degree$.
            Note that in order to use the lemma formally,
            we had to extend the input space of $\genpoly^{\prime}$ to be of $\blocklength^{\prime}$ variables (and make it depend only on the first $\blocklength$ variables as it used to).
            Because lemma~\ref{preserving-degree-starting-field} require bounds independent of $\blocklength$, this is done smoothly.
            \newline
            Now we move to the second part of the lemma: bounding the approximation of $\vec{\onvarfunc}$ to $\vec{\genpoly^\prime}$.
            Denote $S \definedas \basefield^{c^\prime + c^{\prime\prime}}$, and for each $s \in S$ denote:
            \[
                p_1(s) \definedas \prex{x \in \variety}{
                \parens{\genpoly[3]^\prime_1(x),...,\genpoly[3]^\prime_{c^\prime}(x), \genpoly[3]^{\prime\prime}_1(x),...,\genpoly[3]^{\prime\prime}_{c^{\prime\prime}}(x)} = s}
            \]
            and as of our choice of $\rankfunc_2$, we have $\rank{\genpolyset[3]^{\prime\prime}} \geq \rankbiasfunc(\epsilon/8\abs{S})$.
            Therefore, if we require that the relative rank-bias relation holds for $\epsilon / 8\abs{S}$, we can use Lemma~\ref{every-linear-combination-has-low-bias-implies-equidistribution} with $A = \variety$
            to get that $p_1$ is $(\epsilon /8\abs{S})$-almost uniform, i.e:
            \[
                p_1(s) = \dfrac{1 \pm \epsilon / 8}{\abs{S}}
            \]
            We show that this can be done in the following claim by choosing a proper $c_1$:
            \begin{claim}
                One can choose $c_1 \definedas c_1(\basefield, \degree, \rankbiasfunc, \epsilonlimitedrankbias)$ such that if $\epsilon \geq c_1$
                we have that $\epsilon / 8\abs{S} \geq c_1$.
            \end{claim}
            \begin{proof}
                This is done by using the bound we already know.
                We need that:
                \[
                    \epsilonlimitedrankbias \leq \dfrac{\epsilon}{8 \abs{\basefield}^{c^\prime + c^{\prime\prime}}}
                \]
                As $c^\prime + c^{\prime\prime} \leq C^{\ref{theorem:regularization-in-X}}_{\rankfunc_2, \degree}(c^\prime)$,
                for the term above to hold it is enough that the following will be true:
                \[
                    \epsilon \geq \epsilonlimitedrankbias \cdot 8 \abs{\basefield}^{C^{\ref{theorem:regularization-in-X}}_{\rankfunc_2, \degree}(c^\prime)}
                \]
                and as $\rankfunc_2, c^\prime$ and thus also $C^{\ref{theorem:regularization-in-X}}_{\rankfunc_2, \degree}(c^\prime)$ are independent of $\blocklength$,
                we can pick $c_1 = c_1(\basefield, \degree, \rankbiasfunc, \epsilonlimitedrankbias)$ and get what we aimed for.
            \end{proof}
            Thus, we can assume that $p_1$ is $(\epsilon/8\abs{S})$-almost uniform.
            Now, let:
            \[
                p_2(s) \definedas \prex{\vec{a} \in \basefield^{\blocklength^\prime}}
                {\parens{\genpoly[3]^\prime_1(a^1),...,\genpoly[3]^\prime_{c^\prime}(a^{c^\prime}), \genpoly[3]^{\prime\prime}_1(b^{1}),...,\genpoly[3]^{\prime\prime}_{c^{\prime\prime}}(b^{c^{\prime\prime}})} = s}
            \]
            Note that the rank of $\vec{\genpolyset[3]}^{\prime\prime} = {\set{\genpoly[3]^\prime_1(a^1),...,\genpoly[3]^\prime_{c^\prime}(a^{c^\prime}), \genpoly[3]^{\prime\prime}_1(b^{1}),...,\genpoly[3]^{\prime\prime}_{c^{\prime\prime}}(b^{c^{\prime\prime}})}}$,
            as a factor defined over $\basefield^{\blocklength^\prime}$, can not be lower than the rank of $\genpolyset[3]^{\prime\prime}$
            and thus we have $\rank{\vec{\genpolyset[3]}^{\prime\prime}} \geq \rankval_{\ref{high-rank-implies-low-bias}}\parens{\dfrac{\epsilon / 8}{\abs{\basefield}^m}}$.
            By using lemma~{\ref{high-rank-implies-low-bias}}, which shows the rank-bias relation for $\basefield^{\blocklength^\prime}$,
            we can similarly use Lemma~\ref{every-linear-combination-has-low-bias-implies-equidistribution} with $A = \basefield^{\blocklength^\prime}$
            to get that $p_2$ is also $(\epsilon/8\abs{S})$-almost-uniform, i.e:
            \[
                p_2(s) = \dfrac{1 \pm \epsilon / 8}{\abs{S}}
            \]
            Now, we show the approximations are the same.
            Denote by $s^\prime$ the restriction of $s$ to its first $c^\prime$ coordinates, and consider the approximation:
            \begin{flalign*}
                \prex{\vec{a} \in \basefield^{\blocklength^\prime}}{\vec{\onvarfunc}(\vec{a}) = \vec{\genpoly}^\prime (\vec{a})} = \\
                &=\sum_{s \in S} {p_2(s) \cdot \existfunc{\Phi(s) = \Gamma_{\genpoly}^\prime (s^\prime)}} \\
                &=\sum_{s \in S} {p_1(s) \cdot \existfunc{\Phi(s) = \Gamma_{\genpoly}^\prime (s^\prime)}} \pm \epsilon / 4 \\
                &=\prex{x \in \variety}{\Gamma^{\prime}_{\onvarpoly}(\genpoly[3]^{\prime}_1(x)),...,\genpoly[3]^{\prime}_{c^\prime}(x)) =\onvarpoly(x)} \pm \epsilon/4
            \end{flalign*}
            This completes the proof the lemma.
        \end{proof}

        The proof is followed by the same methods used in~\cite{bhowmick2014list}.
        We repeat if for completeness.
        We next restate a lemma proved in~\cite[Claim 4.2]{bhowmick2014list}, which is a varaiant of the Schwartz-Zippel lemma~\cite{10.1145/322217.322225,Zippel1979ProbabilisticAF}:
        \begin{lemma}\label{lemma-schwarz-zippel-for-comparing-polynomial-to-function-with-less-variables}
            Let $\degree$, $\blocklength_1$, $\blocklength_2 \in \naturalnumbersset$ be integers.
            Let $\genpoly_1 \in \allpolyset{\leq \degree}{\basefield^{\blocklength_1 + \blocklength_2}}{\basefield}$,
            and let $\funcdef{\genfunc_1}{\basefield^{\blocklength_1}}{\basefield}$ be a function.
            Assume the polynomial is $\normalizedcodedistance{\basefield}{\degree}$-close to the function, i.e:
            \[
                \prex{x_1,...,x_{\blocklength_1+\blocklength_2} \in \basefield}
                        {\genpoly_1(x_1,...,x_{\blocklength_1+\blocklength_2}) = \genfunc_1(x_1,...,x_n)} > 1 - \normalizedcodedistance{\basefield}{\degree}
            \]
            Then, $\genpoly_1$ does not depend on $x_{\blocklength_1 + 1},...,x_{\blocklength_1 + \blocklength_2}$.
        \end{lemma}
        Now, apply Lemma~\ref{lemma-schwarz-zippel-for-comparing-polynomial-to-function-with-less-variables} to
        $\genpoly_1 = \vec{\genpoly^{\prime}}$, $\genfunc_1 = \vec{\onvarfunc}$, $\blocklength_1 = \blocklength c^\prime$, $\blocklength_2 = \blocklength c^{\prime\prime}$.
        We obtain that $\vec{\genpoly^{\prime}}$ does not depend on its last $c^{\prime\prime}$ variables, and thus by denoting $C_{i} \definedas \genpoly[3]^{\prime\prime}_i(0)$ for $i \in \sparens{c^{\prime\prime}}$ we have:
        \[
            \vec{\genpoly^{\prime}}(\vec{a}) = \Phi \parens{\genpoly[3]^\prime_1(a^1),...,\genpoly[3]^\prime_{c^\prime}(a^{c^\prime}), C_1,...,C_{c^{\prime\prime}}}
        \]
        Now, for every $a \in \field$, if we substitute $a$ in the $i$-th component of $\vec{a}$ for every $i \in \sparens{c^{\prime}}$ in the equation above, we get the following is true:
        \[
            \genpoly^{\prime}(a) = \Phi \parens{\genpoly[3]^\prime_1(a),...,\genpoly[3]^\prime_{c^\prime}(a), C_1,...,C_{c^{\prime\prime}}}
        \]
        Hence $\genpoly^{\prime}$ does not depend on its last $c^{\prime\prime}$ variables.
        As explained earlier, this implies that $\genpoly$ is measurable in respect of $\genpolyset[3]^{\prime}$ in $\variety$.
        This completes the proof of the theorem.
    \end{lemma}

\end{proof}

%% file: equidistribution.tex
\section[Equidistribution of Functions]{Equidistribution of Functions}\label{sec:equidistribution-of-functions}
Assume we have a collection of functions $(\genfunc_1,...,\genfunc_c$), where $\funcdef{\genfunc_i}{A}{\basefield}$ for some finite set $A$.
We are interested in showing that the functions are equidistributed, which means that their values behave close to independent random variables.
We begin by formulating this definition:
\begin{definition}[Equidistribution of Functions]
    Given $\epsilon > 0$ and $A \subseteq \field$,
    we say a collection of functions $\genfuncset = (\genfunc_1,...,\genfunc_c)$ where $\funcdef{\genfunc_i}{\genset}{\basefield}$ is $\epsilon$-equidistributed in $A$ if for all $\vec{\alpha} = (\alpha_1,...,\alpha_c) \in \basefield^c$ we have:
    \[
        \prex{x \in A}{(\genfunc_1(x),...,\genfunc_c(x)) = \vec{\alpha}} = \frac{1}{\abs{\basefield}^c} \pm \epsilon
    \]
\end{definition}

The following is a standard lemma that shows that if every linear combination of a collection of functions has low bias, the collection is equidistributed.
We repeat the steps of the proof of \cite[Lemma 7.24]{book}, but here, we think of $A$ as any finite set (and not particularly $\field$):
\begin{lemma}\label{every-linear-combination-has-low-bias-implies-equidistribution}
Let $\epsilon > 0$, and let $A$ be a finite set.
Let $\genfuncset = (\genfunc_1,...,\genfunc_c)$ be a collection of functions defined over $A$, i.e. $\funcdef{\genfunc_i}{A}{\basefield}$.
Assume each linear combination of the collection has low bias, i.e for each $\lambda = (\lambda_1,...,\lambda_c) \in \basefield^c$ such that $\lambda \neq \vec{0}$ we have:
\[
    \relbias{x \in A}{\sum_{i=1}^{c}{\lambda_i \genfunc_i}} < \epsilon
\]
Then, the collection $\genfuncset$ is $\epsilon$-equidistributed over $A$.
\newline
In particular, for $\epsilon < \frac{1}{\abs{\basefield}^c}$, the lemma shows that each atom of $\genfuncset$ is not empty i.e for all $\vec{\alpha}$ there is some $x \in A$ such that $(\genfunc_1(x),...,\genfunc_c(x)) = \vec{\alpha}$.
\end{lemma}
\begin{proof}
    We wish to show that for each $\vec{\alpha} \in \basefield^c$ we have:
    \[
        \prex{x \in A}{(\genfunc_1(x),...,\genfunc_c(x)) = \vec{\alpha}} = \frac{1}{\abs{\basefield}} \pm \epsilon
    \]
    We express the fraction of inputs that are in the atom $\vec{\alpha}$ the following way:
    \[
        \prex{x \in A}{(\genfunc_1(x),...,\genfunc_c(x))} =
        \expectation{x \in A}{\prod_{i=1}^c {1_{[\genfunc_i(x) = \alpha_i]}}}
    \]
    We use the fact that for every $0 \neq x \in \basefield$, we have $\sum_{\lambda = 0}^{\basefieldsize - 1} \charfunc{\lambda x} = 0$,
    and if $x = 0$ we have $\sum_{\lambda = 0}^{\basefieldsize - 1} \charfunc{\lambda x} = \basefieldsize$.
    Therefore, the expression above equals:
    \[
        =\expectation{x \in A}{\prod_{i=1}^c \parens {{\frac{1}{\basefieldsize} \cdot \sum_{\lambda_i = 0}^{\basefieldsize - 1} {\charfunc{\lambda_i (\genfunc_i(x) - \alpha_i)}}}}} =\\
        \frac{1}{\basefieldsize^c} \cdot {\expectation{x \in A}{\prod_{i=1}^c \sum_{\lambda_i = 0}^{\basefieldsize - 1} {\charfunc{\lambda_i (\genfunc_i(x) - \alpha_i)}}}}
    \]
    By the definition of character functions, we have that $\charfunc{a+b} = \charfunc{a} \cdot \charfunc{b}$, and therefore the expression above equals:
    \[
        \frac{1}{\basefieldsize^c} \cdot \sum_{(\lambda_1,...,\lambda_c) \in \prod_{i=1}^c [0, \basefieldsize - 1]} \parens{\expectation{x \in A}{\charfunc{\sum_{i = 0}^{c} {\lambda_i (\genfunc_i(x) - \alpha_i)}}}}
    \]
    Now, we use the fact that:
    \[
        \relbias{x \in A}{\sum_{i=1}^c (\lambda_i (\genfunc_i(x) -\alpha_i)} = \relbias{x \in A}{\sum_{i=1}^c (\lambda_i \genfunc_i(x))} < \epsilon
    \]
    and get that:
    \[
        \prex{x \in A}{(\genfunc_1(x),...,\genfunc_c(x)) = \vec{\alpha}} = \frac{1}{\basefieldsize^c} \cdot \parens{1 \pm \epsilon \prod_{i=1}^c{\basefieldsize}} = \frac{1}{\abs{\basefield}^c} \pm \epsilon
    \]
\end{proof}

%
%

%% file: comparing_ranks.tex
\section[Comparing Ranks]{Comparing Ranks}\label{sec:comparing-ranks}
In this section, we compare the definition of rank we used in this paper to another definition of rank used implicitly throughout this paper.
This comparison is crucial, as there is no universally accepted definition of rank;
different theorems presented throughout this paper employ distinct definitions.
We demonstrate that our definition is sufficiently comprehensive, in that a polynomial (or a factor) classified as having high rank according to our criteria
also exhibits high rank according to the second implicitly-used definition.
While in many cases the comparison may appear straightforward, we include it for the sake of completeness.
\newline
Specifically, we compare our definition of rank with the definition established in ~\cite{lampert2021relative}.
The paper~\cite{lampert2021relative} extended the original definition of rank that was presented in \cite{10.1007/BF02392473},
to include also the concept of relative rank.
It is important to note that this definition is specifically defined to subsets $\variety \subseteq \field$ that can be expressed as sets in the form $\variety = \zerofunc{\varpolyset}$ for some set of polynomials $\varpolyset$,
and not to a general set $\variety \subseteq \field$.
\newline
First, we present a useful notation that is used in the definition presented in ~\cite{lampert2021relative}:
\begin{notation*}[Largest Degree Homogenous Part]
    For a polynomial $\genpoly$ of degree $\degree$, we denote by $\homopart{\genpoly}$ its degree-$\degree$ homogenous component.
    In other words, $\homopart{\genpoly}$ is the sum of all the monomials of $\genpoly$ of degree exactly $\degree$.
    For a set of polynomials $\genpolyset = \set{\genpoly_1,...,\genpoly_c}$, we define $\homopart{\genpolyset} \definedas \set{\homopart{\genpoly_i} \suchthat i = 1,...,c}$.
\end{notation*}
Next, we present the exact definition of rank for a polynomial:
\begin{definition}[Schmidt Rank of a Polynomial]
    The schmidt rank of a homogenous polynomial $\funcdef{\genpoly[1]}{\field}{\basefield}$, noted as $\schmrank{\genpoly[1]}$, is the minimal $r$ such that there exist $(\genpoly[2]_i, \genpoly[3]_i)_{i\in[r]}$
    with $\deg{\genpoly[2]_i}, \deg{\genpoly[3]_i} < \deg{\genpoly[1]}$ such that:
    \[
        \genpoly[1](x) = \sum_{i=1}^{r}(\genpoly[2](x) \cdot \genpoly[3](x))
    \]
    For a general polynomial $\genpoly$ of degree $\degree$, we set its rank to be the rank of its degree-$\degree$ homogenous component, i.e. $\schmrank{\genpoly} \definedas \schmrank{\homopart{\genpoly}}$.
\end{definition}

\begin{remark}[High rank implies high schmidt rank]\label{high-rank-implies-high-schmidt-rank}
If $\rank{\genpoly} \geq 2 \cdot r + 1$ for some constant $r \in \naturalnumbersset$, then $\schmrank{\genpoly} \geq r$.
\end{remark}
\begin{proof}
    For homogenous polynomial $\genpoly$, assume $\schmrank{\genpoly} < r$.
    Then, there exist $r^\prime < r$ such that there exist $(\genpoly[2]_i, \genpoly[3]_i)_{i=1}^{r^\prime}$ with $\deg{\genpoly[2]_i}, \deg{\genpoly[3]}_i < \deg \genpoly$ such that:
    \[
        \genpoly[1](x) = \sum_{i=1}^{r^\prime}(\genpoly[2](x) \cdot \genpoly[3](x))
    \]
    Then we can choose $\funcdef{\Gamma}{\basefield^{2r^\prime}}{\basefield}$ to be a sum of multiples of each two consecutive variables to get that
    $\genpoly(x) = \Gamma(\genpoly[2]_1(x), \genpoly[3]_1(x),...,\genpoly[2]_{r^\prime}(x), \genpoly[3]_{r^\prime}(x))$, where the polynomials are from a degree $< \deg(\genpoly)$.
    This means that $\rank{\genpoly} \leq 2r^{\prime} < 2r$ as we requested.
    \newline
    If we do not assume $\genpoly$ is homogenous, by adding $\genpoly - \homopart{\genpoly}$ as an input to $\Gamma$,
    one can create a $\funcdef{\Gamma^\prime}{\basefield^{2r^\prime+1}}{\basefield}$
    which equals to $\genpoly$ when substituting the inputs with some polynomials with degree $< \deg{\genpoly}$,
    which concludes the proof in a similar way.
\end{proof}

Next, we present the definition of Schmidt rank of a factor as defined in ~\cite{lampert2021relative}.
\begin{definition}[Schmidt Rank of a Factor]
    For a factor of homogenous polynomials $\genpolyset = (\genpoly_1,...,\genpoly_{c})$, the schmidt rank of the factor is defined as:
    \[
        \schmrank{\genpolyset} \definedas \min \parens{\schmrank{\sum_{i=1}^{c}\lambda_i \genpoly_i} \suchthat 0 \neq (\lambda_1,...,\lambda_c) \in \basefield^c}
    \]
    Similarly, for a factor of general polynomials $\genpolyset$, we set its rank to be the rank of its matching homogenous-factor,
    i.e. $\schmrank{\genpolyset} \definedas \schmrank{\homopart{\genpolyset}}$
    For a factor $\factor$ generated by $\genpolyset$, we define $\schmrank{\factor} \definedas \schmrank{\genpolyset}$.
\end{definition}

To establish the equivalence of this definition with the one employed throughout the paper, we must first acknowledge two key distinctions between the definitions.
The first distinction is that this definition focuses on the largest-degree homogeneous components of the polynomials involved in the factor, rather than considering linear combinations of polynomials from the factor.
The second distinction pertains to the treatment of $\degree$ in the computation of $\degree$-rank of each linear combination.
This definition uses the degree of the linear combination directly to calculate the rank that participates in the minimum, in contrast to our definition which uses $\max_{i}{\deg(\lambda_i \genpoly_i)}$.
Despite these differences, we will demonstrate that both definitions ultimately yield a similar rank assessment, thereby affirming their equivalence.
\begin{remark}[High Rank Implies High Schmidt Rank for Factors]
    Let $\genpolyset = \parens{\genpoly_1,...,\genpoly_c}$ be a set of polynomials and let $\rankval \in \naturalnumbersset$ be a positive integer, i.e. $\rankval > 0$.
    If $\rank{\genpolyset} \geq 2 \cdot \rankval + 1$, then $\schmrank{\genpolyset} \geq \rankval$.
    \begin{proof}
        Assume that $\schmrank{\genpolyset} \leq \rankval$ for $\rankval > 0$.
        We will show that $\rank{\genpolyset} \leq 2 \rankval + 1$.
        By definition, there exists a linear combination of polynomials in $\homopart{\genpolyset}$ with rank $\leq \rankval$.
        In other words, there exists $\vec{0} \neq \lambda \in \basefield^{c}$ such that $\schmrank{\sum_{i=1}^{c} {\lambda_i \homopart{\genpoly_i}}} \leq \rankval$.
        Denote $\vec{\genpoly_h} \definedas \sum_{i=1}^{c} {\lambda_i \homopart{\genpoly_i}} $.
        As was shown in a previous remark, a rank of a polynomial is smaller than its schmidt rank up to a constant factor,
        thus $\rank{\vec{\genpoly_h}} \leq 2\rankval + 1$ (see Remark~\ref{high-rank-implies-high-schmidt-rank}).
        \newline
        Next, we denote $\vec{\genpoly} \definedas \sum_{i=1}^{c} {\lambda_i \genpoly_i}$, and $\degree_M \definedas \max_{i \in \sparens{c}}{\lambda_i \genpoly_i}$.
        Note that  $\deg(\vec{\genpoly}) \leq \degree_M$.
        We wish to show that $\drank{\degree_M}{\vec{\genpoly}} \leq 2 \rankval + 1$.
        First, we observe that the $\degree_M$-degree homogenous component of $\vec{\genpoly_h}$ equals the $\degree_M$-degree homogenous component of $\vec{\genpoly}$.
        This is true because every highest-degree component of polynomials in the linear combination that generated $\vec{\genpoly}$,
        also exists in the linear combination that generates $\vec{\genpoly_h}$.
        In particular, all homogenous components of degree $\degree_M$ exists in both linear combinations $\vec{\genpoly_h}$ and $\vec{\genpoly}$.
        Therefore, if the degree of $\vec{\genpoly}$ equals $\degree_M$, we have that $\drank{\degree_M}{\vec{\genpoly}} = \rank{\vec{\genpoly}} \geq 2 \rankval + 1$.
        Otherwise, if $\deg(\vec{\genpoly}) < \degree_M$, then $\drank{\degree_M}{\vec{\genpoly}} = 1 \leq 2 \rankval + 1$.
        This completes the proof.
        \newline
    \end{proof}
\end{remark}
\begin{note*}
    In the case discussed above, if $\deg(\vec{\genpoly}) < \degree_M$, then $\schmrank{\genpolyset} = 0$.
    \begin{proof}
        Assume that $\deg(\vec{\genpoly}) < \degree_M$.
        Therefore, the degree of the linear combination $\vec{\genpoly} = \sum_{i=1}^{c} {\lambda_i \genpoly_i}$ is strictly smaller than the degree of at least one of the polynomials participating in it.
        Denote by $\vec{\lambda}^\star$ the sub-combination of $\vec{\lambda}$ that consists only the polynomials that participated in $\vec{\genpoly}$ that are of degree $= \degree_M$.
        Trivially, $\vec{\lambda}^\star \neq \vec{0}$.
        Additionally, we have $\deg(\sum_{i=1}^{c} {\lambda_i^\star \genpoly_i}) < \degree_M$.
        Now, we use the following observation: the linear combination above, when summing only the homogenous components of each polynomial, equals $0$, i.e. $\sum_{i=1}^{c} {\lambda_i^\star \homopart{\genpoly_i}} \equiv 0$.
        By this, we found a linear combination of $\homopart{\genpolyset}$ that is $\equiv 0$.
        Thus by definition, we have $\schmrank{\genpolyset} = 0$.
    \end{proof}
\end{note*}
\begin{note*}
    This shows that if we compare only the differences in the definition of rank of a factor, i.e. the focus on linear combinations of the largest-degree homogenous components in contrast to the use of the maximal degree $\degree$-rank,
    the two definitions for a rank of a factor are equal up to $\pm 1$
    (in case we use the same definition of rank for a single polynomial).
    To avoid confusion, we omit the exact definitions and respective proof.
\end{note*}

We now present the definition of relative rank as stated in ~\cite[Definition 1.6]{lampert2021relative}:
We remind the reader that this definition is specifically defined to subsets $\variety \subseteq \field$ that can be expressed by $\variety = \zerofunc{\varpolyset}$ for some set of polynomials $\varpolyset$, and not to a general set $\variety \subseteq \field$.
\begin{definition}[Relative Schmidt Rank of a Polynomial]
    The relative schmidt rank of a homogeneous polynomial $\genpoly[1]$ relative to a collection of homogeneous polynomials $\varpolyset=(\series{\varpoly}{1}{\varietypolycount})$ is
    \[
        \relschmrank{\varpolyset}{\genpoly[1]} \definedas
        \min \set{\schmrank{P+\sum_{i=1}^{\varietypolycount}{\remainderpoly_{i}\varpoly_{i}}}
            \suchthat
            \deg(\varpoly_{i})+\deg(\remainderpoly_{i}) \leq \deg(\genpoly[1]), \forall i \in [\varietypolycount]}
    \]
    Note that whenever $\deg{\varpoly_{i}}>\deg{\genpoly}$, this implies $\remainderpoly_{i}=0$.
    \newline
    For general polynomial $\genpoly$ and general collection of polynomials $\varpolyset$, we define
    the schmidt rank of the former in respect to the latter by the relative rank of their largest-degree homogenous component,
    i.e. $\relschmrank{\varpolyset}{\genpoly} \definedas \relschmrank{\homopart{\varpolyset}}{\homopart{\genpoly}}$.
\end{definition}
\begin{remark}[High Relative Rank $\Rightarrow$ High Relative Schmidt Rank]\label{remark-high-relative-rank-implies-high-relative-schmidt-rank}
    Let $\genpoly$ and $\varpolyset = \set{\varpoly_1,...,\varpoly_{\varietypolycount}}$ be polynomials,
    and let $\variety \subseteq \field$ be defined as $\variety = \zerofunc{\varpolyset}$.

    If $\relrank{\variety}{\genpoly} \geq 2 \cdot \rankval + 2$ for some constant $\rankval \in \naturalnumbersset$,
    then $\relschmrank{\varpolyset}{\genpoly} \geq r$.
\end{remark}
\begin{proof}
    Let $\genpoly$ and $\varpoly_1,...,\varpoly_{\varietypolycount}$ be polynomials.
    Assume that $\relschmrank{\varpolyset}{\genpoly} \leq \rankval$.
    Then, there exists $\remainderpoly_1,...,\remainderpoly_\varietypolycount$ with
    $\deg(\varpoly_{i})+\deg(\remainderpoly_{i}) \leq \deg(\genpoly[1])$ for all $i \in \sparens{\varietypolycount}$
    such that:
    \[
        \schmrank{\homopart{\genpoly} + \sum_{i=1}^{\varietypolycount}{\remainderpoly_i \homopart{\varpoly_i}}} \leq \rankval
    \]
    Denote $\relativeremainder{\genpoly_h} \definedas \sum_{i=1}^{\varietypolycount}{\remainderpoly_i \homopart{\varpoly_i}}$.
    As we have shown earlier, a rank of a polynomial is smaller than its schmidt rank up to a constant factor (See Remark~\ref{high-rank-implies-high-schmidt-rank}).
    Thus:
    \[
        \rank{\homopart{\genpoly} + \relativeremainder{\genpoly_h}} \leq
        2 \cdot \schmrank{\homopart{\genpoly} + \relativeremainder{\genpoly_h}} + 1 \leq
        2 \cdot \relschmrank{\variety}{\genpoly} + 1 =
        2 \rankval + 1
    \]
    Next, we denote the respective remainder polynomial for the non-homogenous analogue, i.e $\relativeremainder{\genpoly} \definedas \sum_{i=1}^{\varietypolycount}{\remainderpoly_i \varpoly_i}$.
    By observing the highest degree homogenous component of each summand, one can see that $\homopart{\genpoly + \relativeremainder{\genpoly}} = \homopart{\homopart{\genpoly} + \relativeremainder{\genpoly_h}}$.
    Therefore, by adding to the decomposition the non higest-degree-homogenous-component, one can see that:
    \[
        \rank{\genpoly + \relativeremainder{\genpoly}} \leq
        \rank{\homopart{\genpoly} + \relativeremainder{\genpoly_h}} + 1 \leq
        2 \rankval + 2
    \]
    This completes the proof as $\relrank{\variety}{\genpoly} \leq \rank{\genpoly + \relativeremainder{\genpoly}} \leq 2 \rankval + 2$.
\end{proof}

\begin{remark}[Relative Schmidt Rank over Varieties of High Degree]\label{relative-schimdt-rank-equals-schmidt-rank-if-the-variety-is-of-high-degree}
If the polynomials defining the variety $\varpolyset = (\varpoly_1,...,\varpoly_{\varietypolycount})$ are of degree $> \deg(\genpoly)$,
then, $\relschmrank{\varpolyset}{\genpoly} = \schmrank{\genpoly}$.
This is true because in this case, in the calculation of the minimum in the definition of relative schmidt rank, we must have $\remainderpoly_i = 1$ for all $i \in [\varietypolycount]$ and therefore the minimum above is simply $\rank{\genpoly}$.
\newline
Note that a similar statement holds for factors aswell.
If $\genpolyset = (\genpoly_1,...,\genpoly_c)$ is a factor of degree $\degree$, then if all the polynomials in $\varpolyset$ have degree $> \degree$, then the statement above is also true i.e $\relschmrank{\varpolyset}{\genpolyset} = \schmrank{\genpolyset}$.
This is true because for every linear combination of $\genpolyset$ has degree $\leq \degree$ and therefore its relative schimdt rank equals its rank.
\end{remark}


Finally, we present the extension of the definition of relative rank for polynomials factors:
\begin{definition}[Relative Schmidt Rank of a Factor]
    The relative rank of a set of homogenous polynomials $\genpolyset = \set{\genpoly_1,...,\genpoly_c}$
    relative to another collection of polynomials $\varpolyset = \set{\varpoly_1,...,\varpoly_\varietypolycount}$ is defined as:
    \[
        \relschmrank{\varpolyset}{\genpolyset} \definedas
        \min \set{\relschmrank{\varpolyset}{\sum_{i=1}^c{\lambda_i \genpoly_i}} \suchthat \vec{0} \neq (\lambda_1,...,\lambda_c) \in \basefield^c}
    \]
    If $\genpolyset$ is a general collection of polynomials, then $\relschmrank{\varpolyset}{\genpolyset} \definedas \relschmrank{\varpolyset}{\homopart{\genpolyset}}$.
    \newline
    For a factor $\factor$ generated by a set of polynomials $\genpolyset$, we define its schmidt rank relative to $\variety = \zerofunc{\varpolyset}$
    to be $\relschmrank{\variety}{\factor} \definedas \relschmrank{\varpolyset}{\genpolyset}$.
\end{definition}

\begin{remark}
    Let $\genpolyset = \set{\genpoly_1,...,\genpoly_c}$ and $\varpolyset = \set{\varpoly_1,...,\varpoly_{\varietypolycount}}$ be sets of polynomials,
    and let $\variety \subseteq \field$ be defined as $\variety = \zerofunc{\varpolyset}$.
    Additionally, let $\rankval \in \naturalnumbersset$ such that $\rankval > 0$.
    If $\relrank{\variety}{\genpolyset} \geq 2 \cdot \rankval + 2$ for some constant $\rankval \in \naturalnumbersset$, then $\relschmrank{\varpolyset}{\genpoly} \geq r$.
\end{remark}
\begin{proof}
    Assume that $\relschmrank{\varpolyset}{\genpolyset} \leq \rankval$.
    We will show that $\relrank{\variety}{\genpolyset} \leq 2\rankval + 2$.
    Let $\vec{0} \neq \vec{\lambda} \in \basefield^c$ be some vector of coefficients.
    Let $\vec{\genpoly} \definedas \sum_{i=1}^{c} {\lambda_i \genpoly_i}$ and $\vec{\genpoly_h} \definedas \sum_{i=1}^{c} {\lambda_i \homopart{\genpoly_i}}$
    be the linear combinations of polynomials in $\genpolyset$ and $\homopart{\genpolyset}$ with coefficients $\vec{\lambda}$ respectively,
    and let $\degree_M \definedas \max_{i \in \sparens{c}}{\deg(\lambda_i \genpoly_i)}$.
    Additionally, denote $\hat{\rankval} \definedas \relschmrank{\varpolyset}{\vec{\genpoly_h}} \leq \rankval$.
    It is enough to show that $\drelrank{\degree_M}{\variety}{\vec{\genpoly}} \leq 2 \hat{\rankval} + 2$,
    If $\deg(\vec{\genpoly}) < \degree_M$, then $\drelrank{\degree_M}{\variety}{\vec{\genpoly}} = 1 \leq 2 \rankval + 2$.
    Otherwise, if $\deg(\vec{\genpoly}) = \degree_M$, then the remark follows from Remark~\ref{remark-high-relative-rank-implies-high-relative-schmidt-rank}
    as:
    \[
        \drelrank{\degree_M}{\variety}{\vec{\genpoly}} =
        \relrank{\variety}{\vec{\genpoly}} \leq
        2 \cdot \relschmrank{\varpolyset}{\vec{\genpoly}} + 2
    \]
    Where:
    \[
        \relschmrank{\varpolyset}{\vec{\genpoly}}
        \relschmrank{\varpolyset}{\homopart{\vec{\genpoly}}} =
        \relschmrank{\varpolyset}{\homopart{\vec{\genpoly_h}}} =
        \relschmrank{\varpolyset}{\vec{\genpoly_h}} =
        \hat{\rankval}
    \]
\end{proof}